\documentclass{article}
\usepackage{custom_tex}
\graphicspath{{./fig/}}

\begin{document}

%\title{Conditional Predictive Inference under Dataset Shift}
\title{Conditional Predictive Inference for Missing Outcomes}
%\title{Simultaneous Conditional Predictive Inference in Dataset Shift}
%\title{Simultaneous Conditional Predictive Inference on Missing Outcomes}
%\title{Conditional predictive inference on multiple missing outcomes}

\date{\today}

\author{Yonghoon Lee}
\author{Edgar Dobriban}
\author{Eric Tchetgen Tchetgen\footnote{E-mail addresses: 
\texttt{yhoony31@wharton.upenn.edu},
\texttt{dobriban@wharton.upenn.edu},
\texttt{ett@wharton.upenn.edu}
}}

\affil{Department of Statistics and Data Science, The Wharton School, University of Pennsylvania}

\maketitle

\begin{abstract}
We study the problem of conditional predictive inference on multiple outcomes missing at random (MAR)---or equivalently, under covariate shift. While the weighted conformal prediction~\citep{tibshirani2019conformal} offers a tool for inference under covariate shift with a marginal coverage guarantee, procedures with conditional coverage guarantees are often desired in many applications to ensure reliable inference for a specific group of individuals. A standard approach to overcoming the fundamental limitation of distribution-free conditional predictive inference is to relax the target and instead aim to control coverage conditional on a local area, subset, or bin in the feature space. However, when the missingness pattern depends on the features, this relaxation remains challenging due to the violation of the MAR assumption with respect to the bins. To address this issue, we propose a \textit{propensity score $\ep$-discretization}, a carefully designed binning strategy based on the propensity score, which enables valid conditional inference. Based on this strategy, we develop a procedure---termed \emph{pro-CP}---that enables simultaneous conditional predictive inference for multiple missing outcomes. We show that \emph{pro-CP} controls the bin-conditional coverage rate in a distribution-free manner when the propensity score is either known exactly or estimated with sufficient accuracy. Furthermore, we provide a theoretical bound on the coverage rate when the propensity score is unknown and must be estimated. Notably, the error bound remains constant and depends only on the estimation quality, not on the sample size or the number of outcomes under consideration. In extensive empirical experiments on simulated data and on a job search intervention dataset, we illustrate that our procedures provide informative prediction sets with valid conditional coverage.
\end{abstract}

\tableofcontents
\medskip

\sloppy

\section{Introduction}
Consider a standard predictive inference problem, where labeled calibration data $(X_i, Y_i)_{i=1}^n$ are used to perform inference on the unknown test outcomes $(Y_{n+j})_{i=1}^m$, given an unlabeled dataset $(X_{n+j})_{i=1}^m$. Conformal prediction~\citep{saunders1999transduction,vovk1999machine,vovk2005algorithmic, papadopoulos2002inductive} provides a methodology for constructing prediction sets with a distribution-free coverage guarantee,
but the applicability of the standard methodology is limited in the following sense:
\begin{enumerate}
    \item The calibration and the test data must be exchangeable.
    \item The coverage guarantee is marginal over all the randomness in the calibration and test data; hence, it may not be suitable in settings where the goal is to obtain personalized/tailored/informative inference conditional on specific individuals or groups.
\end{enumerate}

The violation of exchangeability can 
arise in many settings involving dataset or distribution shift, see e.g., \cite{shimodaira2000improving,quinonero2009dataset,Sugiyama2012,ben2006analysis,lipton2018detecting}, 
or in scenarios including randomized controlled trials \citep{kalton2020introduction,hariton2018randomised},
etc. 
For example, consider a randomized controlled trial where we observe data of the form $(X_i, A_i, (1 - A_i)Y_i(0) + A_i Y_i(1))$, where $A_i \in \{0,1\}$ denotes the treatment assignment and $Y_i(0)$ and $Y_i(1)$ represent the potential outcomes without and with treatment, respectively, for the $i$-th individual. 
For inference on the counterfactual variables $\{Y_i(0) : A_i = 1\}$, the calibration data used are $\{(X_i, Y_i(0)) : A_i = 0\}$, which are sampled from $P_{X \mid A=0} \times P_{Y(0) \mid X}$, whereas the unlabeled test features $\{X_i : A_i = 1\}$ are drawn from $P_{X \mid A=1}$. 
For such covariate shift settings, \citet{tibshirani2019conformal} propose weighted conformal prediction, which provides valid (marginal) coverage when the \textit{propensity score} $p_{A \mid X} : x \mapsto \PPst{A=1}{X=x}$ is known;
followed by several extensions and other developments  \citep[e.g.,][etc]{podkopaev2021distribution,gibbs2021adaptive,qiu2022prediction}. 

Achieving a conditional inference within the distribution-free framework---e.g., constructing a prediction set $\ch(X_{n+1})$ with a provable control of conditional miscoverage rate $\PPst{Y_{n+1} \in \ch(X_{n+1})}{X_{n+1}}$---has been of significant interest recently, but many works have shown that it is generally impossible to attain meaningful conditional inference without distributional assumptions.
For instance, \citet{vovk2012conditional} shows that any distribution-free method with strict conditional coverage control must output a prediction set with infinite measure—meaning that it is uninformative—and \citet{foygel2021limits} show that a similar impossibility result holds even for a weaker target. 
Consequently, different forms of relaxation of the inferential target are often considered. 
For example, \citet{gupta2020distribution} explores inference conditional on a bin instead of the full feature vector, and \citet{jung2023batch,gibbs2025conformal} introduce methods that controls $\PPst{Y_{n+1} \in \ch(X_{n+1})}{X_{n+1} \in G}$ for all $G$ in a collection of subsets.\footnote{As a remark, the approach in~\citet{gibbs2025conformal} aims to control conditional coverage in an i.i.d.~setting by addressing multiple pre-specified sets of covariate shifts, but it is not directly related to the problem of inference under covariate shift. 
For example, even when the propensity score $p_{A \mid X}$ is known, this information cannot be directly incorporated into their method without compromising theoretical guarantees, nor is it their intended focus.}

%\ed{consider case of known propensity score. then our procedure works under MAR and is richer; \cite{gibbs2025conformal}  requires that $Y$ given something in the basis is independent of $A$, which means that the propensity score needs to be in the basis. analogy in causal: reg vs prop score matching. }

We explore the setting in which \emph{both} issues arise---namely, there is covariate shift, breaking exchangeability, and conditional inference is desired at the same time. 
Alternatively, as we will show below (\Cref{ps}),
in a missing data scenario,
one can view this problem as inference on outcomes missing at random, where prediction sets are constructed based on data points with observed outcomes.

Moreover, we are interested in drawing simultaneous  inferences on \emph{multiple outcomes} 
$(Y_{n+j})_{i=1}^m$, given their features $(X_{n+j})_{i=1}^m$.
While the problem of conditional inference under covariate shift for a single outcome is already challenging and remains unaddressed, 
we consider multiple outcomes
due to several reasons.
First, considering the setting of missing data mentioned above, this setting allows us to perform simultaneous inference on multiple missing outcomes; which can be of practical interest.
Second, simultaneous inference on multiple outcomes allows a tighter control of error rates than inference on one outcome at a time: for instance, we can be sure that, say, 95\% of the test outcomes of the given features $(X_{n+j})_{i=1}^m$ are covered, with 99\% probability.

More formally, given a calibration sample $(X_i,Y_i)_{i=1}^n$ from $P_{X \mid A=1} \times P_{Y \mid X}$ and a test sample $(X_{n+j},Y_{n+j})_{j=1}^m$ from $P_{X \mid A=0} \times P_{Y \mid X}$---with $(Y_{n+j})_{j=1}^m$ unobserved---we aim to control the following test-conditional coverage rate:
\begin{equation}\label{eqn:test_set_conditional}
    \EEst{\frac{1}{m}\sum_{j=1}^m \One{Y_{n+j} \in \ch(X_{n+j})}}{(X_{n+j})_{j=1}^m}.
\end{equation}
Thus, we want to cover most of the outcomes $Y_{n+j}$, with high probability.
This metric generalizes the conditional coverage rate $\mathbb{P}\{Y_{n+1} \in \ch(X_{n+1}) \mid X_{n+1}\}$, which corresponds
to the case $m=1$. Controlling conditional coverage is often desired in applications where high-quality inference for specific individuals is important---for example, doctors often require reliable diagnoses for an individual patient, and recruiters are interested in accurate evaluations of the of the particular applicant under review. In such cases, a marginal coverage guarantee may not be the appropriate target. The conditional coverage in the sense of~\eqref{eqn:test_set_conditional} extends the notion of conditional coverage to the finite population setting. It captures the quality of a procedure conditional on a specific test set or finite population of interest---e.g., a group of patients under treatment, students in a new teaching program, etc, where conditional inference on a specific group, rather than marginal inference with respect to a hypothetical infinite population, is more desirable.

%\ed{Make better case that we need our method: conditional would be desired goal, as it would give the highest level of precision and personalization. However, it is not possible, even in the case of no covariate}
%Observe that the term inside the expectation represents the proportion of test outcomes that are covered by their corresponding prediction sets. 
Controlling this error rate under covariate shift appears to be nontrivial, and 
cannot be directly achieved by combining approaches developed separately for covariate shift and conditional coverage.
%, especially when the goal is inference on multiple outcomes.
In particular, since the test sample size $m$ can be small (we even consider $m=1)$, using concentration inequalities over the $m$ summands leads to loose and conservative results.

To better understand the challenge, 
consider for instance \emph{binning}, one of the key strategies used in prior work to relax conditional inference to a feasible constraint \citep{gupta2020distribution}.
Suppose we bin each feature $X_i$, mapping it to $B_i = b(X_i)$, with some map $b$.
In the standard i.i.d. setting, such binning/discretization enables some level of conditional inference, since it can lead to multiple outcome datapoints with the same (discretized) feature observation---enabling learning about $P_{Y \mid b(X)}$. 
However, this approach fails for an arbitrary binning strategy under covariate shift/missing at random data, since
\emph{arbitrary binning does not preserve covariate shift}.
To wit, the distribution of $Y \mid b(X)$ is a mixture of distributions $Y|X'$ for $X'$ such that $b(X')=b(X)$, and thus can depend on the distribution of $X$. 
Therefore, even though $Y|X$ initially  has the same distribution under $A=0$ and $A=1$, this does not necessarily hold after binning.\footnote{From the perspective of missing data, this means that as once the features are discretized, $Y$ is not necessarily missing at random but rather missing not at random, which makes the problem more challenging.}

To overcome this challenge, we carefully examine the source of the violation of covariate shift after binning. 
Since the distribution after binning is a mixture of different conditional distributions, this raises the possibility of constructing a binning scheme that mixes together only similar conditional distributions. We take up this approach, and show that 
such a binning strategy can indeed be developed, 
by leveraging the odds of the propensity score.
Then, we construct a simultaneous inference procedure for multiple missing outcomes that controls the bin-conditional coverage rate---a surrogate for the feature-conditional coverage rate, which is unattainable with nontrivial prediction sets. Our contributions are summarized below.

\subsection{Main contributions}
We develop methods for predictive inference of multiple outcomes under covariate shift (or equivalently, missing at random), with conditional guarantees.

\begin{enumerate}
    \item {\bf Inferential target: feature-conditional coverage.}
    We discuss which conditional inferential goals are appropriate, depending on the discreteness or continuity of of the data distribution (\Cref{sec:in_exp}).
    For discrete-valued features (or, more generally, features whose distribution has point masses)
    we provide a method that satisfies
feature-conditional coverage guarantee (\Cref
{thm:MAR_in_exp}).
We show that this method is valid as long as the per-feature observations are exchangeable, which enables  
using it to construct 
narrower prediction sets via partitioning the observations (Corollary \ref{cor:in_exp_discrete_partition}).

\item {\bf Propensity score discretization-based conformal prediction (pro-CP).}
To handle general continuous feature distributions, 
we face the challenge that, since we do not make any assumptions  (such as smoothness) on the distribution  of the outcome given the features, 
we cannot borrow information across feature values. 
Thus, we introduce methods based on discretized feature values.
Since feature-discretized data generally do not remain missing at random, we propose a carefully crafted binning/discretization strategy based on the propensity score $x \mapsto p_{A \mid X}(x) = \PPst{A = 1}{X = x}$.
We show that \emph{approximate within-bin exchangeability} is retained.\footnote{Stratifying the propensity score (e.g., into several quantiles) has previously appeared in the causal inference literature—see, for example, \citet{rosenbaum1983central, rosenbaum1984reducing}—in the context of approximate balancing. In this work, our innovations are to (1) propose a specific discretization strategy which is linear in the log-odds of the propensity score and (2) theoretically prove that this ensures ``approximate within-bin exchangeability".} 
We refer to this this discretization scheme as \textit{propensity score $\ep$-discretization}.
We then introduce \emph{propensity score discretization-based conformal prediction} (pro-CP), and show that it achieves propensity score-discretized feature-conditional coverage when the propensity score is known exactly (\Cref{thm:MAR_known_p_A_X}),
We also discuss a use case of the procedures we introduced, to
obtain inference for individual treatment effects (\Cref{ite}). 

\item {\bf Propensity score discretization leads to approximate balancing.}
Our analysis crucially relies on a new theoretical result (\Cref{lem:dtv_0_1}) showing that the proposed propensity score $\eps$-discretization leads to ``approximate independence" between the distributions of $Y$ and $A$.
While classical results show that 
the propensity score has the balancing property,
so that $Y$ is independent of $A$ given $ p_{A \mid X}(X)$~\citep{rosenbaum1983central}, our results show that
\emph{propensity score $\eps$-discretization} leads to ``$\eps$-approximate balancing", which we detail later.

\begin{figure}[htbp]
  \centering
  \includegraphics[width=0.8\textwidth]{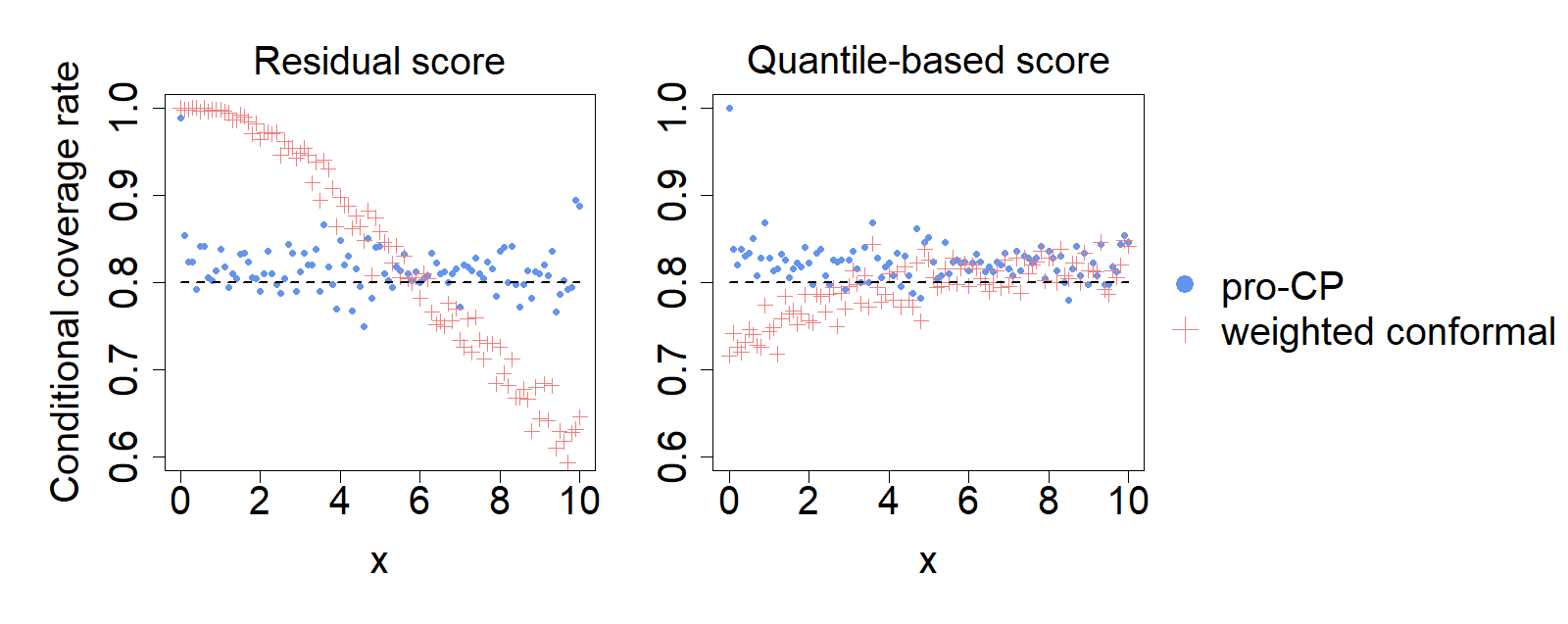}
  \caption{Conditional coverage rates of pro-CP and the method based on weighted conformal prediction~\citep{tibshirani2019conformal}---which targets a marginal coverage guarantee---under two choices of the nonconformity score. Note that the quantile-based score combined with weighted conformal prediction corresponds to the method of~\citet{lei2021conformal}. Our method shows approximate feature-conditional coverage, whereas the weighted conformal prediction does not. See \Cref{sec:sim} for details.
  }
  \label{fig:single_conditinal_coverage_0}
\end{figure}

\item {\bf Empirical evaluation.}
We evaluate our methods empirically (\Cref{exp}), 
both in numerical simulations (\Cref{sim} and \Cref{sec:sim_mult}),
and using an empirical data set (\Cref{dat2}) \footnote{Code to reproduce the experiments is available at \url{https://github.com/yhoon31/pro-CP}.} .
The empirical results support that our method satisfies the desired theoretical guarantees---namely, discretized feature-conditional coverage control; 
while also providing good control of the feature-conditional coverage. Additionally, they show that the pro-CP procedure does not generate overly conservative prediction sets. 
For inference on counterfactuals and individual treatment effects, our method has an advantage over the popular method of~\cite{tibshirani2019conformal} or its variant in~\cite{lei2021conformal}, in that it provides approximate feature-conditional coverage (See Figure \ref{fig:single_conditinal_coverage_0})---with provable theoretical conditional guarantees, rather than relying on the (unknown) conditional quality of estimates.

\end{enumerate}

\subsection{Problem setting}
\label{ps}
We consider the \textit{missing at random} (MAR)---or equivalently, covariate shift---setting, where the distribution of the outcome depends on the feature but not on whether it is observed. That is, denoting the feature, outcome, and observation indicator by $X$, $Y$, and $A$, respectively, we assume $P_{Y \mid X, A} = P_{Y \mid X}$, or equivalently, $Y \indep A \mid X$. 
We consider calibration and the test data $(X_i,Y_i)_{i=1}^{n+m} \subset \X \times \Y$, $m\ge 1$, drawn as follows:
\begin{equation}\label{eqn:model}
    \begin{split}
    &(X_1,Y_1),(X_2,Y_2), \ldots, (X_n,Y_n) \iidsim P_{X \mid A=1} \times P_{Y\mid X} \\
    &(X_{n+1},Y_{n+1}), (X_{n+2},Y_{n+2}) \ldots, (X_{n+m},Y_{n+m}) \iidsim P_{X \mid A=0} \times P_{Y \mid X},
    \end{split}
\end{equation}
where we only observe $(X_i,Y_i)_{i=1}^n$ and $(X_{n+j})_{j=1}^m$. 
The task is to construct prediction sets for the unobserved missing outcomes $Y_{n+1}, \cdots, Y_{n+m}$. 
Specifically, we aim to construct an algorithm $\widehat{C} : \X \rightarrow \mathcal{P}(\Y)$ such that most of the sets $\widehat{C}(X_{n+j})$ for $j=1,2,\cdots,m$ cover their corresponding missing outcome $Y_{n+j}$.

This formulation is equivalent to being given data with missing outcomes $(X_i, A_i, Y_i A_i)_{i \in [n]}$, where we use $(X_i, Y_i)_{i : A_i = 1}$ as calibration data to perform inference on the unobserved outcomes $(Y_i)_{i : A_i = 0}$---where all theoretical guarantees are stated conditional on the missingness indicators $(A_i)_{i \in [n]}$.

\subsection{Related work}

Prediction sets go back to \protect\cite{Wilks1941}, \protect\cite{Wald1943}, \protect\cite{scheffe1945non}, and \protect\cite{tukey1947non,tukey1948nonparametric}.
Distribution-free inference and the conformal prediction framework have been extensively studied in recent works \protect\citep[see, e.g.,][]{saunders1999transduction,vovk1999machine,papadopoulos2002inductive,vovk2005algorithmic,vovk2012conditional, Chernozhukov2018,dunn2022distribution,lei2013distribution,lei2014distribution,lei2015conformal,lei2018distribution,angelopoulos2021gentle,guan2022prediction,guan2023localized, guan2023conformal,romano2020classification,bates2023testing,einbinder2022training,liang2022integrative,liang2023conformal}. 
Conformal prediction provides a general framework for constructing prediction intervals with a marginal coverage guarantee under the exchangeability of datapoints.
Predictive inference methods
\citep[e.g.,][etc]{geisser2017predictive}  
have been developed under various assumptions
\protect\citep[see, e.g.,][]{Park2020PAC,park2021pac,park2022pac,sesia2022conformal,qiu2022prediction,li2022pac,kaur2022idecode,si2023pac}. 
% Split conformal prediction~\citep{ papadopoulos2002inductive,vovk2005algorithmic}
% achieves the same goal with
% a lower computational cost through data splitting. 
 
Several works have explored the possibility of attaining stronger guarantees. 
\citet{vovk2012conditional} shows that split conformal prediction provides good control of conditional coverage if the sample size is sufficiently large; equivalently to the coverage properties of tolerance regions \citep{Wilks1941}. 
\cite{lei2014distribution} shows
that finite sample validity
conditional on all feature values implies infinite-length prediction sets almost surely.
For test-conditional predictive inference, \citet{foygel2021limits} proves an impossibility result for the goal of bounding the conditional coverage rate when the feature distribution is nonatomic (i.e., has no point mass). \citet{barber2020distribution} and \citet{lee2021distribution} discuss a distribution-free regression problem where the goal is to cover the conditional mean $\EEst{Y_{n+1}}{X_{n+1}}$ and provide impossibility results for nonatomic features. Consequently, several works have explored relaxed targets for conditional predictive inference. For example, \citet{hore2023conformal} propose a method that approximately controls the coverage rate conditional on a neighborhood of the test input, while \citet{gibbs2025conformal} discuss a relaxation of the multi-accuracy condition, which, in special cases, leads to bin-conditional coverage control.

Inference on the missing outcomes is closely related to the problem of prediction under covariate shift \citep[see e.g.,][]{shimodaira2000improving, quinonero2009dataset, Sugiyama2012}. 
We further discuss the relation between the missing data problem and the covariate shift problem in Section~\ref{sec:in_exp}, but to briefly outline the rationale, suppose we are interested in the prediction of only one missing outcome. Considering data points with outcomes as training data and the target outcome as the test outcome reduces the problem to prediction under covariate shift. On the other hand, in the missing data problem, our focus is on simultaneously inferring multiple missing outcomes, which is different; see Section~\ref{sec:in_exp}.

For the related problem of prediction under covariate shift, %recent works establish impossibility results for an exact finite sample distribution-free guarantee and provide methods with additional assumptions or with asymptotic validity. 
\citet{tibshirani2019conformal} introduces weighted conformal prediction, which handles the effect of covariate shift by putting weights on the scores based on the likelihood ratio of the feature distributions. 
Their method provides a finite sample guarantee 
when the likelihood ratio is known; see also \cite{park2021pac} for the case of a PAC guarantee for weights that are known up to being in a hyper-rectangle. 
This methodology is further explored by \citet{lei2021conformal}, where the authors prove the asymptotic validity of weighted conformal prediction with quantile estimate-based scores. 
When the likelihood ratio and outcome model are estimated, methods with doubly robust asymptotic coverage under distribution shift are introduced in 
\citet{yang2022doubly} with a marginal guarantee and in
\citet{qiu2022prediction} with a PAC guarantee.

Inference on missing outcomes has been extensively studied in the context of multiple imputation. Comprehensive overviews of multiple imputation methods are provided by \citet{harel2007multiple} and \citet{rubin1996multiple}. 
Multiple imputation was introduced by \citet{rubin1978multiple}, 
who proposes a Bayesian approach to achieve a distribution-like imputation rather than a single imputation. \citet{reilly1993data} studies nonparametric approaches, and asymptotic properties of parametric imputation models are examined by \citet{wang1998large} and \citet{robins2000inference}. \citet{rubin1986multiple} proposes a method for interval estimation, constructing an interval that covers the mean of the missing outcome.

\subsection{Notations}
We write $\R$ to denote the set of real 
numbers 
and $\R_{\ge 0}$ to denote the set of nonnegative reals.
For a positive integer $n$, we write $[n]$ to denote the set $\{1,2,\ldots,n\}$ and write $X_{1:n}$ to denote the vector $(X_1,X_2,\ldots,X_n)^\top$. 
For a distribution $P$ on $\R$ and a constant $\alpha \in (0,1)$, we write $Q_{1-\alpha}(P)$ to denote the $(1-\alpha)$-quantile of $P$, i.e., 
\[Q_{1-\alpha}(P) = \inf\left\{t \in \R : \Pp{T \sim P}{T \le t} \ge 1-\alpha \right\}.\]
For numbers $v_1,v_2,\ldots,v_m \in \R \cup \{\infty\}$ and $p_1,p_2,\ldots,p_m \in [0,1]$ satisfying $p_1+p_2+\ldots+p_m = 1$, 
we write $\sum_{i=1}^m p_i \delta_{v_i}$ to denote the discrete distribution that has probability mass function $p : \R \cup \{\infty\} \rightarrow [0,1]$ with
$p(x) = p_i$ 
 if 
 $x = v_i$, for $i\in[m]$, and $p(x) = 0$ otherwise.
For non-negative integers $a\le b$, we denote the binomial coefficient by $\binom{b}{a}$; the same expression is interpreted as zero for other values $a,b$. For a vector $(a_1,a_2,\ldots,a_m)^\top$ and a set of indices $I=\{i_1,i_2,\ldots,i_k\} \subset [m]$ with $i_1 < i_2 < \ldots < i_k$, we write $(a_i)_{i \in I}$ to denote the sub-vector $(a_{i_1},a_{i_2}, \ldots, a_{i_k})^\top$, also writing $a_{u:v}:=a_{u, u+1, \ldots, v}$ for positive integers $u<v$, $u,v\in [m]$. For an event $E$, we write $\One{E}$ to denote its corresponding indicator variable. For a set $D$, $\mathcal{P}(D)$ denotes its power set.

\section{Main results}\label{sec:in_exp}

\subsection{Inferential targets}
Given data $\D = (X_i,Y_i)_{i \in [n]}$, we aim to construct a procedure $\C$ such that $\ch = \C(\D)$ provides prediction sets $\{\ch(X_{n+j}) : j\in [m]\}$ for the missing outcomes $\{Y_{n+j} : j\in [m]\}$. 
The \emph{realized coverage}
for the
missing outcomes is 
the fraction
$\frac{1}{m}\sum_{j = 1}^m \One{Y_{n+j} \in \widehat{C}(X_{n+j})}$ of outcomes $Y_{n+j}$ covered by the prediction sets.

Our goal is to control the expected coverage, possibly conditional on certain functions of the data. The simplest property one might consider is unconditional coverage, namely
\begin{equation}\label{eqn:in_exp_guarantee}
    \EE{\frac{1}{m}\sum_{j=1}^m \One{Y_{n+j} \in \widehat{C}(X_{n+j})}}
    = \PP{Y_{n+1} \in \widehat{C}(X_{n+1})}
    \ge 1-\alpha.
\end{equation}
The expectation is taken with respect to the distribution~\eqref{eqn:model}; and the simplification happens due to the i.i.d.~sampling of $(X_{n+j},Y_{n+j})_{j=1}^m$.  
Applying methods such as weighted conformal prediction~\citep{tibshirani2019conformal}---when $P_{A \mid X}$ is known---separately to individual test points can guarantee coverage.
% \begin{equation*}
% \mathbb{E}\bigg[\frac{1}{m}\sum_{j=1}^m \One{Y_{n+j} \in \widehat{C}(X_{n+j})}\bigg]=\frac{1}{m}\sum_{i=1}^m \PP{Y_{n+j} \in \widehat{C}(X_{n+j})} \ge 1-\alpha.
% \end{equation*}
In this work, we aim to achieve stronger guarantees for conditional inference. An ideal guarantee would be the following test-input-conditional coverage property:

\begin{definition}[Feature-conditional coverage guarantee]\label{fmc}
A procedure $\ch(\cdot) = \ch(\cdot ; (X_i,Y_i)_{i \in [n]}, (X_{n+j})_{j \in [m]})$ satisfies 
feature-conditional 
coverage guarantee at level $1-\alpha$ if
\begin{equation}\label{eqn:in_exp_guarantee_X_conditional}
    \EEst{\frac{1}{m}\sum_{j=1}^m \One{Y_{n+j} \in \widehat{C}(X_{n+j})}}{(X_{n+j})_{j \in [m]}} \ge 1-\alpha,\qquad \textnormal{ almost surely.}
\end{equation}
\end{definition}
Intuitively, the above condition guarantees that we obtain reliable inference for any set of realized test inputs $X_{n+1}, \cdots, X_{n+m}$.
Clearly, it implies\footnote{As for the marginal guarantee, the feature-conditional can be expressed alternatively as 
$\frac{1}{m}\sum_{j=1}^m\PPst{Y_{n+j} \in \widehat{C}(X_{n+j})}{(X_{n+j})_{j \in [m]}} \ge 1-\alpha$. Since $\widehat{C}$ will in general depend on all of $(X_{n+j})_{j \in [m]}$, the probabilities do not in general reduce to $\PPst{Y_{n+j} \in \widehat{C}(X_{n+j})}{X_{n+j}}$.}
the marginal guarantee \eqref{eqn:in_exp_guarantee}.
When $m = 1$, the above guarantee reduces to the standard target of conditional predictive inference for a single test point: $\PPst{Y_{n+1} \in \ch(X_{n+1})}{X_{n+1}} \ge 1 - \alpha$ \citep{vovk2012conditional}. 
This is a very strong requirement, 
and it is not attainable by any distribution-free procedure with bounded average prediction set width, if the feature distribution is nonatomic---even in the simplest setting of $m = 1$ and no covariate shift~\citep{vovk2012conditional}.

However, several prior works have shown that approximate forms of feature-conditional coverage become possible if the feature space is approximately discretized or partitioned \citep[see e.g.,][etc]{gupta2020distribution,jung2023batch,gibbs2025conformal}.
In particular, discretization induces an atomic feature distribution, which avoids the above impossibility results.
Inspired by these works, our first method also concerns discrete-valued features, or more generally features whose distribution has point masses. 
There, we develop a new method that achieves the above guarantee.

For more general feature distributions, it is also reasonable to consider discretizing the feature space. 
However, this approach runs into a roadblock in the setting of covariate shift/data missing at random.
The problem is that 
$Y$ and $A$ \emph{may not be retain
independence}
conditional on the discretized $X$, i.e., the missing at random assumption may not be preserved for the discretized features.
This leads to a setting of arbitrary distribution shift after discretization, for which only weaker guarantees are known to be possible to achieve \citep[see e.g.,][etc]{bastani2022practical}.
One might consider---as detailed in Section~\ref{bin-wCP}---a straightforward approach of applying weighted conformal prediction within each bin, separately for each test point. However, this significantly reduces the sample size available for inference on each test point and may be impractical in many scenarios of interest (e.g., by producing trivial prediction sets of infinite width for many test points).

To overcome this challenge, 
when we know the propensity score exactly, such as in randomized trials or two-phase sampling studies,
we propose a bespoke binning scheme based on the 
 propensity score $x\mapsto p_{A \mid X}(x) = \PPst{A = 1}{X=x}$.
Our approach is inspired by the balancing property of the propensity score, which ensures that conditioning on its precise value, $Y$ and $A$ remain independent \citep{rosenbaum1983central}.
However, for continuous-valued features, we need to discretize the propensity score to ensure there are multiple datapoints in each bin after discretization.
Therefore, going beyond the known exact balancing property of the propensity score, we show that by discretizing it appropriately,
we retain \emph{approximate independence}. When the propensity score needs to be estimated, in general we  incur an additional slack in our coverage guarantee, which we characterize precisely. 

Specifically, we will consider the following guarantee:

\begin{definition}[Propensity score discretized feature-conditional coverage]\label{proc}
  Suppose 
 the propensity score function satisfies
  $0 < p_{A \mid X}(x) < 1$ for all $x \in \X$. 
Consider
 a strictly increasing sequence 
 $(z_k)_{k \in \mathbb{Z}}$
in $[0,1]$ with $\lim_{k\to-\infty} z_k = 0$ and $\lim_{k\to\infty} z_k = 1$, and
the partition $\mathcal{B}$ of the feature space $\X$ given by
\begin{equation}\label{eqn:eps_partition}
    \mathcal{B} = \{D_k : k \in \mathbb{Z}\},\,\,  D_k = \left\{x : p_{A \mid X}(x) \in [z_k,z_{k+1})\right\}.
\end{equation}
For $i \in [n]$, let 
$B_i$ be the unique index
$k \in \mathbb{Z}$ such that $D_k$ contains $X_i$.
A procedure $\ch$ satisfies 
\emph{propensity score discretized feature-conditional coverage guarantee}\footnote{As above, this reduces to 
$\frac{1}{m}\sum_{j=1}^m\PPst{Y_{n+j} \in \widehat{C}(X_{n+j})}{(B_{n+j})_{j \in [m]}} \ge 1-\alpha$, but the probabilities do not in general simplify to $\PPst{Y_{n+j} \in \widehat{C}(X_{n+j})}{B_{n+j}}$ since $\hat C$ can depend on all of $(X_{n+j})_{j \in [m]}$ and thus $(B_{n+j})_{j \in [m]}$.} 
at level $1-\alpha$ if
\begin{equation}\label{psmcc}
\EEst{\frac{1}{m}\sum_{j=1}^m \One{Y_{n+j} \in \widehat{C}(X_{n+j})}}{(B_{n+j})_{j \in [m]}} \ge 1-\alpha, \qquad \text{ almost surely}.
\end{equation}
\end{definition}

This guarantee depends on the partition $\mathcal{B}$, and we will  later specify the form we use in our results. The discretized feature-conditional guarantee~\eqref{psmcc} can be considered a surrogate for the original feature-conditional guarantee~\eqref{eqn:in_exp_guarantee_X_conditional}, in the sense that the discretized features $(B_{n+j})_{j \in [m]}$ contain partial information about the true features $(X_{n+j})_{j \in [m]}$. In the experiments, we will show that controlling the bin-conditional coverage indeed leads to control of the feature-conditional coverage in most cases.

\subsubsection{Overview}
Here, we briefly outline the organization of the remainder of this section.

\begin{enumerate}
\item We first discuss the setting where $X$ is discrete, and provide a procedure that achieves the feature-conditional coverage guarantee~\eqref{eqn:in_exp_guarantee_X_conditional}. This procedure serves as the first key step in deriving the main procedure, pro-CP, for general feature distributions.

\item We then introduce \textit{propensity score $\eps$-discretization}, which constitutes the second key step. We show that the proposed propensity score discretization scheme induces ``approximate exchangeability" within the bins formed by the discretization. 

\item By combining these two key steps, we propose the main procedure---\textit{propensity score discretization-based conformal prediction (pro-CP)}---and demonstrate that it achieves the bin-conditional coverage guarantee~\eqref{psmcc} at level $1-\alpha-\eps$, when the propensity score is known.

\item For the case where the propensity score is unknown and an estimate is used, we derive a bound on the additional error in the conditional coverage.
\end{enumerate}

Additionally, as an extension, we present in Appendix~\ref{sq} a procedure that satisfies an even stronger condition---namely, a squared coverage guarantee.

\subsection{Inference for missing outcomes with discrete features}

We begin with a simpler case, where the feature distribution is discrete---or more generally, has atoms or point masses. 
We introduce a procedure that achieves\footnote{The procedure we introduce in this subsection attains the target guarantee in a completely distribution-free sense, but its usefulness is limited to the case of discrete features.} the conditional coverage guarantee~\eqref{eqn:in_exp_guarantee_X_conditional}.
While our main focus is on general feature distributions, the procedure for the discrete case serves as an important step toward deriving our main method, pro-CP. 

\subsubsection{Naive approach---conformal prediction for each distinct feature}\label{naive}
A direct 
approach to achieve \eqref{eqn:in_exp_guarantee_X_conditional}---which we present just as a warm-up example and a baseline for the case of discrete features---is to run 
standard conformal prediction for each distinct value of $X_i$, $i\in[n]$. 
To make this clear, let
$M\ge1 $
 be the number of distinct values in $(X_1,X_2,\ldots,X_{n+m})$, and let $\{X_1',X_2',\ldots,X_M'\}$ be those values. 
For each $k\in [M]$, define the sets
\[I_k = \{i \in [n+m] : X_i = X_k'\},\,\,
I_k^1 = \{i \in [n] : X_i = X_k'\}, \,\,
I_k^0 = \{j \in [m] : X_{n+j} = X_k'\},\]
and let $N_k = |I_k|, N_k^0 = |I_k^0|, N_k^1 = |I_k^1|$. 
Let $s:\mX\times\mY\to\R$ be a score function, constructed based on independent data. See \cite{vovk2005algorithmic,angelopoulos2021gentle} for standard examples.
For example, one can apply data splitting and construct an estimated mean function $\hat{\mu}(\cdot)$ with one of the splits by applying any regression method, and then choose to work with the residual score $s: (x,y) \mapsto |y-\hat{\mu}(x)|$---the following procedure is then applied to the other split. 

Then for each unique feature index $k \in [M]$, 
one can construct a standard split conformal prediction set for $\{Y_{n+j} : j \in I_k^0\}$ as
\begin{equation}\label{eqn:discrete_split_conformal}
    \ch(X_k') = \left\{y \in \Y : s(X_k',y) \le Q_{1-\alpha}\left(\sum_{i \in I_k^1} \frac{1}{N_k^1+1}\cdot \delta_{s(X_k',y)} + \frac{1}{N_k^1+1}\cdot\delta_\infty\right)\right\}.
\end{equation}
Such a set has the well-known property
 that for all $j \in I_k^0$, $\PPst{Y_{n+j} \in \ch(X_k')}{X_k'} \ge 1-\alpha$~\citep{papadopoulos2002inductive,vovk2005algorithmic}.
By a simple calculation, this implies the guarantee\eqref{eqn:in_exp_guarantee_X_conditional}.

While the prediction sets from~\eqref{eqn:discrete_split_conformal} provide
valid distribution-free inference,
they can be excessively wide to be informative. 
For example, if there is a missing outcome with a ``rare''  feature value, i.e., where $N_k$ is small, \eqref{eqn:discrete_split_conformal} 
can be the entire set $\Y$. We introduce below an alternative procedure, which provides a uniform bound on the scores of missing outcomes.

\subsubsection{Conformal-type method for simultaneous inference}
\label{hier}

Next, we discuss an approach that can pool datapoints across feature values.
A key technical observation is that the 
coverage property~\eqref{eqn:in_exp_guarantee} or~\eqref{eqn:in_exp_guarantee_X_conditional}
is equivalent to coverage for a randomly chosen missing outcome. 
Suppose we draw an index $j^*$ from 
the uniform measure
$\textnormal{Unif}([m])$, independently of the data. Then the coverage rate of $\widehat{C}(X_{n+j^*})$ can be represented as
\begin{equation}\label{isu}
  \PP{Y_{n+j^*} \in \widehat{C}(X_{n+j^*})} = \EE{\EEst{\One{Y_{n+j^*} \in \widehat{C}(X_{n+j^*})}}{(X_i,Y_i)_{i \in [n+m]}}} = \EE{\frac{1}{m}\sum_{i=1}^m \One{Y_{n+j} \in \widehat{C}(X_{n+j})}}.
\end{equation}
%Note that the right hand side is the term that appears in the target guarantee.
Similar representations are also possible for the conditional coverage guarantee~\eqref{eqn:in_exp_guarantee_X_conditional}.
To bound $\PP{Y_{n+j^*} \in \widehat{C}(X_{n+j^*})}$, we use that conditionally on $X_{1:(n+m)}$, 
the distribution of all outcomes
$Y_{1:(n+m)}$ is invariant under the group of permutations that keeps the feature values fixed.
%; which is a specific form of distributional invariance. %\citep{dobriban2023symmpi}.  
% Based on this observation, we construct a prediction set that follows the general principle of prediction sets for data with such invariances \citep{dobriban2023symmpi}.
We construct a score $s$ from a separate dataset, and define $S_i = s(X_i,Y_i)$ for $i\in[n]$.
For each $x\in \mathcal{X}$, let
\beq\label{eqn:CI_MAR_E}
\ch(x) = \left\{y \in \Y : s(x,y) \le Q_{1-\alpha}\left(\sum_{k=1}^M \sum_{i \in I_k^1}\frac{1}{m}\cdot \frac{N_k^0}{N_k}\cdot \delta_{S_i} + \frac{1}{m}\sum_{k=1}^M \frac{(N_k^0)^2}{N_k} \cdot\delta_{+\infty}\right)\right\}.
\eeq
In $\ch(x)$, the score $s$ is bounded above uniformly for any value $x$ of the features. 
Hence these sets are likely to be better controlled over $x$ than the standard conformal ones.
We prove the following result, 
under the assumption that 
the random variables within each collection $(Y_i : i \in I_k)$, $k \in [M]$ are  simultaneously exchangeable\footnote{We use the term ``simultaneous exchangeability" to refer to a set of random variables being invariant in distribution under an associated group of permutations.} 
conditional on $X_{1:(n+m)}$---which is a weaker assumption than the model~\eqref{eqn:model}.

\begin{theorem}\label{thm:MAR_in_exp}
Suppose that  
the random variables within each collection $(Y_i : i \in I_k)$, $k \in [M]$ are  simultaneously exchangeable 
conditional on $X_{1:(n+m)}$.
Then, the prediction set $\ch$ from~\eqref{eqn:CI_MAR_E} satisfies the feature-conditional coverage guarantee~\eqref{fmc}.
\end{theorem}

In the prediction set~\eqref{eqn:CI_MAR_E}, feature values without any missing outcomes are not used for  inference. 
Specifically, in the prediction set~\eqref{eqn:CI_MAR_E}, 
scores $S_i$ with $i\in I_k^1$ 
have a zero point mass if $N_k^0=0$.
This is reasonable, 
since in a distribution-free setting where $P_{Y \mid X}$ is unrestricted, 
the outcomes for one feature value do not provide information about the conditional distribution of the outcome at another feature value. 

\begin{remark}
In the proof of Theorem~\ref{thm:MAR_in_exp}, we show that
\[\EEst{\frac{1}{m}\sum_{j=1}^m \One{Y_{n+j} \in \widehat{C}(X_{n+j})}}{(X_i)_{i \in [n+m]}} \ge 1 - \alpha, \qquad \textnormal{almost surely,}\]
which is a stronger guarantee than the condition in~\eqref{fmc}, as it additionally conditions on the calibration feature observations $(X_i)_{i \in [n]}$. 
Nevertheless, we choose to adopt Definition~\eqref{fmc} as the main representation of the guarantee to make the test-conditional nature of the inference clear.

\end{remark}

\subsubsection{Constructing narrower prediction sets via partitioning the test set}
\label{pa}

The method from~\eqref{eqn:CI_MAR_E} may still provide a conservative prediction set 
if the overall missingness probability is high. 
Indeed, suppose the proportion of unobserved outcomes for each value of $x$ is around $\tau > 0$. 
Then  the mass at $+\infty$ in~\eqref{eqn:CI_MAR_E} is approximately
$\frac{1}{m}\sum_{k=1}^M (N_k^0)^2/N_k \approx \frac{1}{m}\sum_{k=1}^M \tau\cdot N_k^0 = \tau$.
Thus, if $\alpha\lesssim\tau$, 
the quantile determining the upper bound equals $+\infty$, and hence the prediction set is trivial.
To deal with such cases, we 
discuss a more general procedure that covers the previous two procedures~\eqref{eqn:discrete_split_conformal} and~\eqref{eqn:CI_MAR_E} at its two extremes.
We will see that we can choose an intermediate setting to avoid the problems of the two extremes.

Let $U=\{U_1,U_2,\ldots,U_L\}$ be a partition of $[m]$, and let $N_{\ell}^0 = |U_{\ell}|$ for $\ell \in [L]$. Now, Theorem~\ref{thm:MAR_in_exp} 
holds if the outcomes are independent conditional on the feature observations and missingness indices. 
Thus, we can apply the procedure $\ch$ 
from \eqref{eqn:CI_MAR_E}
to the subset $(X_{n+j})_{j \in U_{\ell}}$ of the test inputs, obtaining a prediction set function $\ch^{\ell}$ such that
\[\EEst{\frac{1}{N_{\ell}^0}\sum_{j \in U_{\ell}} \One{Y_{n+j} \in \ch^{\ell}(X_{n+j})}}{(X_{n+j})_{j \in U_{\ell}}}\ge 1-\alpha.\]
For indices $j \in [m]$  of unobserved outcomes,
let $\ell_j$ denote the unique partition index $\ell \in [L]$ such that $U_{\ell}$ contains $j$.
Repeating this procedure for all $\ell = 1,2,\ldots,L$, we obtain prediction sets $\{\ch^{\ell_j}(X_{n+j}) : j \in [m]\}$. These sets satisfy the feature-conditional coverage guarantee, as a direct consequence of Theorem~\ref{thm:MAR_in_exp}.

\begin{corollary}\label{cor:in_exp_discrete_partition}
Under the assumptions of Theorem~\ref{thm:MAR_in_exp},
the prediction sets $\{\ch^{\ell_j}(X_{n+j}) : j \in [m]\}$ satisfies the feature-conditional coverage guarantee~\eqref{fmc}.
\end{corollary}

At one extreme, if we choose the set of singletons $U^{(1)} = \{\{j\} : j\in[m]\}$ as the partition $U$, the above procedure reduces to running split conformal prediction separately for each test point, using all calibration data points with the same feature value as the holdout set, as in~\eqref{eqn:discrete_split_conformal}. At the other extreme, if we choose $U^{(m)}=\{[m]\}$, then the procedure $\ch_{U^{(m)}}$ is equivalent to~\eqref{eqn:CI_MAR_E}.

Generally, $U$ can be any partition that is independent of $Y_{1:(n+m)}$ conditionally on $X_{1:(n+m)}$. 
For example, it can be determined using a separate dataset, such as the one used to construct the score. Alternatively, 
it may depend on $X_{1:(n+m)}$, aiming to achieve a small probability mass on $+\infty$ in each prediction set. 
This can be achieved by ensuring that only test datapoints for each feature value are included in any given element of $U$.
In Section \ref{testpart}, we develop some optimal partitioning methods based on integer programming. We also show that when there are an equal number of test datapoints for each distinct feature value, then partitioning such that each partition element has one test datapoint per distinct feature allows for the largest coverage without trivial (full-$\mathcal Y$) prediction sets.  
In our experiments, we find that partitioning the test set uniformly at random into a small number---say $L=10$---of partition elements works well.

\subsection{Conditional inference for general feature distributions}\label{sec:conformal_discretization}

The above methods
provide simultaneous distribution-free inference for missing outcomes, and are useful 
when the feature distribution has a small support size compared to the sample size $n$, 
so that we have large enough 
numbers $N_1,N_2,\ldots,N_M$ of repeated feature values. Now we discuss methods for more general feature distributions.

One can consider discretizing the observed features,
so as to repeatedly sample datapoints falling within each feature-bin as before. However, 
this is not straightforward, 
since $Y$ and $A$ may not be
independent conditional on the discretized $X$, i.e., the missing at random assumption may not hold for the discretized data. 
To overcome this challenge, 
we propose discretizing based on the 
 propensity score $x\mapsto p_{A \mid X}(x) = \PPst{A = 1}{X=x}$, 
 in which case we show that we  retain \emph{approximate independence} after discretization.
%We discuss below a strategy to achieve valid inference. 
%Throughout this section we will additionally assume that the missingness probability is between zero and 1.
We first consider the case where we 
know 
the propensity score, 
as in randomized trials
and two phase sampling \citep{breslow2007weighted,saegusa2013weighted}.
We then characterize the impact of the additional uncertainty one must incur when, as typically the case in practice, the propensity score is not known and therefore must be estimated from the observed sample. 
%In practice, it can be estimated using a split of dataset, and we will show empirically that the inference using the estimates still tends to provide a correct coverage. 

We also mention that, while a direct per-bin application of weighted conformal prediction \citep{tibshirani2019conformal} can attain theoretically valid coverage, it is severely hampered because it does not allow pooling datapoints across bins, hence reducing the effective sample size. In contrast, by using hierarchical exchangeability as in Section \ref{hier}, pro-CP is more effective as it is able to pool datapoints across bins.

\subsubsection{Propensity score \texorpdfstring{$\ep$}{epsilon}-discretization}
\label{pss}

Given the propensity score function $p_{A \mid X}$, 
we choose a discretization level $\eps > 0$ and construct the partition \eqref{eqn:eps_partition} of the feature space $\X$ with $z_k = (1+\eps)^k/[1+(1+\eps)^k]$ for each $k \in \mathbb{Z}$. 
This is a valid partition if $0 < p_{A \mid X}(x) < 1$ for all $x \in \X$. 
Each bin $D_k$ contains feature values with similar odds values of the propensity score---by construction, for any $x \in D_k$, it holds that
\begin{equation}\label{di}
(1+\eps)^k \le \frac{p_{A \mid X}(x)}{1-p_{A \mid X}(x)} < (1+\eps)^{k+1}.
\end{equation}
%In other words, we partition the space $(0,\infty)$ of the odds of propensity scores with the intervals $[(1+\eps)^k,(1+\eps)^{k+1})$, and then consider the ``induced partition''  of the feature space. 
We call this step \textit{propensity score $\ep$-discretization}.

We prove the following property, which serves as a key lemma for the main theorem.

\begin{lemma}[Bounded propensity score implies closeness of conditional distributions for observed and missing outcome]\label{lem:dtv_0_1}
    Suppose that $(X,Y,A) \sim P_X \times P_{Y \mid X} \times \textnormal{Bernoulli}(p_{A \mid X})$ on $\X \times \Y \times \{0,1\}$, 
    and that for a set $D \subset \X$ and $t \in (0,1)$, $\ep\ge 0$,
    \[t\le \frac{p_{A \mid X}(x)}{1-p_{A \mid X}(x)} \le t(1+\eps), \textnormal{ for all } x \in D.\]
    Let $s : \X \times \Y \rightarrow \R$ by any measurable function and let $S = s(X,Y)$. 
    Then
    $\dtv(P_{S \mid A=1, X \in D}, P_{S \mid A=0, X \in D}) \le \eps$.
\end{lemma}

Lemma~\ref{lem:dtv_0_1} essentially states that if we construct bins based on $\ep$-discretization of the propensity score, then the distribution of missing outcomes within each bin is approximately the same as that of the observed outcomes, with the total variation distance controlled by $\ep$. We apply this property to deduce ``approximate within-bin exchangeability", which enables conformal-type predictive inference with a provable coverage guarantee.

\begin{remark}
From another perspective, the result of Lemma~\ref{lem:dtv_0_1} can also be interpreted as stating that the ``$\ep$-discretized propensity score exhibits $\ep$-approximate balancing". 
The propensity score has the balancing property \citep{rosenbaum1983central}, i.e., $Y \indep A \mid p_{A \mid X}(X)$. 
Now, 
Lemma~\ref{lem:dtv_0_1} with the score function $s(x,y) = y$
implies that $Y$ is ``approximately independent'' 
of $A$ conditional on the discretized propensity score.
Here, 
conditioning on the event $X \in D_k$ is equivalent to conditioning on $\lfloor\log_{1+\ep} \frac{p_{A \mid X}(X)}{1-p_{A \mid X}(X)}\rfloor = k$, i.e, the event that the discretized propensity score equals $k$. 
In summary, the proposed propensity score $\ep$-discretization leads to approximate independence conditional on the discretized score, with the ``approximate independence" being characterized by total variation distance bounded by $\eps$.
\end{remark}

\subsubsection{Propensity score discretized feature-conditional coverage control with pro-CP}
Now we consider the procedure where we apply~\eqref{eqn:CI_MAR_E} to the discretized data
obtained via propensity score $\ep$-discretization. 
Define $B_i$ for $i \in [n]$
as in Definition \ref{proc}
and
apply the procedure~\eqref{eqn:CI_MAR_E} on the data $(B_i, Z_i)_{i \in [n]}$ and $(B_{n+j})_{j \in [m]}$, where $Z_i = (X_i,Y_i)$ for $i \in [n+m]$, with score $s(b,z) = s(x,y)$ for all $b,x,y$. 

Write $S_i = s(X_i,Y_i)$  for all $i\in[n+m]$. Let $\{B_1',B_2',\ldots,B_M'\}$ be the set of distinct values in $(B_1,B_2,\ldots,B_{n+m})$, 
and for each $k \in [M]$, define
\[I_k^\B = \{i \in [n+m] : B_i = B_k'\},\,\,
I_k^{\B,0} = \{i \in [n] : B_i = B_k'\},\,\,
I_k^{\B,1} = \{j \in [m] : B_{n+j} = B_k'\},\]
and let $N_k^\B = |I_k^\B|, N_k^{\B,0} = |I_k^{\B,0}|, N_k^{\B,1} = |I_k^{\B,1}|$.
Here $I_k^\B$ is the index set of datapoints in a specific bin,
where datapoints with unobserved and observed outcomes are indexed by $I_k^{\B,0}$ and $I_k^{\B,0}$, respectively.
We propose the following procedure, which we call \textit{propensity score discretization-based conformal prediction (pro-CP)}.
For all $x\in\mathcal{X}$, let
\beq\label{eqn:pro_CP}
\chp(x) = 
\left\{y \in \Y : s(x,y) \le Q_{1-\alpha}\left(\sum_{k=1}^M \sum_{i \in I_k^{\B,1}}\frac{N_k^{\B,0}}{ m N_k^\B} \cdot \delta_{S_i} 
+ \frac{1}{m}\sum_{k=1}^M \frac{(N_k^{\B,0})^2}{N_k^\B} \cdot\delta_{+\infty}\right)\right\}.
\eeq
We prove the following result.
\begin{theorem}[Coverage of pro-CP]\label{thm:MAR_known_p_A_X}
   Suppose $0 < p_{A \mid X}(x) < 1$ holds for all $x \in \X$. Then
   the pro-CP procedure
   $\chp$ from~\eqref{eqn:pro_CP} satisfies
the propensity score discretized feature-conditional coverage
as per Definition \ref{proc}
at level $1-\alpha-\ep$.
\end{theorem}

Thus the coverage in Theorem~\ref{thm:MAR_known_p_A_X} 
is lower 
than the target coverage level $1-\alpha$ by at most $\eps$, 
where $\eps$ is due to the discretization step. 
We will see in experiments that 
this bound represents a worst case scenario, and 
the coverage empirically tends to still be close to $1-\alpha$. 
Observe that, to attain a provable $(1-\alpha')$-coverage for some $\alpha'\in[0,1]$, we can set 
$\eps$ and $\alpha$ appropriately---for instance, $\alpha = 0.8\alpha'$ and $ \eps = 0.2\alpha'$. That is, the pro-CP procedure provides an exact distribution-free control of the relaxed feature-conditional guarantee.

% Propensity score $\ep$-discretization 
% controls the discretization 
% error term at the level $\eps$.
We provide a brief intuition 
for the discretization strategy and the proof of Theorem~\ref{thm:MAR_known_p_A_X}. 
The method ensures 
that the odds ratio of the propensity scores of two features $x_1$ and $x_2$ in the same bin is in the interval $[({1+\eps})^{-1}, 1+\eps]$. This results
in the approximate exchangeability of outcomes in the same bin. We show that such an exchangeable distribution is within a 
total variation distance of $\eps$ after binning, leading to our bound. 
The proof relies on a
 new theoretical result (\Cref{lem:dtv_0_1}) showing that propensity score discretization leads to approximate independence between the distributions of $Y$ and $A$.
We think that this result may have further uses, such as in studying propensity score matching in causal inference \citep{abadie2016matching}.

%In other words, we construct bins such that the points in the same bin approximate the MCAR setting with a total variation distance bounded by $\eps$.

Similarly to the discrete case from \Cref{pa}, we can obtain narrower prediction sets via partitioning. Specifically, let $U = \{U_1,U_2,\ldots,U_L\}$ be a partition of $[m]$, and let $\ch^{\ell}$ be the prediction algorithm obtained by applying the procedure $\chp$ from~\eqref{eqn:pro_CP} to the subset $(X_{n+j})_{j \in U_{\ell}}$ of the test data, for each $\ell \in [L]$. 
Then for $j \in [m]$, define
$\chp_U(x,j) = \ch^{\ell_j}(x)$, where for $j \in [m]$, 
$\ell_j$ denotes the unique $\ell \in [L]$ that contains $j$. Then by similar arguments, we can show that
$\chp_U$ satisfies propensity score discretized feature-conditional coverage as per Definition \ref{proc}
at level $1-\alpha-\ep$---with prediction sets $\chp_U(X_{n+j},j)$, $j \in [m]$. Again, applying this procedure with $U = {[m]}$ recovers the original pro-CP procedure. The complete procedure, including this partitioning strategy, is summarized in Algorithm~\ref{alg:pro_CP}.

\begin{algorithm}
\caption{Propensity score discretization-based conformal prediction (pro-CP)}
\label{alg:pro_CP}
\textbf{Input:} Calibration data $(X_i,Y_i)_{i \in [n]}$, test inputs $(X_{n+j})_{j \in [m]}$, partition $U = \{U_1,\cdots,U_L\}$ of $[m]$, score function $s : \X \times \Y \rightarrow \R^+$, propensity score function $p_{A \mid X} : \X \rightarrow (0,1)$, discretization level $\eps \in (0,1)$.

\textbf{Step 1:} Compute the calibration scores $S_i = s(X_i, Y_i)$ for all $i \in [n]$.

\textbf{Step 2:} Discretize the features $(X_i)_{i \in [n+m]}$ based on propensity score $\eps$-discretization:
\[B_i = \left\lfloor\log_{1+\eps} \left( \frac{p_{A \mid X}(X_i)}{1 - p_{A \mid X}(X_i)} \right)\right\rfloor,\quad \text{for all } i \in [n+m].\]

\textbf{Step 3:} Identify the distinct values in $(B_i)_{i \in [n+m]}$ and denote them by $\{B_1', \cdots, B_M'\}$.

\textbf{Step 4:} For each $\ell = 1, \ldots, L$, do:

\textbf{Step 4-1:} Define:
\[I_k^\B = \{i \in [n] \cup U_{\ell} : B_i = B_k'\},\quad
I_k^{\B,0} = \{i \in [n] : B_i = B_k'\},\quad
I_k^{\B,1} = \{j \in U_{\ell} : B_{n+j} = B_k'\},\]
for $k \in [M]$, and compute their sizes: $N_k^\B = |I_k^\B|, N_k^{\B,0} = |I_k^{\B,0}|, N_k^{\B,1} = |I_k^{\B,1}|$.

\textbf{Step 4-2:} For each $j \in U_{\ell}$, construct the prediction set $\ch(X_{n+j})$ as:
\scalebox{0.95}{
$\ch(X_{n+j}) = 
\left\{y \in \Y : s(X_{n+j}, y) \le Q_{1-\alpha}\left(\sum_{k=1}^M \sum_{i \in I_k^{\B,1}} \frac{N_k^{\B,0}}{m N_k^\B} \cdot \delta_{S_i} 
+ \frac{1}{m} \sum_{k=1}^M \frac{(N_k^{\B,0})^2}{N_k^\B} \cdot \delta_{+\infty} \right)\right\}.$
}

\textbf{Return:} Prediction sets $(\ch(X_{n+j}))_{j \in [m]}$.
\end{algorithm}

\subsubsection{Optimality of propensity score discretization}

Moreover, we claim that the proposed propensity score discretization scheme is optimal in a sense described below. To introduce this result, consider a space $\mathcal{X}$ and a binning scheme $\mathcal{D} = \{D_\lambda : \lambda \in \Lambda\}$, where $\Lambda$ is an at most countable index set, that partitions the feature space $\mathcal{X}$. Let $P = P_X \times P_{Y \mid X} \times P_{A \mid X}$ be the distribution of $(X, A, Y)$.
For each $\lambda \in \Lambda$, 
consider sets $V_\lambda \subseteq \mathcal{Y}$, 
viewed as prediction sets for $Y$ corresponding to feature values in $D_\lambda$. These sets are considered fixed---e.g., by conditioning on the calibration data. The coverage of these sets $\mathcal{V} = (V_\lambda)_{\lambda \in \Lambda}$ under $A = a$ is given by
\[
\operatorname{Cover}(\mathcal{V}, P, A = a) = \sum_{\lambda \in [\Lambda]} \PP{X \in D_\lambda} \PPst{Y \in V_\lambda}{X \in D_\lambda, A = a}.
\]
The absolute difference between the coverage probabilities under $A = 0$ and $A = 1$---i.e., the error from covariate shift---is
\[
\Delta_\mathcal{V}(P) = \left| \operatorname{Cover}(\mathcal{V}, P, A = 0) - \operatorname{Cover}(\mathcal{V}, P, A = 1) \right|.
\]
We aim to design a binning scheme that controls this gap in a distribution-free 
sense—regardless of the distribution $P_{X,Y} = P_X \times P_{Y \mid X}$. 
That is, we seek a partition $\mathcal{D}$ such that the worst-case coverage gap
\[ \Delta_\mathcal{V}(P_{A|X}) := \sup_{P_{X,Y}} \Delta_\mathcal{V}(P) \]
is small. Next, define the \emph{propensity discretization error} as
$\mathcal{E}_\mathcal{V}(\mathcal{D}, P_{A|X})
:=\sup_{\lambda \in \Lambda_\mathcal{V}} \mathcal{E}(D_\lambda, P_{A|X})$,
where for $D\subset \mathcal{X}$,
\[
\mathcal{E}(D, P_{A|X}) :=
\sup_{x, x' \in D}
\left| 
{\dfrac{\PP{A=1 \mid X = x}}{1 - \PP{A=1 \mid X = x}}}
\Big/
    {\dfrac{\PP{A=1 \mid X = x'}}{1 - \PP{A=1 \mid X = x'}}} - 1 
\right|,
\]
and $\Lambda_\mathcal{V} = \{\lambda \in \Lambda : V_\lambda \neq \emptyset \text{ and } V_\lambda \neq \Y\}$ denotes the set of $\lambda$'s for which the corresponding set $V_\lambda$ is a nontrivial subset of $\Y$.
Indeed, \eqref{di} ensures that for propensity score $\ep$-discretization, $\mathcal{E}_\mathcal{V}(\mathcal{D}, P_{A|X}) \le \ep$.
The next result implies that this strategy is optimal in a sense, by showing that 
propensity score discretization error tightly controls the coverage gap.

\begin{theorem}[Propensity score discretization error controls the coverage gap]\label{opt}
Consider any 
feature space $\mathcal{X}$, any partition $\mathcal{D} = (D_\lambda)_{\lambda \in \Lambda}$ of $\mathcal{X}$, and any collection of sets $\mathcal{V} = (V_\lambda)_{\lambda \in \Lambda}$.
Then, for any
missingness distribution $P_{A|X}$,
the worst-case coverage gap $\Delta_\mathcal{V}(P_{A|X})$ is at least on the order of the propensity discretization error $\mathcal{E}_\mathcal{V}(\mathcal{D}, P_{A|X})$:
\[
\Delta_\mathcal{V}(P_{A|X}) \ge  \min\{\mathcal{E}_\mathcal{V}(\mathcal{D}, P_{A|X}),1\}/8.
\]
\end{theorem}

In most cases of interest, $\mathcal{E}_\mathcal{V}(\mathcal{D}, P_{A|X})$ is always less than unity, for instance $\mathcal{E}_\mathcal{V}(\mathcal{D}, P_{A|X}) \le 0.1$ in our examples.
Then, 
the above result 
$\Delta_\mathcal{V}(P_{A|X}) \ge  \mathcal{E}_\mathcal{V}(\mathcal{D}, P_{A|X})/8$
shows that a large propensity discretization error $\mathcal{E}_\mathcal{V}(\mathcal{D}, P_{A|X})$ implies a large coverage gap between $A = 0$ and $A = 1$. 

\begin{remark}[Optimality of propensity score discretization]
Propensity score $\eps$-discretization constructs a partition 
whose elements $D$ satisfy 
$\mathcal{E}(D, P_{A|X}) \le \ep$
and are as large as possible.
Indeed, consider any distribution
$P_{A \mid X}$ such that the random variable
$\frac{p_{A \mid X}(X)}{1-p_{A \mid X}(X)}$
has a continuous distribution on 
$[(1+\eps)^k,(1+\eps)^{k'})$ for some $k<k'$.
Then, clearly, propensity score $\varepsilon$-discretization provides the coarsest binning that ensures $\mathcal{E}_\mathcal{V}(\mathcal{D}, P_{A \mid X}) \le \eps$ for any $\lambda_\mathcal{V} \subset \Lambda$. Hence, propensity score $\eps$-discretization can be viewed optimal in this sense.    
\end{remark}

\subsubsection{Approximately valid inference via estimation of missingness probability}
In practice, the propensity score can be unknown an may need to be estimated.
If we apply pro-CP  
discretizing with 
an estimate
$\hat{p}_{A \mid X}$ instead of $p_{A \mid X}$, how much does the error in $\hat{p}_{A \mid X}$ affect the coverage?

Let $\ctp$ denote the pro-CP procedure from~\eqref{eqn:pro_CP}, applied with the propensity score $\ep$-discretization based on $\hat{p}_{A \mid X}$ instead of $p_{A \mid X}$. Define the odds ratio function
$f_{p,\hat{p}} : \X \rightarrow (0,\infty)$ 
 between $\hat{p}_{A \mid X}$ and $p_{A \mid X}$
by
\begin{equation}\label{eqn:f_p_phat}
f_{p,\hat{p}}(x) = \frac{p_{A \mid X}(x) / [1-p_{A \mid X}(x)]}{\hat{p}_{A \mid X}(x) / [1-\hat{p}_{A \mid X}(x)]}
\end{equation}
for all $x$.
 Then we prove the following.

\begin{theorem}\label{thm:MAR_estimated_p_A_X}
    Suppose $0 < p_{A \mid X}(x) < 1$ and $0 < \hat{p}_{A \mid X}(x) < 1$ hold for all $x \in \X$. Then $\ctp$ satisfies
    propensity score discretized feature-conditional coverage as per Definition \ref{proc}
at level $1-\alpha-(\eps+ \delta_{\hat{p}_{A \mid X}}+\eps\delta_{\hat{p}_{A \mid X}})$,
    where
    $\delta_{\hat{p}_{A \mid X}} = e^{2\|\log f_{p,\hat{p}}\|_\infty} - 1$.
\end{theorem}

This result shows that a uniformly accurate estimate 
$\log[{\hat{p}_{A \mid X}(x) / (1-\hat{p}_{A \mid X}(x))}]$
of the log-odds ratio 
$\log[{p_{A \mid X}(x) / (1-p_{A \mid X}(x))}]$
guarantees
that the loss in coverage is small compared to
the case of a known propensity score from
\Cref{thm:MAR_known_p_A_X}. Note that this is a worst-case error bound, and in practice, the actual error is often much smaller.

We briefly discuss scenarios where we may have a small error bound $\delta_{\hat{p}_{A \mid X}}$. First, we can consider settings where the propensity score follows a parametric model, allowing for accurate estimation. 
As a simple example,
suppose that the propensity score follows a 
single-index model 
$p_{A \mid X}(x) = \sigma(\beta^\top x)$ for an $L'$-Lipschitz function $\sigma$ for some $L'>0$, 
and that the feature space $\mathcal{X} \subset \mathbb{R}^d$ is bounded, i.e., $\|x\| \le C$ for some $C > 0$. 
Also, suppose there exists $0 < c < 1/2$ such that $c \le p_{A \mid X}(x) \le 1-c$ for all $x \in \mathcal{X}$. Then, for the least squares estimator $\hat{\beta}$ and the corresponding propensity score estimator $\hat{p}_{A \mid X}(x) = \hat{\beta}^\top x$, we have 
$$\left|\log{\frac{p_{A \mid X}(x)}{1-p_{A \mid X}(x)}} - \log{\frac{\hat{p}_{A \mid X}(x)}{1-\hat{p}_{A \mid X}(x)}}\right| 
\le \frac{|p_{A \mid X}(x) - \hat{p}_{A \mid X}(x)| }{c(1-c)}
\le \frac{L'}{c(1-c)} |(\beta-\hat{\beta})^\top x| \le \frac{CL'}{c(1-c)} \|\beta-\hat{\beta}\|$$
for any $x \in \mathcal{X}$ (where the first inequality applies the intermediate value theorem). Consequently, $\|\log f_{p,\hat{p}}\|_\infty$ has an upper bound that scales as $1/\sqrt{n}$ \citep{van2000asymptotic}.

We can also consider nonparametric settings; 
for instance 
let $\mathcal{X} \subset \mathbb{R}^d$ and 
for $\beta>0$, let $l\ge0$
be the integer part 
of $\beta$.
Suppose that for $L>0$,
$p_{A \mid X}$ belongs to a Hölder class $\Sigma(\beta,L)$
of functions
$f:\mathcal{X} \to \R$ 
such that 
for all tuples $(l_1,\ldots,l_d)$ of non-negative integers with $l_1+\ldots+l_d=l$,
one has 
$|f^{(l_1,\ldots,l_d)}(x) - f^{(l_1,\ldots,l_d)}(x')| \le L |x-x'|^{\beta-l}$, 
for all $x, x' \in \mathcal{X}$, where $f^{(l_1,\ldots,l_d)}$ is the $(l_1,\ldots,l_d)$-th partial derivative of $f$ with respect to $(x_1,\ldots,x_d)^\top$.
Then under certain assumptions,
it is known that a local polynomial estimator has an $L_\infty$ norm-error bound that scales as $((\log n)/n)^{\beta / (2\beta+d)}$~\citep{stone1982optimal,tsybakov2009}. Applying a similar argument as above, 
we then find a bound for $\|\log f_{p,\hat{p}}\|_\infty$ of the same order $((\log n)/n)^{\beta / (2\beta+d)}$, 
if for some $c>0$,
$c < p_{A\mid X}(x) < 1-c\; \text{ for all } x \in \X$.

\subsection{Use case---inference for individual treatment effects}
\label{ite}
We discuss a use case of the procedures we introduced, to
obtain inference for individual treatment effects (see e.g., \citet{hernan2020causal}). 
Suppose
\[(X_i,T_i,Y_i(0), Y_i(1))_{1 \le i \le n} \iidsim P_X \times P_{T \mid X} \times P_{Y(1) \mid X} \times P_{Y(0) \mid X},\]
where 
for observation unit $i \in [n]$, 
$X_i \in \X$ denotes the features, 
$T_i \in \{0,1\}$ denotes the binary treatment indicator, and $Y_i(1), Y_i(0) \in \R$ denote the counterfactual outcomes under treatment and control conditions. 
We 
make the standard consistency assumption 
where
we observe
$Y_i = T_i Y_i(1) + (1-T_i) Y_i(0)$,
$i \in [n]$ \citep{hernan2020causal}. 
%Observe that this can be viewed as having two sets of data with missing outcomes, $(X_i, A_i, Y_i(1) A_i)_{1 \le i \le n}$ and $(X_i, A_i', Y_i(0) A_i')_{1 \le i \le n}$, where $A_i = T_i$ and $A_i' = 1-T_i$.
The task is to achieve valid inference on individual treatment effects $Y_i(1) - Y_i(0)$, for the untreated individuals $I_{T=0} = \{i \in [n]: T_i = 0\}$.

We first observe that we can construct prediction sets for the unobserved counterfactuals $\{Y_i(1) : T_i = 0\}$ by applying the procedure for missing outcomes to the dataset $(X_i, T_i, Y_i(1)T_i)_{1 \le i \le n}$, to which we have full access (since $Y_i(1)T_i = Y_iT_i$). For example, if the treatment assignment probability is known, the pro-CP procedure can be applied with $(X_i, Y_i(1))_{i : T_i = 1}$ as the calibration data and $(X_i)_{i : T_i = 0}$ as the test inputs, to construct $\ch^\text{counterfactual}$ such that the following condition holds.

\[\EEst{\frac{1}{|I_{T=0}|}\sum_{i \in I_{T=0}} \One{Y_i(1) \in \ch^\text{counterfactual}(X_i)}}{(B_i)_{i \in I_{T=0}}} \ge 1-\alpha.\]

Next, since we have access to $Y_i(0)$ for the individuals in $I_{T=0}$, we can immediately construct prediction sets for $Y_i(1) - Y_i(0)$. 
Specifically, by letting $\ch^\text{ITE}_i = \{y-Y_i(0) : y \in \ch^\text{counterfactual}(X_i)\}$, we obtain

\begin{equation}\label{eqn:coverage_ITE}
    \EEst{\frac{1}{|I_{T=0}|}\sum_{i \in I_{T=0}} \One{Y_i(1)-Y_i(0) \in \ch^\text{ITE}_i}}{(B_i)_{i \in I_{T=0}}} \ge 1-\alpha.
\end{equation}
Thus, we obtain a 
simultaneous inference procedure for the individual treatment effects with valid coverage.

\section{Experimental results}
\label{exp}

\subsection{Simulations with illustrative examples}\label{sec:sim}
\label{sim}
We present simulation results to illustrate the performance of the proposed procedure. Here, we present results in a simple univariate feature setting to illustrate the difference between conditional and marginal coverage control. In the next section, we provide additional experiments under a more complex setting with multivariate features.

We generate the data $(X_i, A_i, Y_i A_i)_{1 \le i \le n}$ as follows:
%, rather than drawing the calibration and test data from some $P_{X \mid A=1}$ and $P_{X \mid A=0}$, respectively.
\begin{align*}
    X &\sim \textnormal{Unif}[0,10],\,\,
    Y \mid X \sim N(X, (3+X)^2),\, \,
    A \mid X \sim \textnormal{Bernoulli}(p_{A \mid X}(X)),
\end{align*}
where we consider 
two settings of $p_{A \mid X}$, such that for all $x\in[0,10]$:
\begin{align*}
    \textnormal{(1)}: p_{A \mid X}(x) = 0.9-0.02x,\quad
    \textnormal{(2)}: p_{A \mid X}(x) = 0.8-0.1(1+0.1x) \sin 3x.
\end{align*}
We then use $\{(X_i,Y_i) : A_i = 1\}$ as the calibration data an construct prediction sets for the unobserved outcomes $\{Y_i : A_i = 0\}$. The above distributions are illustrated in Figure~\ref{fig:scatterplot}.

\begin{figure}[ht]
  \centering
  \includegraphics[width=0.9\textwidth]{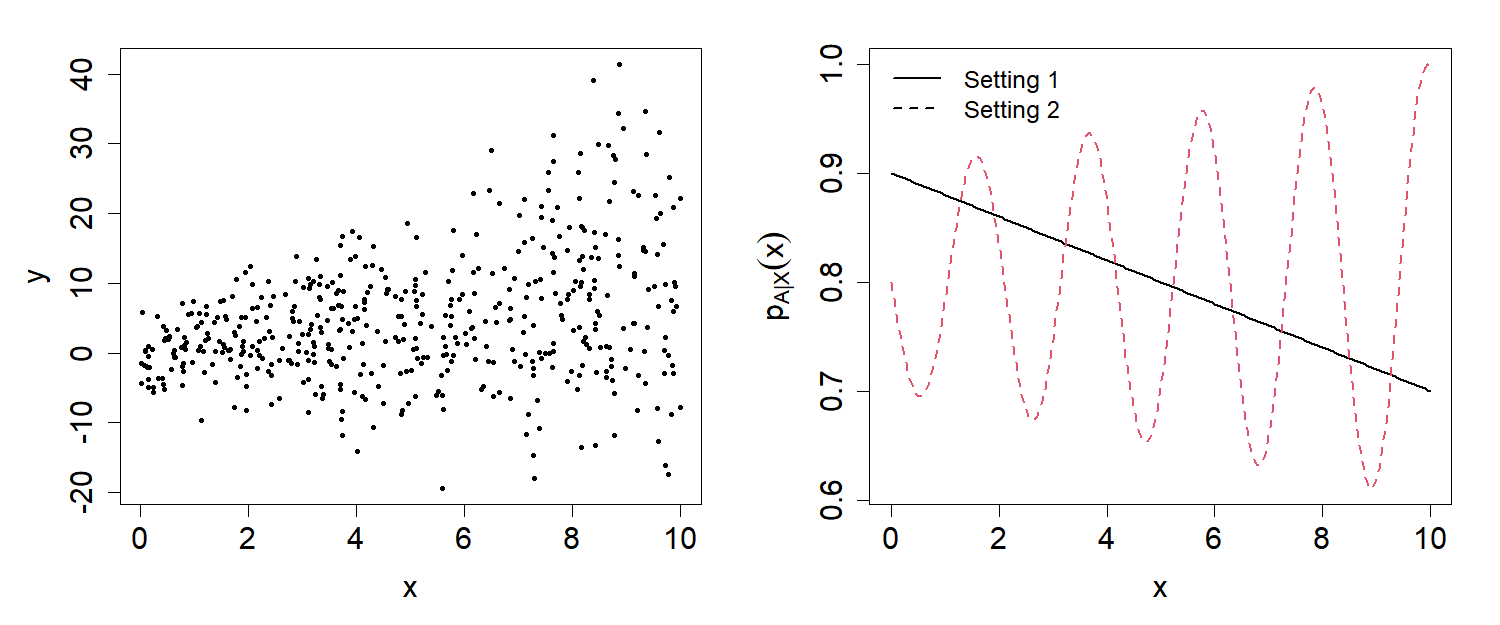}
  \caption{Scatterplot of the full dataset (without missing outcomes), and the graphs of the missingness probability in Settings 1 and 2.}
  \label{fig:scatterplot}
\end{figure}

In Setting 1,
the missingness probability is linearly decreasing in $x$, 
so that binning based on $p_{A \mid X}$ leads to intervals. 
In Setting 2, 
partition elements can include non-neighboring points with different spreads of $Y$ given $X$, 
and consequently the resulting coverage guarantee does not imply a local coverage property. 

We first generate a training dataset $(X_i',A_i',Y_i'A_i')_{1 \le i \le n_\text{train}}$ of size $n_\text{train} = 500$, and then fit quantile linear regressions on the subset $\{(X_i',Y_i') : A_i' = 1\}$ of the data to construct estimates $\hat{q}_{\alpha/2}(\cdot)$ and $\hat{q}_{1-\alpha/2}(\cdot)$ for the $\alpha/2$- and $(1-\alpha/2)$-conditional quantiles, respectively. Then we consider the quantile-based score
$s(x,y) = \max\{\hat{q}_{\alpha/2}(x) - y, y-\hat{q}_{1-\alpha/2}(x)\}$.

We first illustrate the conditional coverage  \eqref{psmcc} of pro-CP, and compare it to the marginal coverage  \eqref{eqn:in_exp_guarantee} 
achieved by applying weighted conformal prediction~\citet{tibshirani2019conformal} for each individual missing outcome.
We show the performance of the two methods in Setting 1 where the 
propensity score discretized feature-conditional coverage
(abbreviated as the bin-conditional coverage rate)  
can be accurately computed, and we further show their
feature-conditional coverage rate from \eqref{eqn:in_exp_guarantee_X_conditional}.

We run 500 independent trials, where in each trial we generate $(X_i,A_i)_{1 \le i \le n}$ of size $n=500$ and then apply propensity score $\ep$-discretization to obtain $(B_i,A_i)_{1 \le i \le n}$, with level $\eps=0.1$. 
Then we generate 
100 samples of $(X_i',Y_i')_{1 \le i \le n}$,
where $(X_i',Y_i') \mid B_i \sim P_{X \mid B} \times P_{Y \mid X}$.
For each sample, we apply $\chp_U$ to $(X_i',A_i,Y_i'A_i)_{1 \le i \le n}$ with
$U$ being an induced-partition from splitting $\{1,2,\ldots,500\}$ into ten intervals uniformly, i.e.,
\begin{equation}\label{eqn:U}
    U = \{\bar{U}_{\ell} \cap I_{A=0} : l \in [10]\}, \text{ where } \bar{U}_{\ell} = \{50\cdot(j-1)+1, 50\cdot(j-1)+2, \ldots, 50\cdot j \} \text{ for } l \in [10],
\end{equation}
where we let $I_{A = 0} = \{i \in [n] : A_i = 0\}$. The level is set as $\alpha=0.2$.
We also run weighted split conformal prediction, following the steps in~\citet{tibshirani2019conformal} with the weights
$w(x) = p_{A \mid X}(x)/[1-p_{A \mid X}(x)]$ for all $x$.
We take the average of the coverage rates of the two methods over the 100 repeats, 
to obtain an estimate of the bin-conditional coverage rate.
In each trial, 
we also generate 100 samples $(Y_i')_{1 \le i \le n}$ from $Y_i \mid X_i \sim P_{Y \mid X}$ and then apply the two methods, to compute the feature-conditional coverage rate. Figure~\ref{fig:bin_conditional_coverage_rate} shows the results.

\iffalse
\begin{figure}[htbp]
  \centering
  \includegraphics[width=\textwidth]{conditional_coverage_hist.png}
  \caption{Histograms of propensity score discretized feature-conditional (bin-conditional)
  and 
  feature-conditional coverage
  coverage rates of pro-CP and weighted split conformal prediction with absolute residual scores over 500 independent trials, in Setting 1.}
  \label{fig:bin_conditional_coverage_rate}
\end{figure}
\fi

\begin{figure}[htbp]
  \centering
  \includegraphics[width=\textwidth]{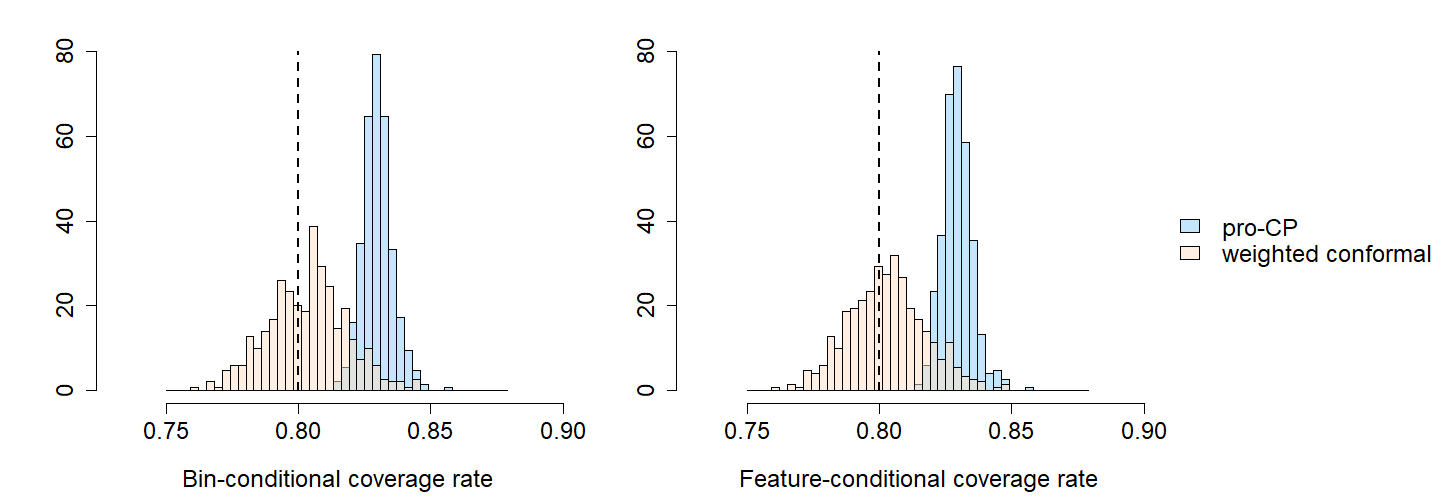}
  \caption{Histograms of propensity score discretized feature-conditional (bin-conditional)
  and 
  feature-conditional coverage
  coverage rates of pro-CP and the method of~\citet{lei2021conformal} conformal prediction over 500 independent trials, in Setting 1.}
  \label{fig:bin_conditional_coverage_rate}
\end{figure}

The result illustrates that both methods work as intended. 
Pro-CP controls the bin-conditional coverage rate~\eqref{psmcc} in every trial at coverage level $1-\alpha=0.8$. On the other hand, weighted conformal prediction allows the conditional coverage rate to be smaller than $1-\alpha$, to tightly attain the marginal coverage rate of $1-\alpha$. The feature-conditional coverage rates show similar trends, implying that the conditioning on the discretized features approximates conditioning on the features fairly well.
The theoretical lower bound
for the  bin-conditional coverage rate
provided by Theorem~\ref{thm:MAR_known_p_A_X} is $1-\alpha-\eps = 0.7$, 
but the procedure tends to control the conditional coverage rate above $1-\alpha$ in practice. 
This is because the $\eps$ term represents a worst-case scenario, which is not reflected here. 

Next,  Figure~\ref{fig:feature_conditional_coverage_rate} shows the distribution of feature-conditional coverage rates and the (feature-conditional expectation of) median widths---$\textnormal{median} \big(\big\{\text{leb}(\ch(X_i)) : A_i = 0\big\}\big)$---of the prediction intervals.
We use the median, 
since there is a nonzero probability---although small---that some prediction sets have infinite width---e.g., if an element of $U$ 
contains only one point with a missing outcome, and the corresponding discretized feature value appears only once in the data. 

\begin{figure}[htbp]
  \centering
  \includegraphics[width=0.9\textwidth]{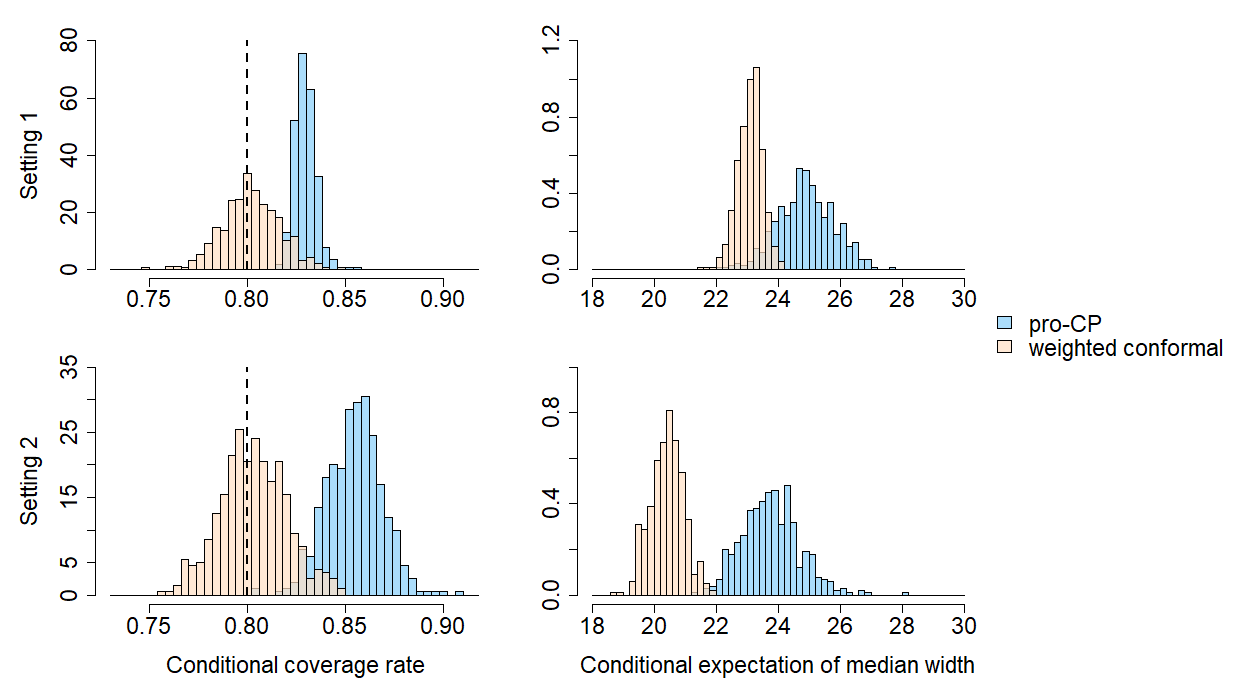}
  \caption{Histograms of (feature-)conditional coverage rates and expected median width of pro-CP and 
  weighted split conformal prediction 
  over 500 independent trials, in Settings 1 (top) and 2 (bottom).}
  \label{fig:feature_conditional_coverage_rate}
\end{figure}

The results show that in all trials, the conditional coverage rate~\eqref{eqn:in_exp_guarantee_X_conditional} is controlled at level $1-\alpha=0.8$ by pro-CP in both settings.
Again, weighted conformal prediction allows for a conditional coverage rate smaller than $1-\alpha$---it tightly controls marginal coverage at the level $1-\alpha$. Pro-CP attains the stronger conditional coverage guarantee by constructing wider prediction sets.

\subsubsection{Comparison with weighted conformal prediction for a single prediction}

Next, for additional illustration, we present the result for the case $m = 1$, i.e., we have a single test point $X_{n+1}$, and the target is the conditional coverage rate $\mathbb{P}\{Y_{n+1} \in \ch(X_{n+1}) \mid X_{n+1}\}$.
For this experiment, we fix the calibration size at $n = 500$, and run the procedures---pro-CP and weighted conformal prediction---with test input values $X_{n+1} = 0, 0.1, 0.2, \cdots, 10$. 
For each value of $X_{n+1}$, we compute the conditional coverage rate based on 500 repeated generations of calibration data and runs of the methods, at level $\alpha=0.2$. 

The results are shown in Figure~\ref{fig:single_conditinal_coverage}. 
The first two plots show the results under $\ep=0.1$. 
In Setting 1, the pro-CP procedure achieves conditional coverage rates around 0.8 for most test input values, while weighted conformal prediction---which only controls the marginal coverage rate---fails to control the conditional coverage for roughly half of the input values.
In Setting 2, weighted conformal prediction shows a similar trend, but the pro-CP procedure now sometimes fails to control the conditional coverage. 
This is because it theoretically controls the bin-conditional coverage, which may not accurately approximate the feature-conditional coverage in Setting 2.

To further illustrate this, we also provide results under $\ep = 0.02$ (third plot), which leads to finer binning and is more likely to yield a better approximation of the feature-conditional coverage by the bin-conditional coverage.
We observe that the pro-CP procedure now tends to control the conditional coverage rates so that they exceed the target level of 0.8 for almost all values of $X_{n+1}$. 
Note, however, that it can provide conservative prediction sets for small $\ep$ values unless the sample size is very large, as there may be only a few 
datapoints in the same bin as the test point. 
Since our main focus is on simultaneous inference on multiple test points, 
we may encounter a few test  datapoints with ``rare" feature values.
%---this additional experiment is provided as a straightforward illustration of conditional coverage control, rather than as a simulation of realistic settings where the proposed method is useful.

\begin{figure}[htbp]
  \centering
  \includegraphics[width=0.9\textwidth]{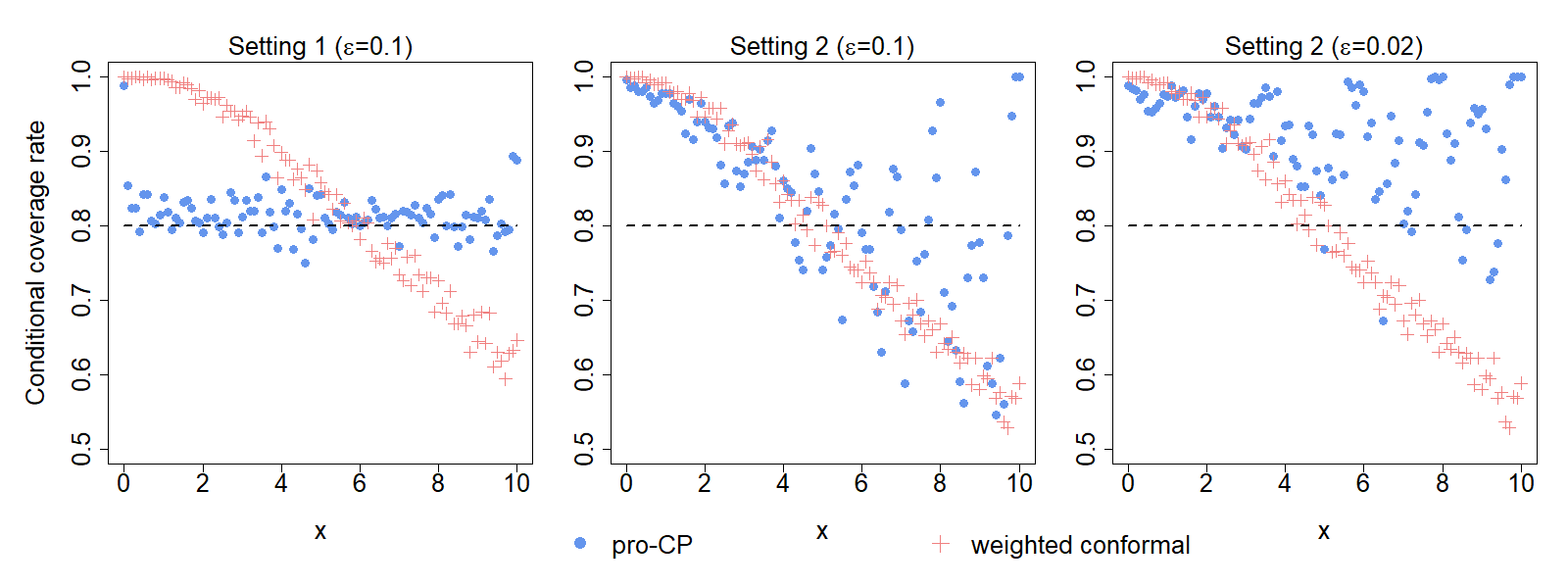}
  \caption{Conditional coverage rates of pro-CP and weighted split conformal prediction, in the case where the test size is one, in Settings 1 and 2.}
  \label{fig:single_conditinal_coverage}
\end{figure}

\subsection{Simulations with a higher-dimensional feature}\label{sec:sim_mult}
In this section, we move beyond simple illustrative examples and present additional simulation results in a higher-dimensional setting.
We draw an  i.i.d.~sample from the following distribution for a dimension $p=30$:

\[X \sim N_p(\mu,\Sigma),\; Y \mid X \sim N(\beta_0+\beta^\top X, \sigma_X^2),\; A \mid X \sim \textnormal{Bernoulli}\left(\frac{\exp(\gamma_0+\gamma^\top X)}{1+\exp(\gamma_0+\gamma^\top X)}\right).\]
We set $\mu = (1,1,\cdots,1)^\top$, $\Sigma = 2 \cdot I_p$, where $I_p$ denotes the $p \times p$ identity matrix, and $\sigma_X = \|X\|_2^2/p$. 
We set $\beta_0 = 5$, and each component of $\beta$ is randomly drawn from
the distribution $\textnormal{Unif}(-2,2)$. 
The parameters for the logistic model are set as $\gamma_0=1.2$ and $\gamma=(0.2,-0.3,0.2,0,0,\cdots,0)^\top$, resulting in an overall missingness probability of approximately 23\%.

We first demonstrate the conditional-coverage control of pro-CP. The simulation steps are analogous to those outlined in Section~\ref{sec:sim}. For the nonconformity score, we use the quantile-based score proposed by \citet{romano2019conformalized}. We compare the performance of pro-CP with that of~\citet{lei2021conformal}, which integrates weighted split conformal prediction~\citep{tibshirani2019conformal} with the quantile-based score of~\citet{romano2019conformalized} for inference on each missing outcome (or equivalently, individual treatment effect). We consider two settings: when we have access to the true propensity score, and when we use an estimated propensity score. 
For the estimation of the propensity score, we apply sparse logistic regression with $\ell_1$ penalization using the \texttt{glmnet} package in R, and the regularization strength 
is selected through cross-validation on the training data. Figure~\ref{fig:logistic_bin_conditional_coverage_rate} and~\ref{fig:higher_lc_box} show the results. 
We observe results similar to those from the low-dimensional setting of Section~\ref{sec:sim}. 
In both settings, the pro-CP procedure provides a better control of the conditional coverage rate, 
exceeding $1-\alpha$ in most trials, while the distribution of the conditional coverage rate of the method of~\citet{lei2021conformal} is centered around $1-\alpha$.

\begin{figure}[htbp]
  \centering
  \includegraphics[width=0.8\textwidth]{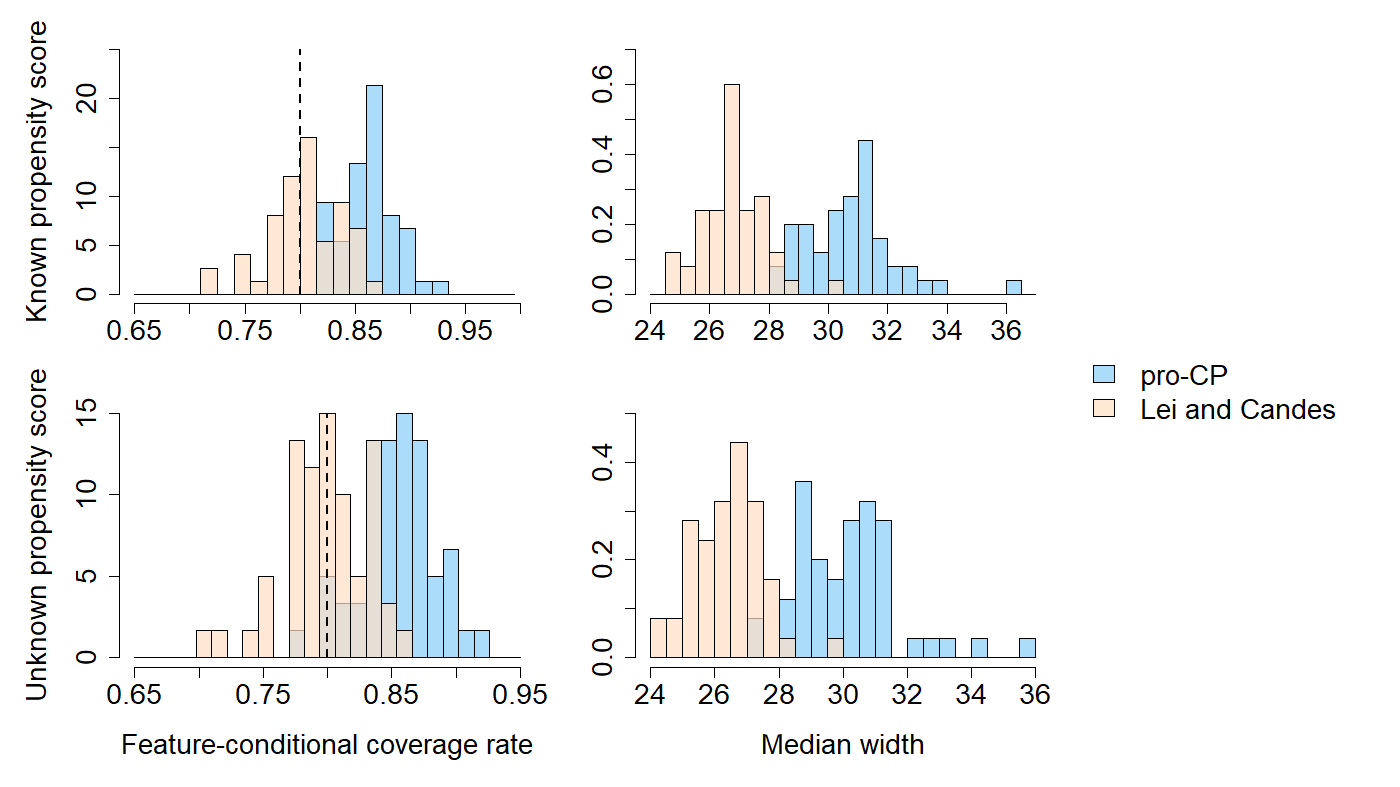}
  \caption{Higher-dimensional setting: Histograms of feature-conditional coverage rates and expected median width of pro-CP and weighted split conformal prediction (in the form discussed in~\citet{lei2021conformal}) over 500 independent trials. Top: known propensity score; Bottom: unknown propensity score.}
  \label{fig:logistic_bin_conditional_coverage_rate}
\end{figure}

\begin{figure}[htbp]
  \centering
  \includegraphics[width=0.8\textwidth]{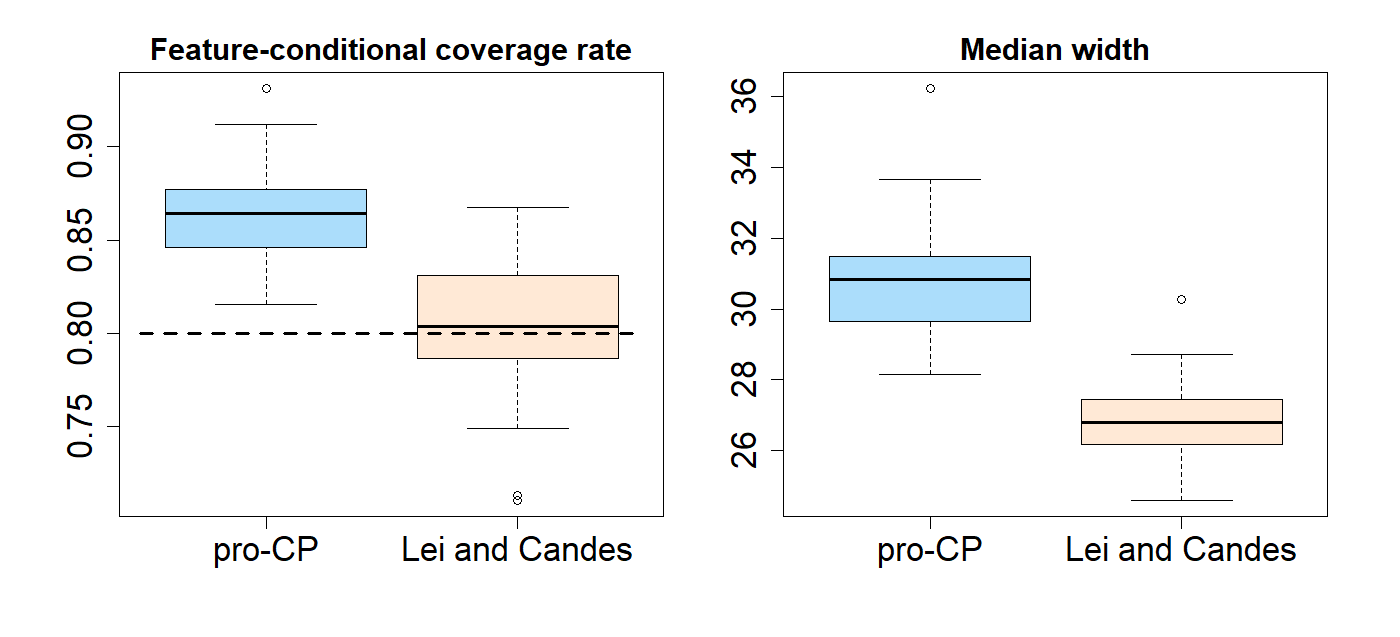}
  \caption{Higher-dimensional setting: Boxplots of feature-conditional coverage rates and expected median width of pro-CP and weighted split conformal prediction (in the form discussed in~\citet{lei2021conformal}) over 500 independent trials.}
  \label{fig:higher_lc_box}
\end{figure}

\subsubsection{Assessing the conservativeness of pro-CP}

Next, we explore the question: How conservative is pro-CP? Recall that the pro-CP procedure provides a conditional coverage guarantee, which is stronger than the marginal coverage guarantee, and therefore produces wider prediction sets than weighted conformal prediction. Given this tradeoff between the strength of the inferential target and the width of the prediction set, does pro-CP make a reasonable choice? Or does it achieve the stronger guarantee simply by being unnecessarily conservative? We address this question through additional experiments. 

For different values of the target level $\alpha$, we repeat the following process 100 times: generate datasets from the same distribution as above, run pro-CP and weighted split conformal prediction---both with quantile-based score, i.e., the weighted conformal prediction corresponds to the method of~\citet{lei2021conformal}---, and compute the median width and (marginal) coverage rates. We then average these results over the 100 trials to produce the width–coverage plot in Figure~\ref{fig:cov_width}. The results show that the two methods yield almost identical prediction set widths when they achieve the same marginal coverage rates. Thus, roughly speaking, the pro-CP method behaves like a ``level-shifted weighted conformal prediction". Importantly, in practice, applying weighted conformal prediction with a level adjustment to attain conditional coverage control is not feasible, as the practitioner does not know the amount of adjustment required. The pro-CP procedure, on the other hand, achieves conditional coverage guarantees with theoretical justification, without being unnecessarily conservative---it is essentially only as conservative as weighted conformal prediction. The wider prediction set from pro-CP should be interpreted as ``making a different choice in the tradeoff" to achieve a stronger target, rather than as being conservative. 

Note also that even if the appropriate level adjustment is known, shifting the level in weighted conformal prediction does not recover the pro-CP prediction sets. This is because the former applies weighted conformal prediction to individual test points separately, potentially resulting in different prediction set widths, whereas pro-CP outputs a shared width---more generally, a shared score bound---within each partition.

\begin{figure}[htbp]
  \centering
  \includegraphics[width=0.6\textwidth]{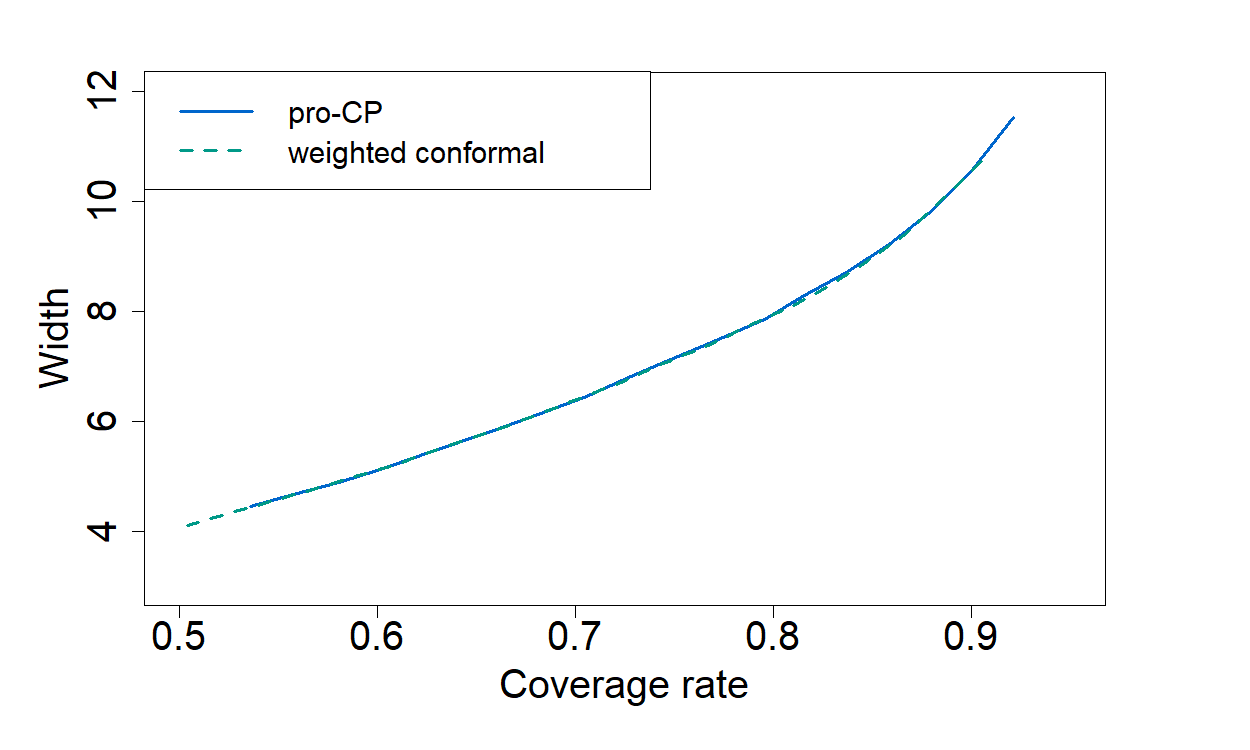}
  \caption{Width–coverage plot of prediction sets from pro-CP and weighted split conformal prediction (in the form discussed in~\citet{lei2021conformal}).}
  \label{fig:cov_width}
\end{figure}

\subsection{Application to a job search intervention study}\label{dat2}
We further illustrate the performance of the procedures on the JOBS II data set~\citep{imai2010general,DVNUMEYXD2010}. 
This dataset consists of observations from 1285 job seekers, before and after participating in a job skills workshop viewed as a treatment assigned to 879 participants, with a control group size of 406.
There are 14  features, such as demographic information of individuals and pre-treatment depression measures. The outcome variable is the post-treatment depression measure. 

We explore the performance of pro-CP---and weighted conformal prediction for comparison---for the task of simultaneously inferring individual treatment effects of the control group, 
as discussed in Section~\ref{ite}. 
Since we do not have access to the counterfactual outcomes, evaluating the prediction sets on the control group is not possible. 
To address this, we create 
a new control group by introducing missingness in the treatment group, and then estimate the coverage rate on the simulated control group. 
Although we do not have access to the counterfactual outcomes on the treatment group and consequently the simulated control group either, 
it is still possible to estimate the coverage rate of prediction sets, as that does not depend on the counterfactuals.
Recall that the prediction set for the ITE ($Y(1)-Y(0)$) is constructed by shifting the prediction set for $Y(1)$ by $Y(0)$. Thus, it is equivalent to estimate the coverage for $Y(1)$ before the shift.

We randomly split the treatment group into 
a training dataset of size $379$ 
and a calibration dataset of size $500$. 
Then we generate the missing outcomes (equivalently, the new control group) based on the logistic model
\begin{equation*}
    A \mid X \sim \textnormal{Bernoulli}\left(\frac{\exp(\beta^\top X)}{1+\exp(\beta^\top X)}\right),
\end{equation*}
with a fixed parameter $\beta$, resulting in approximately $22\%$ missingness. We construct an estimate of the propensity score using random forests and then compute the nonconformity score $s(x,y) = |y - \hat{\mu}(x)|$ by fitting $\hat{\mu}$ with random forest regression.
We then run the pro-CP procedure and the weighted split conformal prediction, using either the true or the estimated propensity score. For the pro-CP procedure, we apply the partitioning scheme based on $U$, constructed as in~\eqref{eqn:U}. Figure~\ref{fig:jobs} shows the coverage proportion (i.e., the term inside the expectation in~\eqref{eqn:coverage_ITE}) for the two procedures. Since the conditional coverage (the conditional expectation of the coverage proportion) cannot be evaluated from a single realized sample, we present this plot instead to illustrate the overall behavior of the methods. The results show that the coverage proportions from weighted conformal prediction are centered around the target level $1-\alpha$, whereas pro-CP tends to yield higher coverage proportions in most trials by producing slightly wider prediction sets.

\iffalse
\begin{table}[H]
\begin{center} 
\begin{tabular}{llll}
\hline
& $\EE{\text{coverage}}$ & $\PP{\text{coverage} \ge 1-\alpha}$ & $\EE{\text{median width}}$ \\
\hline
$\chp$ & 0.8394 (0.0018) & 0.8380 (0.0165)  & 1.4520 (0.0045)\\
$\chps$ & 0.8975 (0.0015) & 0.9980 (0.0020) & 1.7895 (0.0039)\\
$\ctp$  & 0.8665 (0.0017) & 0.9500 (0.0098) & 1.5510 (0.0046)\\
$\ctps$ & 0.9267 (0.0014) & 1.0000 (0.0000) & 1.9233 (0.0064)\\\hline
\end{tabular} 
\end{center}
\caption{Results of the JOBS II data set: The mean coverage proportion, probability of coverage proportion being larger than $1-\alpha$, 
and the mean of the \emph{median prediction interval width} of the prediction sets $\chp$, $\chps$, $\ctp$, and $\ctps$, with standard errors.}
\label{table:jobs}
\end{table}
\fi

\begin{figure}[ht]
  \centering
  \includegraphics[width=0.8\textwidth]{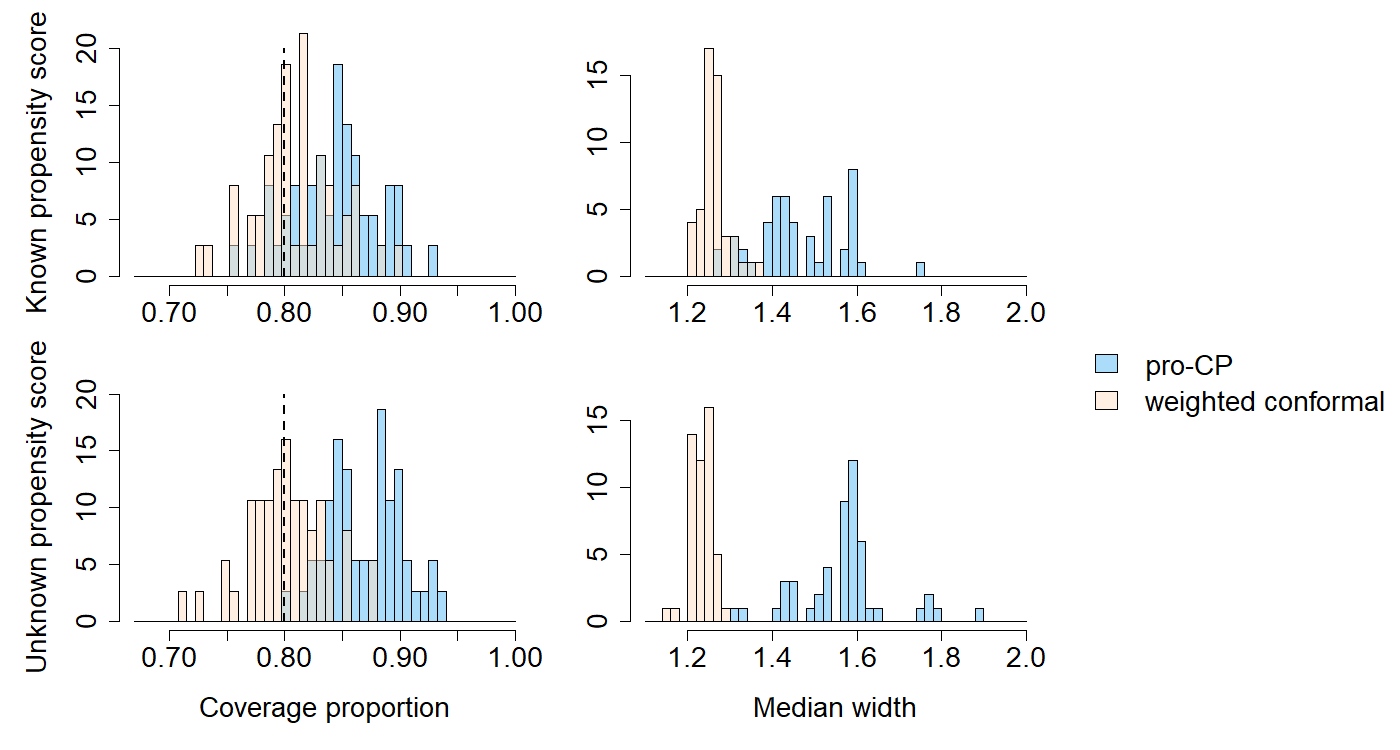}
  \caption{Results for the JOBS II data set: histograms of coverage proportion and median width of pro-CP and weighted split conformal prediction. Top: known propensity score; Bottom: unknown propensity score.}
  \label{fig:jobs}
\end{figure}

\section{Discussion}

In this work, we investigated
predictive inference for multiple unobserved outcomes, 
where the  propensity score can depend on the features. 
We proposed methods
that control the proportion of missing outcomes covered by the prediction sets, with marginal (in-expectation) and squared-coverage guarantees.

Several open questions remain.
Our procedures 
use binning and provide 
coverage conditionally on the bins that the features belong to. 
The bins are constructed based on the propensity score, but this leaves open the question of 
what an optimal binning scheme is.
Our method has strong 
theoretical properties, 
but might result in conservative prediction sets if the propensity score is close to zero or one with substantial probability. 
Indeed, our binning scheme is very fine-grained for those cases. 
Would simpler binning, e.g., uniform mass binning, work fairly well under additional assumptions?

For cases where an estimate of the propensity score is utilized for inference, our result provides a general error bound that depends on the accuracy of the estimator. Can we achieve a more refined, or doubly robust bound under a specific score function or by imposing a weak distributional assumption? We leave these questions to future work.

\section*{Acknowledgements}
This work was supported in part by 
NIH R01-AG065276, R01-GM139926, NSF 2210662, P01-AG041710, R01-CA222147, 
ARO W911NF-23-1-0296,
NSF 2046874, ONR N00014-21-1-2843, and the Sloan Foundation. 

\bibliographystyle{plainnat}
\bibliography{bib}

%\newpage

\appendix
% \bigskip
% \bigskip
\section*{Appendix}
{\bf Additional notation.}
We will use the following additional notation.
For a positive integer $n$,
we let $\mathcal{S}_n$ be the set of permutations of $[n]$. 

\section{Alternative naive method: binned wCP}
\label{bin-wCP}

Here, we aim to discuss another alternative naive method, binned weighted conformal prediction (binned wCP), a natural way of applying weighted conformal prediction \citep{tibshirani2019conformal} in our setting.
Suppose we bin each feature $X_i$, mapping it to $B_i = b(X_i)$, $i\in[n+m]$, with some map $b:\mathcal{X}\to \mathbb{Z}$. At the moment, we allow the binning map to be arbitrary.
Then, we can apply weighted conformal prediction \emph{for each specific bin separately}.
Specifically, 
letting 
%$F_k = \{j\in[n+m]: b(X_j) = k\}$, 
$F_k^0 = \{j\in[m]: b(X_{n+j}) = k\}$ for
$k\in \mathbb{Z}$, 
for each $k$ such that $F_k^0$ is non-empty, weighted conformal prediction provides a prediction set
$\ch^{k}(\cdot) = \ch^{k}(\cdot,; (X_{n+j})_{j \in F_{k}^0},(B_{n+j})_{j \in [m]})$
such that for each $j \in F_k^0$,

\[\PPst{Y_{n+j} \in \ch^{k}(X_{n+j})}{(B_{n+j})_{j \in [m]}}\ge 1-\alpha.\]
Hence, we also have
\begin{equation}\label{bcc}
    \EEst{\frac{1}{m}\sum_{j \in [m]} \One{Y_{n+j} \in \ch^{B_{n+j}}(X_{n+j})}}{(B_{n+j})_{j \in [m]}}\ge 1-\alpha.
\end{equation}
This shows the following proposition:
\begin{proposition}
    For any binning scheme,  binned weighted conformal prediction achieves the bin-conditional coverage guarantee \eqref{bcc}.
\end{proposition}

In particular, this holds for propensity-score $\ep$-discretization, and hence---if the propensity score is known---wCP satisfies
propensity score discretized feature-conditional coverage as per Definition \ref{proc}.
This gives it a seeming advantage compared to pro-CP, which, as per
Theorem \ref{thm:MAR_known_p_A_X}, only achieves coverage 
at level $1-\alpha-\ep$ in the same scenario.
However, this seeming advantage is completely washed out by the far greater disadvantage that binned wCP \emph{does not pool data across bins}, and thus has much fewer datapoints from which to estimate quantiles. In practice, this means that 
pro-CP can often achieve finite size prediction sets even when
binned wCP produces sets of an infinite width. In other words, while each prediction set from binned wCP is based on a small sample size, the data pooling with pro-CP increases the effective sample size used to construct each prediction set, leading to overall more informative prediction sets.

How does the pooling by pro-CP work? The core idea lies in the simultaneous approximate exchangeability across bins, as in Section \ref{hier}, which enables conformal-type inference using the pooled dataset.
%---and does not directly use weights to determine its quantile in \eqref{eqn:pro_CP}---which is readily usable when one pools datapoints across bins.
However, pooling arbitrary datapoints can introduce a large bias, and hence it is crucial to discretize the datapoints into bins appropriately in order to attain approximate exchangeability and to preserve approximate coverage, which is exactly what propensity score $\ep$-discretization achieves.

As a remark, in a simple setting with $m=1$ and many calibration data points falling into the same bin as the test point, the two methods output nearly identical prediction sets. To see that, consider the simplest setting where all $X_1,\cdots X_n$ and $X_{n+1}$ fall into the same bin, 
and observe that each weight in weighted conformal prediction has the following lower and upper bounds---under propensity score $\ep$-discretization---both of which are close to $1/(n+1)$:
\[\frac{1/(1+\eps)}{1+\cdots+1+1/(1+\eps)} \le w_i = \frac{\frac{p_{A\mid X}(X_i)}{1-p_{A\mid X}(X_i)}}{\frac{p_{A\mid X}(X_1)}{1-p_{A\mid X}(X_1)} + \cdots \frac{p_{A\mid X}(X_n)}{1-p_{A\mid X}(X_n)} + \frac{p_{A\mid X}(X_{n+1})}{1-p_{A\mid X}(X_{n+1})}} \le \frac{1+\eps}{1+\cdots+1+(1+\eps)}.\]
Since pro-CP reduces to standard conformal prediction in this simple setting, which corresponds to weighted conformal prediction with all weights equal to $1/(n+1)$, this implies that the two methods nearly coincide 
when the sample size in each bin is large. 
However, in more general settings where we cannot expect large sample sizes within all bins, pooling data can be advantageous by effectively increasing the sample size used in inference on the outcomes.

\section{Inference with a stronger guarantee}
\label{sq}

Recall that our ultimate goal is to construct
prediction sets for each missing outcome,
ensuring we have a small miscoverage proportion
$\hat{p} = \frac{1}{m}\sum_{j=1}^m \One{Y_{n+j} \notin \widehat{C}(X_{n+j})}$. From the discussion in the main sections, we know how
to achieve the following guarantees for $\hat{p}$ if $p_{A \mid X}$ is known:
\begin{enumerate}
\item $\EE{\hat{p}} \le \alpha$, \qquad through weighted conformal prediction.
\item $\EEst{\hat{p}}{B_{(n+1):(n+m)}} \le \alpha$, \qquad through $\chp$ given by~\eqref{eqn:pro_CP}.
\end{enumerate}

These guarantees bound the expectation of 
$\hat{p}$, and are especially informative
if the sample size is large so that $\hat{p}$, the sample mean of the miscoverage indicators, concentrates tightly around its mean. 
However, for a moderate sample size 
where $\hat{p}$ can be highly variable, 
bounding the mean does not necessarily imply a 
precise control of $\hat{p}$.
%bounding the expectation of does not guarantee that we are likely to obtain a small $\hat{p}$. For example, the guarantee $\EE{\hat{p}} \le \alpha$ allows the case of $\hat{p}=2\alpha$ with probability $1/2$ and $\hat{p} = 0$ with probability $0$. 
Can we construct 
procedures with stronger guarantees, and if so with what guarantees?
We explore these questions in this section.

\subsection{Squared-coverage guarantee}
An ideal condition one could aim for is 
an almost sure bound on the miscoverage proportion:
\begin{equation}\label{eqn:guarantee_ideal}
    \hat{p} \le \alpha \textnormal{ almost surely}.
\end{equation}
However, this 
is unlikely to be achievable 
unless we make strong distributional assumptions. 
A natural relaxation one might consider is the following high-probability---or, Probably Approximately Correct (PAC)---guarantee inspired by the properties of tolerance regions \citep{Wilks1941,Wald1943} and inductive conformal prediction \citep{vovk2012conditional,Park2020PAC}:
\begin{equation}\label{eqn:guarantee_pac_1}
    \PP{\hat{p} \ge 1-\alpha} \ge 1-\delta,
\end{equation}
where $\alpha \in (0,1)$ and $\delta \in (0,1)$ are predefined levels. 
However, achieving this guarantee proves challenging 
due to several reasons, see Section~\ref{sec:pac}. 
Briefly,
the in-expectation guarantee~\eqref{eqn:in_exp_guarantee} is equivalent to covering one randomly drawn missing outcome, 
enabling the control of the total variation distance through propensity score $\ep$-discretization. 
In contrast, \eqref{eqn:guarantee_pac_1}
concerns 
the joint distribution of all missing outcomes,
implying that a larger sample size or number of missing outcomes may lead to a larger error.

As an alternative relaxation, we consider the following \emph{squared-coverage} guarantee. 
\begin{equation}\label{eqn:squared_cov_guarantee}
    \EE{\hat{p}^2} = \EE{\Bigg(\frac{1}{m}\sum_{j=1}^m \One{Y_{n+j} \notin \widehat{C}(X_{n+j})}\Bigg)^2} \le \alpha^2.
\end{equation}
This condition is motivated by the work of~\citet{lee2023distribution}, where the authors provide a discussion 
of possible targets of predictive inference. We provide below some possible interpretations of the squared-coverage guarantee.
\paragraph{Closer proxy of the ideal condition.} 
The ideal condition~\eqref{eqn:guarantee_ideal} 
of course implies the squared-coverage guarantee~\eqref{eqn:squared_cov_guarantee}, 
which in turn implies the marginal coverage guarantee~\eqref{eqn:in_exp_guarantee}. 
%To see that, note that the guarantee~\eqref{eqn:in_exp_guarantee} is equivalent to $\EE{\hat{p}} \le \alpha$ and recall that $\EE{X}^2 \le \EE{X^2}$ holds for any random variable $X$.
Therefore, the squared-coverage guarantee can be seen as a closer proxy of the almost-sure requirement~\eqref{eqn:guarantee_ideal}.
\par

\paragraph{Penalty on the spread of coverage proportion.} Alternatively, one can view the guarantee~\eqref{eqn:squared_cov_guarantee} as a condition that penalizes both the mean $\EE{\hat{p}}$ and the spread $\textnormal{var}\left[\hat{p}\right]$, in the sense that it can also be written as
$\EE{\hat{p}}^2 + \textnormal{var}\left[\hat{p}\right] \le \alpha^2$.
The ideal condition can be considered as a special case with $\EE{\hat{p}} \le \alpha$ and $\textnormal{var}\left[\hat{p}\right]=0$. However, in practice, achieving $\textnormal{var}\left[\hat{p}\right] = 0$ 
is hard, due to the randomness in the data. 
However, 
the squared coverage guarantee 
can be viewed as controlling the variance $\textnormal{var}\left[\hat{p}\right] > 0$ in addition to ensuring $\EE{\hat{p}}^2 \le \alpha^2$.
\par

\paragraph{Surrogate of PAC guarantee.}
More intuitively, the squared-coverage guarantee can be viewed as an approximation of
the PAC guarantee~\eqref{eqn:guarantee_pac_1}, by providing a smaller upper bound on the probability of obtaining a large miscoverage proportion compared to the target $\alpha$. 
For example, for any $\delta > 0$, the in-expectation guarantee $\EE{\hat{p}} \le \alpha$ provides the following tail bound for $\hat{p}$ via Markov's inequality:
$\PP{\hat{p} \ge \alpha+\delta} \le 
\EE{\hat{p}}/(\alpha+\delta) \le \alpha/(\alpha+\delta)$.
This implies the PAC-type guarantee at level $(\alpha+\delta,\alpha/(\alpha+\delta))$---however, $\alpha/(\alpha+\delta)$ might not be sufficiently small.
On the other hand, the stronger guarantee~\eqref{eqn:squared_cov_guarantee} provides a tighter bound
$\PP{\hat{p} \ge \alpha+\delta} \le 
\EE{\hat{p}^2}/(\alpha+\delta)^2 \le \left(\alpha/(\alpha+\delta)\right)^2$,
which implies the $(\alpha+\delta,(\alpha/(\alpha+\delta))^2)$-PAC-type guarantee, where now the failure probability $(\alpha/(\alpha+\delta))^2$ is smaller.
\par
\medskip
In the following sections, we introduce procedures for discrete features that achieves the squared-coverage guarantee, and then discuss discretization-based methods for general feature distributions.

\subsection{Inference with a squared-coverage guarantee for discrete features}

We first consider 
discrete feature distributions.
Following the notation from Section~\ref{sec:in_exp}, let $\{X_1',X_2',\ldots,X_M'\}$ be the set of distinct values among the observed features $(X_i)_{i \in [n+m]}$, and define $I_k, I_k^0, I_k^1$ and $N_k, N_k^0, N_k^1$ as before. Let $s$ be a score function, constructed independently of the data.  
For all $i\in [n]$, let $S_i = s(X_i,Y_i)$ 
and define
\begin{align*}
    \bar{S}_i = 
    \begin{cases}
        S_i &\text{ for } i \in [n],\\
        +\infty &\text{ for } n < i \le n+m.\\
    \end{cases}
\end{align*}
Then we define a prediction set for all $x\in\mathcal{X}$ as

\begin{multline}\label{eqn:CI_discrete_squared}
    \ch^2(x) = \left\{\rule{0cm}{1cm} y \in \Y : s(x,y) \le Q_{1-\alpha^2}\left(\rule{0cm}{1cm} \sum_{k=1}^M \sum_{i \in I_k} \frac{N_k^0}{m^2N_k}\cdot\delta_{\bar{S}_i}\right. \right. \\
    \left. \left. 
    + \sum_{k=1}^M\sum_{\substack{i,j\in I_k \\ i \neq j}} \frac{N_k^0(N_k^0-1)}{m^2N_k(N_k-1)}\delta_{\min\{\bar{S}_i,\bar{S}_j\}} 
    + \sum_{1 \le k \neq k' \le M}\sum_{i \in I_k} \sum_{j \in I_{k'}} \frac{N_k^0N_{k'}^0}{m^2N_kN_{k'}}\delta_{\min\{\bar{S}_i,\bar{S}_j\}} \rule{0cm}{1cm}\right) \rule{0cm}{1cm}\right\}.
\end{multline}
We define $N_k^0(N_k^0-1)/N_k(N_k-1)$ as zero if $N_k=1$. We prove the following.
\begin{theorem}\label{thm:MAR_squared}
Suppose that the random variables within each collection $(Y_i : i \in I_k)$, $k \in [M]$ are simultaneously exchangeable conditional on $X_{1:(n+m)}$. Then the prediction set $\ch^2$ from~\eqref{eqn:CI_discrete_squared} satisfies
\begin{equation}\label{eqn:conditional_sq_guarantee}
    \EEst{\Bigg(\frac{1}{m}\sum_{j=1}^m \One{Y_{n+j} \notin \ch^2(X_{n+j})}\Bigg)^2}{X_{1:(n+m)}} \le \alpha^2.
\end{equation}
\end{theorem}

The proof is given in the Appendix, but we briefly go over the idea here. The key observation is that the condition~\eqref{eqn:conditional_sq_guarantee} is equivalent to bounding the simultaneous miscoverage probability of two randomly and independently chosen missing outcomes by $\alpha^2$. 
By letting $j_1^*,j_2^*$ be two independent random draws from $\textnormal{Unif}([m])$, the target inequality~\eqref{eqn:conditional_sq_guarantee} is equivalent to
\[\PPst{Y_{n+j_1^*} \notin \ch^2(X_{n+j_1^*}), Y_{n+j_2^*} \notin \ch^2(X_{n+j_2^*})}{X_{1:(n+m)}} \le \alpha^2.\]
This observation is related to an intuitive interpretation of the distribution inside the quantile term of~\eqref{eqn:CI_discrete_squared}---the conditional distribution of $\min\{\bar{S}_{n+j_1^*}, \bar{S}_{n+j_2^*}\}$ given the set of scores $\{S_i : i \in [n]\}$.

Similarly to the prediction set~\eqref{eqn:CI_MAR_E}, 
the prediction set~\eqref{eqn:CI_discrete_squared} can be conservative if the proportion of the missing outcomes is high, and the distribution inside the quantile function has a large mass on $+\infty$. 
Specifically, the probability mass on $+\infty$ can be computed as
\[\frac{1}{m^2}\left[\sum_{k=1}^M \frac{N_k^0}{N_k}\cdot N_k^0 + \sum_{k=1}^M \frac{N_k^0(N_k^0-1)}{N_k(N_k-1)} \cdot N_k^0 (N_k^0-1)+\sum_{1 \le k \neq k' \le M} \frac{N_k^0 N_{k'}^0}{N_k N_{k'}}\cdot N_k^0 N_{k'}^0\right],\]
which can be approximated as
\[\frac{1}{m^2} \left(\tau \cdot m + \tau^2 \cdot \sum_{k=1}^M {N_k^0}^2 + \tau^2 \cdot \sum_{1 \le k \neq k' \le M} N_k^0 N_{k'}^0\right) = \tau^2 + \frac{\tau}{m},\]
where $\tau$ denotes the empirical missingness probability. To deal with the case where $\tau$ is large, we can apply a partitioning strategy, as before.

Let $U = \{U_1,\ldots,U_L\}$ be a partition of $[m]$, and let $N_{\ell}^0 = |U_{\ell}|$, for $\ell \in [L]$. For each $l$, let $\ch^{\ell}$ be the prediction set obtained by applying~\eqref{eqn:CI_discrete_squared} to the subset $(X_{n+j})_{j \in U_{\ell}}$  of the test
data, with level
\begin{equation}\label{eqn:alpha_l}
    \alpha_{\ell} = \frac{N_{\ell}^0 m}{\sum_{l'=1}^L (N_{l'}^0)^2} \cdot \alpha.
\end{equation}
Intuitively, $\alpha_{\ell}$ is proportional to $N_{\ell}^0$, which distributes the  error level proportionally to the number of missing values across partition elements. 
Let $\C_U^2$ denote this procedure, i.e., 
for all $j \in [m]$,
$\ch_U^2 = \C_U^2(\D)$ is given by $\ch_U^2(x,j) = \ch^{\ell_j}(x)$, where $\ell_j$ denotes the unique $\ell \in [L]$ such that $U_{\ell}$ contains $j$. We prove that this procedure satisfies the same guarantee.

\begin{corollary}\label{cor:sq_partition}
Under the assumptions of Theorem~\ref{thm:MAR_squared}, the procedure $\ch_U^2$ satisfies
\begin{equation*}
    \EEst{\Bigg(\frac{1}{m}\sum_{j=1}^m \One{Y_{n+j} \notin \ch_U^2(X_{n+j}, j)}\Bigg)^2}{X_{1:(n+m)}} \le \alpha^2.
\end{equation*}
\end{corollary}

For this procedure, the set of singletons $U^{(1)} = \{\{j\} : j \in [n]\}$
is not a desirable choice of $U$, 
since it is equivalent to constructing the prediction sets~\eqref{eqn:discrete_split_conformal} at level $\alpha^2$, 
likely leading to overly conservative prediction sets. 
To see that, suppose $m=1$ and the unique missing outcome occurs in $B_k$, i.e., $X_{n+1} = X_k'$. Then
for all $x$, the $\ch^2$ in~\eqref{eqn:CI_discrete_squared} can be simplified to 

\begin{align*}
    \ch^2(x) &= \left\{y \in \Y : s(x,y) \le Q_{1-\alpha^2}\left(\sum_{i \in I_k} \frac{1}{N_k} \cdot \delta_{\bar{S}_i}\right)\right\}\\
    &= \left\{y \in \Y : s(x,y) \le Q_{1-\alpha^2}\left(\sum_{i \in I_k \backslash \{n+1\}} \frac{1}{N_k} \cdot \delta_{S_i} + \frac{1}{N_k}\delta_{\infty}\right)\right\},
\end{align*}
and thus we simply obtain a split conformal prediction set at level $\alpha^2$. Thus, even if we have a large sample size so that each $N_k$ is sufficiently large, $\ch_{U^{(1)}}^2$ is likely to be conservative.

As before, a reasonable choice of $U$ could 
minimize the partition size $|U|$, 
while keeping the ratio $|U_{\ell}|/n$ small for each $\ell \in [L]$. For example, if the overall proportion of missing outcomes is $0.2$ and we aim to have probability mass on $+\infty$ less than $0.01$, 
one can choose to have a partition of size $20$, with
nearly equal-size partitions. 

\subsection{Inference for general feature distributions via propensity discretization}

Now we consider the general case where $X$ can be continuous. If $p_{A \mid X}$ is known, we can apply a strategy similar to \eqref{eqn:pro_CP}. 
Construct the partition $\B$ as in~\eqref{eqn:eps_partition}.
%, and define $b:\X \rightarrow \Z$ as~\eqref{eqn:b_fn}.
We use notations such as $I_k^\B, I_k^{\B,0}, I_k^{\B,1}, N_k^\B, N_k^{\B,0}, N_k^{\B,1}$ as defined in Section~\ref{sec:conformal_discretization}. 
Then, for all $x\in\mathcal{X}$, 
we construct the prediction set, 
with 
$\bar{S}_{ij} := \min\{\bar{S}_i,\bar{S}_j\}$ for all $i,j$,
\begin{multline}\label{eqn:CI_discrete_squared_p_A_X}
    \chps(x) = \left\{\rule{0cm}{1cm} y \in \Y : s(x,y) \le Q_{1-\alpha^2}\left(\rule{0cm}{1cm} \sum_{k=1}^M \sum_{i \in I_k^\B} \frac{1}{m^2}\cdot \frac{N_k^{\B,0}}{N_k^\B}\cdot\delta_{\bar{S}_i}\right. \right. \\
    \left. \left.
    + \sum_{k=1}^M\sum_{\substack{i,j\in I_k^\B \\ i \neq j}} \frac{N_k^{\B,0}(N_k^{\B,0}-1)}{m^2N_k^\B(N_k^\B-1)}\delta_{\bar{S}_{ij} } 
    + \sum_{1 \le k \neq k' \le M}\sum_{i \in I_k^\B} \sum_{j \in I_{k'}^\B} \frac{N_k^{\B,0}N_{k'}^{\B,0}}{m^2N_k^\B N_{k'}^\B}\delta_{\bar{S}_{ij} } \rule{0cm}{1cm}\right) \rule{0cm}{1cm}\right\},
\end{multline}
which is obtained by applying $\ch$ from~\eqref{eqn:CI_discrete_squared} to the discretized data $(B_i,Y_i)_{i \in [n]}$ and $(B_{n+j})_{j \in [m]}$.

\begin{theorem}\label{thm:squared_p_A_X}
    Suppose $0 < p_{A \mid X}(x) < 1$ for any $x \in \X$. Then $\chps$ from~\eqref{eqn:CI_discrete_squared_p_A_X} satisfies
    \begin{equation*}
    \EEst{\Bigg(\frac{1}{m}\sum_{j=1}^m \One{Y_{n+j} \notin \chps(X_{n+j})}\Bigg)^2}{B_{1:(n+m)}} \le \alpha^2+2\eps.
    \end{equation*}
\end{theorem}

Again, we can apply the partitioning strategy to obtain narrower prediction sets. Specifically, given a partition $U = \{U_1,U_2,\ldots,U_L\}$ of $[m]$, let $\ch^{\ell}$ be the prediction set we obtain by applying $\chps$ to the subset $(X_{n+j})_{j \in U_{\ell}}$ of the test data at level $\alpha_{\ell}$ given by~\eqref{eqn:alpha_l}, and then define
\begin{equation}\label{eqn:CI_general_squared_partition}
    \chps_U(x,j) = \ch^{\ell_j}(x),
\end{equation}
where $\ell_j$ denotes the unique $\ell \in [L]$ such that $U_{\ell}$ contains $j$. Then by the same logic as Corollary~\ref{cor:sq_partition}, with the result of Theorem~\ref{thm:squared_p_A_X}, we have the following.

\begin{corollary}\label{cor:sq_partition_general}
The prediction set $\chps_U$ from~\eqref{eqn:CI_general_squared_partition} satisfies
\begin{equation*}
    \EEst{\Bigg(\frac{1}{m}\sum_{j=1}^m \One{Y_{n+j} \notin \chps_U(X_{n+j}, j)}\Bigg)^2}{B_{1:(n+m)}} \le \alpha^2+2\eps.
\end{equation*}
\end{corollary}

\subsubsection{Approximate inference via estimation of missingness probability}

Next, consider the setting where we do not have access to $p_{A \mid X}$, and instead have an estimate $\hat{p}_{A \mid X}$, and let $\ctps$ be the procedure obtained by constructing $\mathcal{B}$ based on this estimate. 
Applying arguments similar  to the proof of Theorem~\ref{thm:MAR_estimated_p_A_X}, we can prove the following approximate guarantee for the procedure $\ctps$.

\begin{theorem}\label{thm:squared_p_A_X_estimated}
Suppose $0 < p_{A \mid X}(x) < 1$ and $0 < \hat{p}_{A \mid X} < 1$ hold for all $x \in \X$. Define $f_{p,\hat{p}}$ as~\eqref{eqn:f_p_phat}. Then $\ctps$ satisfies
    \begin{equation*}
    \EEst{\Bigg(\frac{1}{m}\sum_{j=1}^m \One{Y_i \notin \ctps(X_i)}\Bigg)^2}{B_{1:(n+m)}} \le \alpha^2+2( \eps + \delta_{\hat{p}_{A\mid X}} + \eps \cdot \delta_{\hat{p}_{A\mid X}}),
    \end{equation*}
    where
    $\delta_{\hat{p}_{A \mid X}} = e^{2\|\log f_{p,\hat{p}}\|_\infty} - 1$.
\end{theorem}
As before, we have the same upper bound for the partition-based procedure $\ctps_U$. We omit this to avoid repetition.

\subsection{Notes on other potential target guarantees}

So far, we have investigated the squared-coverage guarantee as a stronger miscoverage proportion-controller. 
%re there other guarantees we can achieve via conformal-type procedures? For example, what if we aim for a higher-order coverage guarantee, so that we have a even stronger requirement for the miscoverage proportion? 
In this section, we explore other possible targets of inference. 
% and 
% limitations of those options.

\subsubsection{Higher-order coverage guarantees}

Suppose we aim for the following $K$-th order coverage guarantee:
\[\EE{\Bigg(\frac{1}{m}\sum_{j=1}^m \One{Y_{n+j} \notin \widehat{C}(X_{n+j})}\Bigg)^K} \le \alpha^K,\]
where $K \ge 3$. 
A larger $K$ means a stronger requirement, in the sense that the $K'$-th order coverage guarantee implies the $K$-th order guarantee if $K' > K$. 
This guarantee is also achievable,
but may require extremely wide prediction sets. 

Recall that the squared coverage guarantee can be achieved by looking at the simultaneous miscoverage of two randomly chosen missing outcomes. 
Similarly, the $K$-th order coverage guarantee can be obtained by investigating the simultaneous miscoverage of $K$ randomly chosen missing outcomes, and it turns out that the resulting prediction set has the form
$\ch(x) = \left\{y \in \Y : s(x,y) \le Q_{1-\alpha^K}(P_K)\right\}$,
where the distribution $P_K$ inside the quantile is supported on the set of observed scores and $+\infty$, i.e., $\{\bar{S}_i : i \in [n+m]\}$, similarly to $\ch^2$.
However, this procedure is unlikely to provide informative prediction sets in practice, as it involves the $(1-\alpha^K)$-quantile of a distribution whose support size is less than $n$. 
For $\alpha = 0.05$ and $K=3$, this already requires $n\ge 8000$ to be non-trivial, and the requirement grows exponentially with $K$,
Thus, we focus on second-order coverage guarantees in this work.

\subsubsection{PAC-type guarantee}\label{sec:pac}
One might be interested in a guarantee of the form
\eqref{eqn:guarantee_pac_1}, which asks the prediction set to cover at least $\lceil m (1-\alpha) \rceil$ missing outcomes with sufficient probability. 
This PAC-type guarantee provides a clean interpretation of the procedure, but it turns out quite challenging to 
handle. 
Indeed, it requires dealing with the set of all missing outcomes, instead of the coverage for one or two randomly chosen missing outcomes. 

A natural approach to achieve the guarantee~\eqref{eqn:guarantee_pac_1} is to consider the distribution of the $\lceil m (1-\alpha) \rceil$-th smallest element among the set of scores with missing outcomes. Specifically, we can construct the prediction set for all $x$ as

\[\ch(x) = \left\{y \in \Y : s(x,y) \le Q_{1-\delta}\left(\sum_{\substack{J = J_1 \cup \ldots J_M \\ J_k \subset I_k, |J_k| = N_k^0}} \frac{1}{\prod_{k=1}^M \binom{N_k}{N_k^0}} \cdot \delta_{\bar{S}_{(\lceil(1-\alpha)m\rceil)}^J}\right)\right\},\]
and show that this procedure satisfies the guarantee~\eqref{eqn:guarantee_pac_1}, 
through a standard exchangeability argument. 
However, we do not have an obvious adjusted procedure for the case of high missingness probability in this case.
We cannot apply the previous partitioning strategy, since the coverage probability does not have the linearity that expectation has.

\section{Discussion of partitioning the test datapoints}
\label{testpart}

Recall that the mass at $+\infty$ 
in the prediction set from \eqref{eqn:CI_MAR_E}
is
\(
\frac{1}{m} \sum_{k=1}^M \frac{(N_k^0)^2}{N_k},
\)
where $m$ is the number of test points, and $N_k^0, N_k^1$ denote the number of datapoints with a given feature value $X_k'$ in the test and training sets, respectively, and $N_k = N_k^0 + N_k^1$.
To lighten notation, write $u_k = N_k^0$, $o_k = N_k^1$, so that the above mass at $+\infty$ becomes
\begin{equation}\label{F}
F(u_1, \ldots, u_K) = \frac{1}{\sum_{k=1}^M u_k} \left( \sum_{k=1}^M \frac{u_k^2}{u_k + o_k} \right).
\end{equation}

A partition $U = \{U_1, \ldots, U_L\}$ corresponds to splitting up the unobserved datapoints.
Let $\vec{u} = (u_1, \ldots, u_M) \in \mathbb{N}^K$, and let
\(
\vec{u}^{(1)}, \ldots, \vec{u}^{(L)} \in \mathbb{N}^K
\)
such that $u_k^{(\ell)} = U_\ell \cap I_k^0$ denotes the number of test datapoints with a given feature $X_k'$ in the $\ell$-th partition element $U_\ell$.
Then,
\(
\vec{u} = \sum_{\ell=1}^L \vec{u}^{(\ell)},
\)
and the mass at $+\infty$ in the $\ell$-th prediction set $\ch^\ell$ is
\(
F(\vec{u}^{(\ell)}).
\)

To determine how to evaluate a partition,
we must decide how to measure the impact of the mass at $+\infty$ on a prediction set.
In the above partitioning method, a mass of $F(\vec{u}^{(\ell)})$ is placed at $+\infty$ for all $\sum_{k=1}^M u_k^{(\ell)}$ test datapoints used in the construction of $\ch^\ell$.
In particular, if $\alpha \le F(\vec{u}^{(\ell)})$, then the prediction sets for these datapoints cover all possible values of the outcome, and are thus uninformative.

A reasonable goal could be to minimize the number of test datapoints with uninformative prediction sets.
This can thus be formulated as the following integer optimization problem, where we denote by $1_M$ the all-ones vector of size $M$:
\begin{equation}\label{co}
\begin{aligned}
&\min_{L,\, \vec{u}^{(1)}, \ldots, \vec{u}^{(L)}} && 1_M^\top \left( \sum_{\ell \in V_\alpha} \vec{u}^{(\ell)} \right) \\
&\text{subject to} \quad
& & L \in \mathbb{N},
\quad 
\vec{u}^{(1)}, \ldots, \vec{u}^{(L)} \in \mathbb{N}^M, \quad
\sum_{\ell=1}^L \vec{u}^{(\ell)} = \vec{u},
\quad\mathrm{ where }\quad
V_\alpha = \left\{ \ell \in [L] : F(\vec{u}^{(\ell)}) \ge \alpha \right\},
\end{aligned}
\end{equation}
and where $F$ is defined above in \eqref{F}.
There are a large number of fast algorithms and corresponding implementations that can be used to approximate the solution; and any approximation could lead to useful gains in statistical performance.

As a somewhat simpler goal, we could aim to maximize the level $\alpha$ for which there are no uninformative prediction sets.
This can be formulated as
\begin{align}\label{mm}
&\min_{L,\, \vec{u}^{(1)}, \ldots, \vec{u}^{(L)}}
\max_{\ell \in [L]} F(\vec{u}^{(\ell)}), 
\end{align}
subject to the constraints from
\eqref{co}.

To gain some insight into this problem, we consider some relaxed and simplified cases.
First, we can obtain an upper bound on the objective by relaxing the constraints $\vec{u}^{(\ell)} \in \mathbb{N}^M$ to $\vec{u}^{(\ell)} \in [0,\infty)^M$.
Next, consider a special case where all entries of $\vec{u}$ are equal; i.e., where there are an equal number of test datapoints for each feature value. Call this value $\rho$.
Then, we observe that for any vector $x \in [0,\infty)^M$,
we can reduce the value of the missingness by replacing $x$ with a vector whose entries are all equal to the mean value $\bar x$ of the coordinates in $x$:
\[
F(\bar x \cdot 1_M) \le F(x), \quad \forall x \in [0,\infty)^M.
\]
This follows from the convexity of the map
\[
x \mapsto H(x):=\sum_{k=1}^M \frac{x_k^2}{x_k + o_k} \quad \text{on } [0,\infty)^M,
\]
since $F(x) = H(x)/1^\top x$.
The convexity of $x$
in turn follows from the convexity of
\(
z \mapsto h(x) = \frac{z^2}{z + o}\)
on $[0,\infty)$, for $o\ge0$.
The latter can be verified because
\(
h''(z) = \left(\frac{z^2}{z + o}\right)'' = \frac{2o^2}{(z + o)^3} > 0.
\)

Therefore, given any $\vec{u}^{(\ell)}$, we can replace it with
\(
v^{(\ell)} = \left( \frac{1}{M} \sum_{k=1}^M u_k^{(\ell)} \right) \cdot 1_M,
\)
and obtain $F(v^{(\ell)}) \le F(\vec{u}^{(\ell)})$.
Now observe that
\[
\sum_{\ell=1}^L v^{(\ell)} 
= \frac{ 1_M}{M} \sum_{\ell=1}^L \sum_{k=1}^M u_k^{(\ell)} 
= \frac{ 1_M}{M} \sum_{k=1}^M \sum_{\ell=1}^L u_k^{(\ell)},
\]
and
\(
\sum_{\ell=1}^L u_k^{(\ell)} = u_k = \rho.
\)
Hence,
\(
\sum_{\ell=1}^L v^{(\ell)} 
= \frac{ 1_M}{M} \sum_{k=1}^M u_k 
= \rho \cdot 1_M = \vec{u},
\)
and thus $(v^{(1)}, \ldots, v^{(L)})$ satisfies the constraints of our optimization problem.

Thus, denoting
\(
w_\ell = \frac{1}{M} \sum_{k=1}^M u_k^{(\ell)},
\)
the problem reduces to
\[
\begin{aligned}
&\min_{L,\, w_1, \ldots, w_L \ge 0} && \max \left\{ G(w_\ell) : \ell = 1, \ldots, L \right\} \quad
\text{subject to} & \sum_{\ell=1}^L w_\ell = \rho,
\end{aligned}
\]
where
\[
G(w) = \frac{1}{M w} \left( \sum_{k=1}^M \frac{w^2}{w + o_k} \right) = \frac{1}{M} \sum_{k=1}^M \frac{w}{w + o_k}.
\]

Now, since $G$ is increasing,
\[
\max_{\ell \in [L]} G(w_\ell) = G\left( \max_{\ell \in [L]} w_\ell \right) \ge G\left( \frac{\rho}{L} \right),
\]
achieved when $w_\ell = \frac{\rho}{L}$ for all $\ell \in [L]$.

This shows that, for a fixed $L$, the optimum is
\[
G_L := \frac{1}{M} \sum_{k=1}^M \frac{\rho/L}{\rho/L + o_k} = \frac{1}{M} \sum_{k=1}^M \frac{\rho}{\rho + L o_k}.
\]
Moreover, since $G_L$ is decreasing in $L$, the optimum is achieved by taking $L$ as large as possible.

Returning to our original problem, this shows that we should split up the test datapoints corresponding to different features as evenly as possible.
Given $\rho$ datapoints for each feature value, the number of datapoints for a split into $L$ partition elements is $\rho / L$, if $L$ divides $\rho$.
The above analysis shows that we should take $L$ as large as possible, i.e., $L = \rho$, so that each bin contains exactly one datapoint for each feature value.
Moreover, in this case, the minimum achievable $\alpha$ is 
\(
\alpha^*= G_\rho = \frac{1}{M} \sum_{k=1}^M \frac{1}{1 + o_k}.
\)
Therefore we have proved the following result:
\begin{proposition}[Optimal partition achieving minimal mass at $+\infty$ for equal test counts]
Suppose that the number of test datapoints is the same for each feature value.
Then, among all partitions of the test datapoints into disjoint groups, 
the minimum achievable worst-case value of the feature-wise worst-case mass at $+\infty$ 
in the prediction set from \eqref{eqn:CI_MAR_E}---i.e., the optimal value of \eqref{mm}---is
\[
\alpha^* := \frac{1}{M} \sum_{k=1}^M \frac{1}{1 + N_k^1}.
\]
This is attained when each partition element consists of exactly one test datapoint for every feature value.
\end{proposition}

Similar qualitative insights extend to a more general, uneven distribution of test datapoints across features; however, finding the optimal solution seems not as straightforward analytically, and to require numerical computation.

\section{Additional experimental results}
\label{addexp}

\subsection{Illustration of pro-CP2}\label{sec:sim_pro_CP_1_2}

Next, we illustrate the coverage of the procedures $\chp$, $\chps$, and also $\ctp$, $\ctps$,
using an estimate $\hat{p}_{A \mid X}$ of $p_{A \mid X}$.
The theoretical results for these procedures control the bin-conditional coverage rate. 
In this simulation, we sample the data  $(X_i,A_i,Y_iA_i)_{1 \le i \le n}$ multiple times from the marginal distribution $P_X \times P_{A \mid X} \times P_{Y \mid X}$, since the goal of this experiment is to illustrate how the procedures perform in terms of controlling the coverage proportion in various trials.
%, rather than confirming the theoretical results empirically.

We first generate training data, and
using 
the normal kernel $K_h(x,y) = \exp(-\frac{1}{2}((y-x)/h)^2)$ for all $x,y$,
apply a kernel regression
\[\hat{p}_{A \mid X}(x) = \frac{\sum_{i=1}^{n_\text{train}} K_h(X_i',x) A_i'}{\sum_{i=1}^{n_\text{train}} K_h(X_i',x)},\]
to construct an estimate $\hat{p}_{A \mid X}$
for all $x$. We select the bandwidth $h$ by applying the method of~\citet{ruppert1995effective}.
We also construct an estimate $\hat{\mu}$ by linear regression, 
and the construct the score function $s(x,y) = |y - \hat{\mu}(x)|$ for all $x,y$.

Then we generate data $(X_i,A_i,Y_iA_i)_{1 \le i \le n}$ with $n=500$, and apply four procedures $\chp_U$, $\chps_U$, $\chp_U$, $\ctps_U$, with $U$ in~\eqref{eqn:U}, under $\alpha=0.2$ and $\eps = 0.1$. 
From now on, we 
denote the four procedures as $\chp$, $\chps$, $\ctp$, $\ctps$ without the subscript $U$ for simplicity. 
For each 
procedure, 
we compute the coverage proportion and the median width of the prediction sets. 
We repeat these steps 500 times, 
and summarize the results in Table~\ref{table:sim}, Figures~\ref{fig:known_ps} and~\ref{fig:unknown_ps}.

\begin{table}%[H]

\begin{center} 
\begin{tabular}{llll}
\hline
& & $\PP{\text{coverage} \ge 1-\alpha}$ & $\EE{\text{median width}}$ \\
\hline
\multirow{4}{*}{Setting 1}&$\chp$  & 0.7560 (0.0192)  & 24.61 (0.0856)\\
&$\chps$  & 0.9920 (0.0040) & 29.09 (0.1072)\\
&$\ctp$   & 0.6880 (0.0207) & 23.95 (0.0870)\\
&$\ctps$  & 0.9860 (0.0053) & 28.39 (0.1067)\\\hline
\multirow{4}{*}{Setting 2}&$\chp$ & 0.9060 (0.0131) & 23.86 (0.0935)\\
&$\chps$  & 0.9980 (0.0020) & 30.24 (0.1307)\\
&$\ctp$   & 0.9160 (0.0124) & 23.75 (0.0816)\\
&$\ctps$  & 1.000 (0.0000)&  29.43 (0.1078)\\\hline
\end{tabular} 
\end{center}
\caption{Results for Settings 1 and 2: The probability of the coverage proportion being larger than $1-\alpha$, and the mean of the \emph{median prediction interval width} of the prediction sets $\chp$, $\chps$, $\ctp$, and $\ctps$, with standard errors.}
\label{table:sim}
\end{table}

\begin{figure}[H]
  \centering
  \includegraphics[width=0.9\textwidth]{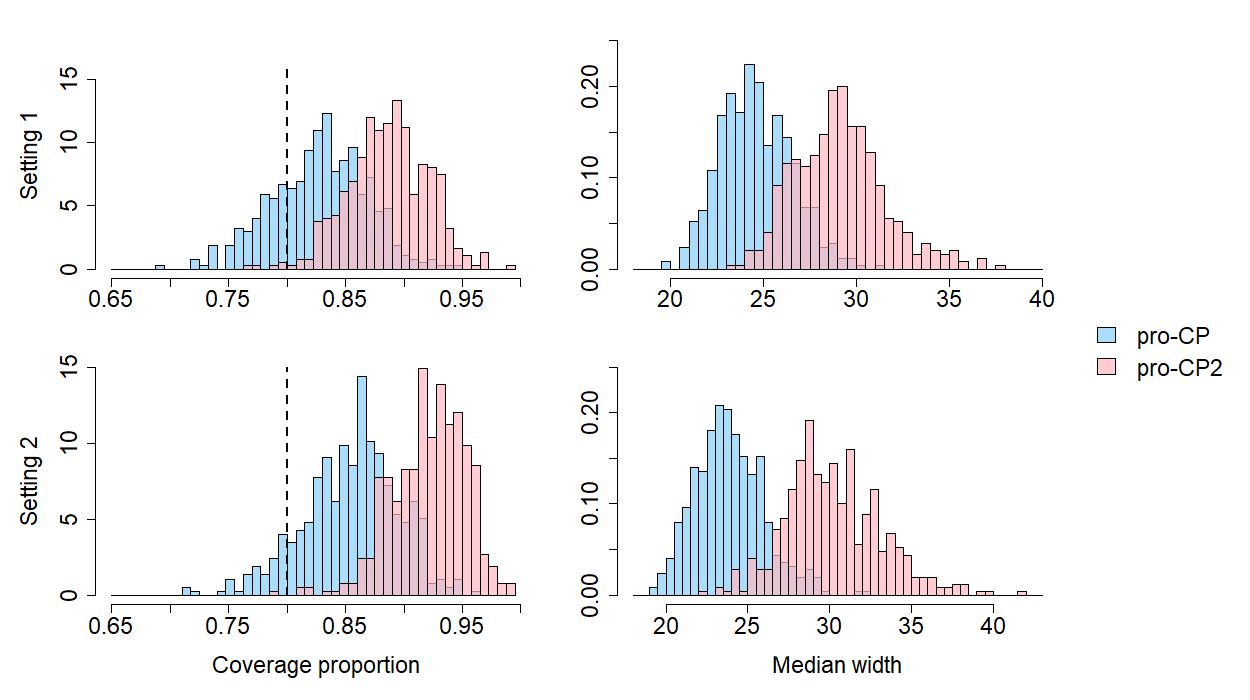}
  \caption{Results for the setting of a known propensity score: histograms of coverage proportion and median width of $\chp$, $\chps$ under 500 independent trials, in Settings 1 and 2.}
  \label{fig:known_ps}
\end{figure}

\begin{figure}[H]
  \centering
  \includegraphics[width=0.9\textwidth]{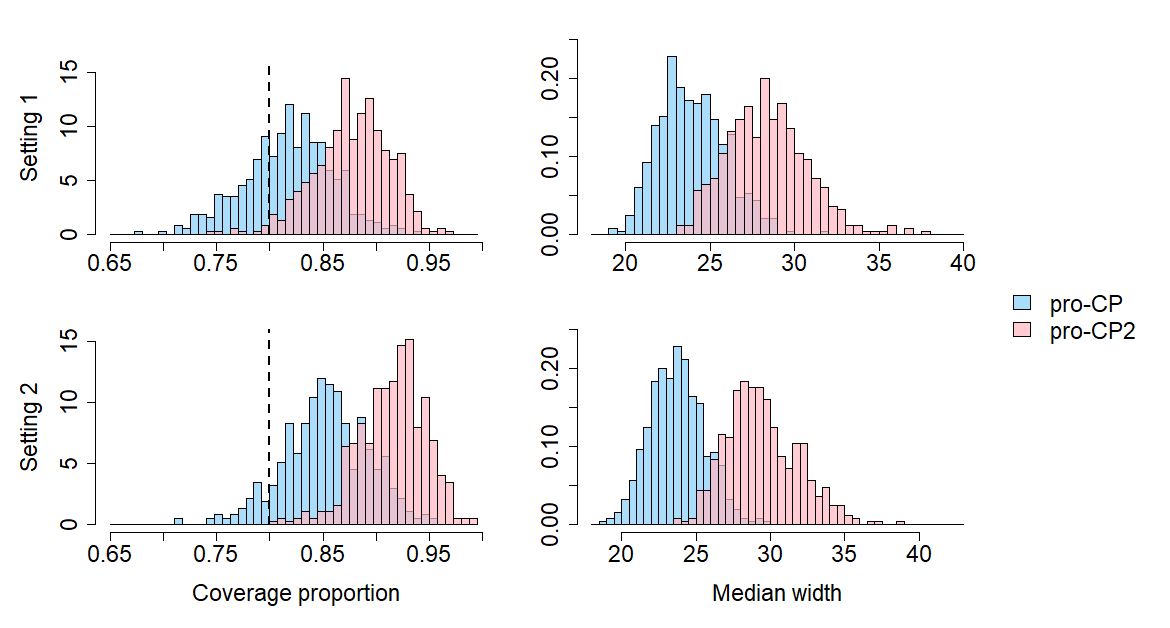}
  \caption{Results for an unknown propensity score: histograms of coverage proportion and median width of $\ctp$, $\ctps$ for 500 independent trials, in Settings 1 and 2.}
  \label{fig:unknown_ps}
\end{figure}

Table~\ref{table:sim} shows estimates of the following values for the four procedures:
\begin{align*}
    &1.\; \text{Probability of coverage proportion at least $ 1-\alpha$: }  \PP{\frac{1}{m} \sum_{ i\in I_{A=0}} \One{Y_i \in \ch(X_i)} \ge 1-\alpha},\\
    &2.\; \text{Expected median interval width: }  \EE{\textnormal{median}\Big(\Big\{\textnormal{leb}(\ch(X_i)) : A_i = 0\Big\}\Big)},
\end{align*}
and Figures~\ref{fig:known_ps} and~\ref{fig:unknown_ps} show the histograms of the coverage proportion and the median width in the two settings.

The results show that pro-CP2 with the squared coverage guarantee tends to provide wider prediction sets, 
to ensure
coverage proportion larger than $1-\alpha$ in most trials. 
This illustrates how the squared coverage guarantee works as an approximation of the PAC-type guarantee~\eqref{eqn:guarantee_pac_1}. The experiments with $\ctp$ and $\ctp$, the procedures based on an estimate $\hat{p}_{A\mid X}$ of the missingness probability, 
show similar results to those with $\chp$ and $\chps$, respectively.

\subsection{Analysis of the conservativeness of pro-CP2}

As the previous experiments illustrate, pro-CP2 provides a stronger guarantee by constructing wider prediction sets. 
A question is whether the procedure increases the width only as needed (compared to pro-CP with the in-expectation guarantee). 
To examine this question, we run pro-CP and pro-CP2 for different values of the level $\alpha$, and compare their widths and coverage rates. Figure~\ref{fig:curve} shows the median width-coverage rate curve of the two procedures in Settings 1 and 2, where the values are the averages over 500 independent trials. 

\begin{figure}[H]
  \centering
  \includegraphics[width=0.9\textwidth]{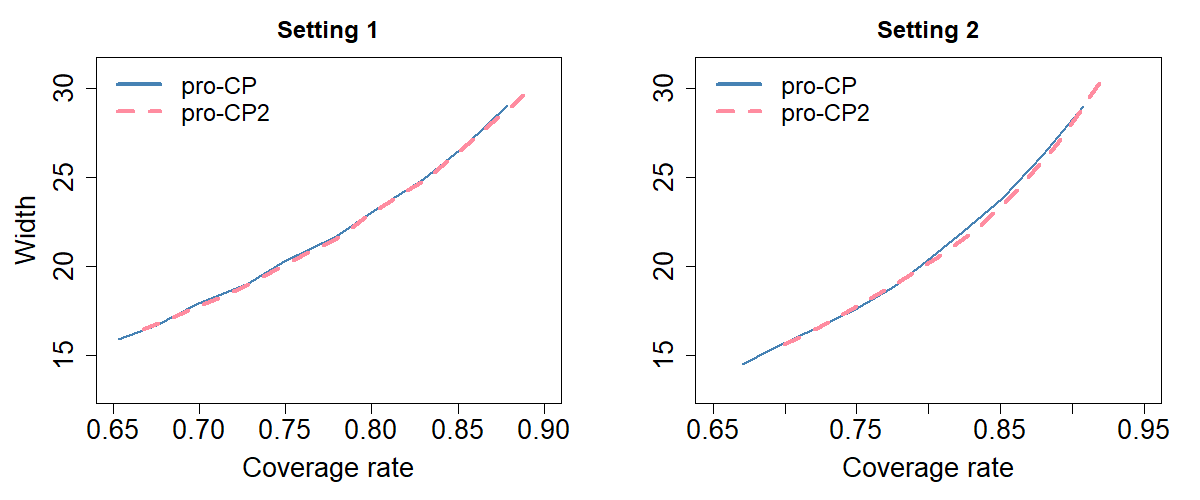}
  \caption{Median width-coverage rate curves of pro-CP and pro-CP2, in Settings 1 and 2.}
  \label{fig:curve}
\end{figure}

The result shows that pro-CP2 provides a similar prediction set to pro-CP, 
for the same coverage rate. In other words, for pro-CP2 run at level $\alpha$, there exists $\alpha'$ that pro-CP run at level $\alpha'$ shows a similar performance (on average). Recall the second interpretation of the squared-coverage guarantee, which suggests that it shifts the expected coverage rate to account for the spread of coverage. The above result demonstrates that pro-CP2 indeed functions like an adjusted pro-CP, without unnecessary widening of the prediction sets.

\section{Additional proofs}

\subsection{Proof of Theorem~\ref{thm:MAR_in_exp}}
We leverage ideas from the proof of the validity of hierarchical conformal prediction \citep[Theorem 1]{lee2023distribution}, see also \cite{dobriban2023symmpi}. Let $j^*$ be a random variable drawn via $j^* \sim \textnormal{Unif}([m])$, independently of the data. 
Then it is enough to prove
\begin{equation}\label{eqn:E_to_P}
    \PPst{Y_{n+j^*} \in \ch(X_{n+j^*})}{X_{1:(n+m)}} \ge 1-\alpha,
\end{equation}
since, as in \eqref{isu},
\begin{align*}
    &\PPst{Y_{n+j^*} \in \ch(X_{n+j^*})}{X_{1:(n+m)}} = \EEst{\One{Y_{n+j^*} \in \ch(X_{n+j^*})}}{X_{1:(n+m)}}\\
    & =\EEst{\EEst{\One{Y_{n+j^*} \in \ch(X_{n+j^*})}}{j^*,\ch, X_{1:(n+m)}}}{X_{1:(n+m)}}
    =  \EEst{\frac{1}{m}\sum_{j=1}^m\One{Y_{n+j} \in \ch(X_{n+j})}}{X_{1:(n+m)}}.
\end{align*}
Define $q_{1-\alpha} : \R^{N_1}\times \R^{N_2} \times \ldots \times \R^{N_M} \rightarrow \R$, 
such that 
for all 
$\tilde{s}_k = (\tilde{s}_{k1},\tilde{s}_{k2}, \ldots, \tilde{s}_{k N_k})^\top \in \R^{N_k}$ for each $k$,
we have
\[q_{1-\alpha}(\tilde{s}_1,\tilde{s}_2,\ldots,\tilde{s}_M) = Q_{1-\alpha}\left(\sum_{k=1}^M \sum_{j=1}^{N_k} \frac{1}{m} \frac{N_k^0}{N_k}\cdot \delta_{\tilde{s}_{kj}}\right).\]
Defining
$\tilde{S}_k = (S_i)_{i \in I_k}$, then it holds that
\[q_{1-\alpha}(\tilde{S}_1, \ldots, \tilde{S}_M) = Q_{1-\alpha}\left(\sum_{k=1}^M \sum_{i \in I_k} \frac{1}{m}\frac{N_k^0}{N_k}\cdot\delta_{S_i}\right).\]
Then by the definition of a quantile, it holds deterministically that
\begin{equation}\label{eqn:weighted_quantile_2}
    \sum_{k=1}^M \sum_{i \in I_k} \frac{1}{m}\frac{N_k^0}{N_k}\cdot \One{S_i \le q_{1-\alpha}(\tilde{S}_1, \ldots, \tilde{S}_M)} \ge 1-\alpha.
\end{equation}
Next, by definition of $q_{1-\alpha}$, we have
\[q_{1-\alpha}(\tilde{s}_1,\ldots,\tilde{s}_M) = q_{1-\alpha}(\tilde{s}_1^{\sigma_1},\ldots,\tilde{s}_M^{\sigma_M}),\]
for any permutations $\sigma_1 \in \mathcal{S}_{N_1}, \ldots, \sigma_M \in \mathcal{S}_{N_M}$, where $\tilde{s}_k^{\sigma_k}$ denotes the $\sigma_k$-permuted vector  $\tilde{s}_k$, for each $k$.
We also have
\[(\tilde{S}_1, \tilde{S}_2, \ldots, \tilde{S}_M) \stackrel{\text{d}}{=} (\tilde{S}_1^{\sigma_1}, \tilde{S}_2^{\sigma_2}, \ldots, \tilde{S}_M^{\sigma_M}) \mid X_{1:(n+m)},\]
by the exchangeability of $(\tilde{S}_1, \tilde{S}_2, \ldots, \tilde{S}_M)$ conditional on $X_{1:(n+m)}$. 
Let 
$\sigma \in \mathcal{S}_n$ be a permutation determined by $\sigma_1,\ldots,\sigma_M$, defined as follows:
For each $k$, with $I_k = \{i_{k1}, i_{k2}, \ldots, i_{kN_k}\}$,  $\sigma$ satisfies $\sigma(i_{kj}) = i_{k \sigma_k(j)}$ for any $1 \le k \le M$ and $1 \le j \le N_k$. 
Let $S^\sigma$ denote the vector $S_{\sigma(1)}, \ldots, S_{\sigma(n)}$---in other words, the components of $S^\sigma$ are given by $((S^\sigma)_i)_{i \in I_k} = \tilde{S}_k^{\sigma_k}$ for each $k\in [M]$. Let $\mathcal{S}_{I_{1:M}}$ be the set of such permutations, i.e.,
\begin{equation}\label{eqn:S_1_M}
    \mathcal{S}_{I_{1:M}} = \{\sigma \in \mathcal{S}_n : \sigma(i_{kj}) = i_{k \sigma_k(j)} \,\forall k \in [M], j \in [N_k], \sigma_1 \in \mathcal{S}_1, \ldots, \sigma_M \in \mathcal{S}_M\}.
\end{equation}
Conditional on $X_{1:(n+m)}$, 
\begin{align*}
    &\One{S_{n+j^*} \le q_{1-\alpha}(\tilde{S}_1, \tilde{S}_2, \ldots, \tilde{S}_M)} \stackrel{\text{d}}{=} \One{S_{n+j^*}^\sigma \le q_{1-\alpha}(\tilde{S}_1^{\sigma_1}, \tilde{S}_2^{\sigma_2}, \ldots, \tilde{S}_M^{\sigma_M})}\\
    & =\One{S_{n+j^*}^\sigma \le q_{1-\alpha}(\tilde{S}_1, \tilde{S}_2, \ldots, \tilde{S}_M)}
    =\One{S_{\sigma(n+j^*)} \le q_{1-\alpha}(\tilde{S}_1, \tilde{S}_2, \ldots, \tilde{S}_M)}.
\end{align*}
Then it follows that
\begin{align*}
&\PPst{S_{n+j^*} \le q_{1-\alpha}(\tilde{S}_1, \ldots, \tilde{S}_M)}{X_{1:(n+m)}}
= \EEst{\One{S_{n+j^*} \le q_{1-\alpha}(\tilde{S}_1, \ldots, \tilde{S}_M)}}{X_{1:(n+m)}}\\
& =\EEst{\frac{1}{|\mathcal{S}_{I_{1:M}}|}\sum_{\sigma \in \mathcal{S}_{I_{1:M}}}\One{S_{\sigma(n+j^*)} \le q_{1-\alpha}(\tilde{S}_1, \ldots, \tilde{S}_M)}}{X_{1:(n+m)}}.
\end{align*}
Now we introduce, for each $k\in[M]$,
the dummy indices 
$i \in I_k$, $j' \in I_k^0$ 
below, 
where $j'$ represents the possible values taken by the random variable $j^*$, 
and $i$ represents the possible values taken by the random variable $\sigma(n+j^*)$---or $\sigma(n+j')$.
We can see that the above quantity equals
\begin{align*}
&
\EEst{\frac{1}{|\mathcal{S}_{I_{1:M}}|}\sum_{\sigma \in \mathcal{S}_{I_{1:M}}} 
\sum_{k=1}^M \sum_{i \in I_k}\sum_{j' \in I_k^0} \One{j^* = j', i = \sigma(n+j'), S_i \le q_{1-\alpha}(\tilde{S}_1, \ldots, \tilde{S}_M)}}{X_{1:(n+m)}}\\
& =\frac{1}{m} \cdot \EEst{\sum_{k=1}^M \sum_{i \in I_k} \sum_{j' \in I_k^0}\frac{|\mathcal{S}_{I_{1:M},k}^{i,j'}|}{|\mathcal{S}_{I_{1:M}}|}\cdot \One{S_i \le  q_{1-\alpha}(\tilde{S}_1, \ldots, \tilde{S}_M)}}{X_{1:(n+m)}},
\end{align*}
where 
$\mathcal{S}_{I_{1:M},k}^{i,j'} = \left\{\sigma \in \mathcal{S}_{I_{1:M}} : \sigma(n+j') = i\right\}$,
and the last step holds since $j^*$ is independent of the data and $\PPst{j^* = j'}{X_{1:(n+m)}} = \frac{1}{m}$ for any $j \in [m]$. 
From the definition of $\mathcal{S}_{I_{1:M}}$ in 
\eqref{eqn:S_1_M}, 
for any $i \in I_k$ and $j' \in I_k^0$,
the number of permutations in $\mathcal{S}_{I_{1:M},k}^{i,j'}$ is
$|\mathcal{S}_{I_{1:M},k}^{i,j'}| = \bigg[\prod_{k' \neq k} N_k!\bigg] \cdot (N_k-1)!$,
because these permutations have only one value within the $k$-th block fixed, and all other values are arbitrary.
Thus
\[\frac{|\mathcal{S}_{I_{1:M},k}^{i,j'}|}{|\mathcal{S}_{I_{1:M}}|} = \frac{\bigg[\prod_{k' \neq k} N_k!\bigg] \cdot (N_k-1)!}{\prod_{k'=1}^M N_{k'}!} = \frac{1}{N_k}.\]
Therefore, putting everything together, we have
\begin{align*}
    &\PPst{S_{n+j^*} \le q_{1-\alpha}(\tilde{S}_1, \ldots, \tilde{S}_M)}{X_{1:(n+m)}}
    = \EEst{\sum_{k=1}^M \sum_{i \in I_k} \sum_{j' \in I_k^0}\frac{1}{m}\frac{1}{N_k}\cdot \One{S_i \le  q_{1-\alpha}(\tilde{S}_1, \ldots, \tilde{S}_M)}}{X_{1:(n+m)}}\\
    & =\EEst{\sum_{k=1}^M \sum_{i \in I_k} \frac{1}{m}\frac{N_k^0}{N_k}\cdot \One{S_i \le  q_{1-\alpha}(\tilde{S}_1, \ldots, \tilde{S}_M)}}{X_{1:(n+m)}} \ge 1-\alpha,
\end{align*}
where the last inequality applies~\eqref{eqn:weighted_quantile_2}. The desired inequality follows by observing that
\begin{align*}
    \sum_{k=1}^M \sum_{i \in I_k^1}\frac{1}{m}\cdot \frac{N_k^0}{N_k}\cdot \delta_{S_i} + \frac{1}{m}\sum_{k=1}^M \frac{(N_k^0)^2}{N_k} \cdot\delta_{+\infty} = \sum_{k=1}^M \sum_{i \in I_k}\frac{1}{m}\cdot \frac{N_k^0}{N_k}\cdot \delta_{\bar{S}_i},
\end{align*}
where $\bar{S}_i = S_i \One{i \le n} + (+\infty) \One{i > n}$, so that
\[Q_{1-\alpha}\left(\sum_{k=1}^M \sum_{i \in I_k^1}\frac{1}{m}\cdot \frac{N_k^0}{N_k}\cdot \delta_{S_i} + \frac{1}{m}\sum_{k=1}^M \frac{(N_k^0)^2}{N_k} \cdot\delta_{+\infty}\right) \ge q_{1-\alpha}(\tilde{S}_1, \ldots, \tilde{S}_M)\]
holds deterministically.

\subsection{Proof of Corollary~\ref{cor:in_exp_discrete_partition}}

By construction, we have that
\begin{align*}
    &\EEst{\frac{1}{m}\sum_{j=1}^m \One{Y_{n+j} \in \ch_U(X_n+j,j)}}{X_{1:(n+m)}}
    =\EEst{\frac{1}{m}\sum_{j=1}^m \One{Y_{n+j} \in \ch^{\ell_j}(X_{n+j})}}{X_{1:(n+m)}}\\
    & =\EEst{\frac{1}{m}\sum_{l=1}^L\sum_{j \in U_{\ell}} \One{Y_{n+j} \in \ch^{\ell_j}(X_{n+j})}}{X_{1:(n+m)}}\\
    &=\sum_{l=1}^L \frac{N_{\ell}^0}{m}\cdot\EEst{\frac{1}{N_{\ell}^0}\sum_{j \in U_{\ell}} \One{Y_{n+j} \in \ch^{\ell_j}(X_{n+j})}}{X_{1:(n+m)}}
    \ge 1-\alpha,
\end{align*}
where the last step holds by Theorem~\ref{thm:MAR_in_exp}. This proves the claim.

\subsection{Proof of Lemma~\ref{lem:dtv_0_1}}
    By the assumption, we have
   \[p_1 \le p_{A \mid X}(x) \le p_2 \textnormal{ for all }x \in D, \textnormal{ where } p_1 = \frac{t}{1+t} \textnormal{ and } p_2 = \frac{t(1+\eps)}{1+t(1+\eps)}.\] 
    Take any measurable set $V \subset \R$. We have
    \begin{align*}
        &\PPst{S \in V}{A=1, X \in D} = \frac{\PP{S \in V, A=1, X \in D}}{\PP{A=1, X \in D}} = \frac{\EE{\PPst{S \in V, A=1, X \in D}{X}}}{\EE{\PPst{A=1, X \in D}{X}}}\\
        &=\frac{\EE{\PPst{S \in V, A=1}{X}\cdot \One{X \in D}}}{\EE{\PPst{A=1}{X}\cdot \One{X \in D}}}
        = \frac{\EE{\PPst{S \in V}{A=1, X}\cdot p_{A \mid X}(X)\cdot \One{X \in D}}}{\EE{p_{A \mid X}(X)\cdot \One{X \in D}}}\\
        &=\frac{\EE{\PPst{S \in V}{A=0, X}\cdot p_{A \mid X}(X)\cdot \One{X \in D}}}{\EE{p_{A \mid X}(X)\cdot \One{X \in D}}}
        \le \frac{p_2 \cdot \EE{\PPst{S \in V}{A=0, X}\cdot \One{X \in D}}}{p_1 \cdot \EE{\One{X \in D}}},
    \end{align*}
    where the last equality holds since $S$ is conditionally independent of $A$ given $X$. Similarly, we have
    \begin{align*}
        \PPst{S \in V}{A=0, X \in D} &= \frac{\EE{\PPst{S \in V}{A=0, X}\cdot (1-p_{A \mid X}(X))\cdot \One{X \in D}}}{\EE{(1-p_{A \mid X}(X))\cdot \One{X \in D}}}\\
        &\ge \frac{(1-p_2) \cdot \EE{\PPst{S \in V}{A=0, X}\cdot \One{X \in D}}}{(1-p_1) \cdot \EE{\One{X \in D}}}.
    \end{align*}
    It follows that with $\zeta = \frac{\EE{\PPst{S \in V}{A=0, X}\cdot \One{X \in D}}}{ \EE{\One{X \in D}}}$,
    \begin{align*}
        &\PPst{S \in V}{A=1, X \in D}-\PPst{S \in V}{A=0, X \in D}\\
        &\le  \zeta\cdot \left[\frac{p_2}{p_1} - \frac{1-p_2}{1-p_1}\right]
        =\zeta \cdot \frac{1-p_2}{1-p_1}\cdot \left[\frac{\frac{p_2}{1-p_2}}{\frac{p_1}{1-p_1}} - 1\right]
        =\zeta \cdot \frac{1-p_2}{1-p_1} \cdot \eps \le \eps.
    \end{align*}
    Similarly,
    we can show that
    \[\PPst{S \in V}{A=1, X \in D}-\PPst{S \in V}{A=0, X \in D} \ge -\eps.\]
    This holds for any measurable set $D \subset \R$, implying that
    \begin{multline*}
        \dtv(P_{S \mid A=1, X \in D}, P_{S \mid A=0, X \in D})\\
        = \sup_{D \subset \R : \text{measurable}} |\PPst{S \in V}{A=1, X \in D}-\PPst{S \in V}{A=0, X \in D}| \le \eps.
    \end{multline*}

\subsection{Proof of Theorem~\ref{thm:MAR_known_p_A_X}}
Let $j^*$ be a random variable drawn from $\textnormal{Unif}([m])$. Then we see that
\begin{align*}
    &\PPst{Y_{n+j^*} \in \chp(X_{n+j^*})}{B_{1:(n+m)}}
    = \EEst{\One{Y_{n+j^*} \in \chp(X_{n+j^*})}}{B_{1:(n+m)}}\\
    &= \EEst{\sum_{j=1}^m\One{Y_{n+j^*} \in \chp(X_{n+j^*}), j^* = j}}{B_{1:(n+m)}}\\
    &= \sum_{j=1}^m \EEst{\One{j^* = j}}{B_{1:(n+m)}} \cdot \EEst{\One{Y_{n+j} \in \chp(X_{n+j})}}{B_{1:(n+m)}}\\
    &= \EEst{\frac{1}{m}\sum_{j=1}^m \One{Y_{n+j} \in \chp(X_{n+j})}}{B_{1:(n+m)}}.
\end{align*}

Therefore,
by the definition of $\chp$,
the target inequality can equivalently be written as
\begin{equation}\label{eqn:target_ineq}
    \PPst{S_{n+j^*} \le Q_{1-\alpha}\left(\sum_{k=1}^M \sum_{i \in I_k^{\B,1}}\frac{1}{m}\cdot \frac{N_k^{\B,0}}{N_k^\B}\cdot \delta_{S_i} + \frac{1}{m}\sum_{k=1}^M \frac{(N_k^{\B,0})^2}{N_k^\B} \cdot\delta_{+\infty}\right)}{B_{1:(n+m)}} \ge 1-\alpha-\eps.
\end{equation}
Here, the probability is taken with respect to the following distribution:
\begin{equation}\label{eqn:P}
    P:
    \begin{cases}
    S_i \indepsim P_{s(X,Y) \mid A = 0, X \in {B_i}} &\textnormal{ for } i\in [n],\\
    S_i \indepsim P_{s(X,Y) \mid A = 1, X \in {B_i}} &\textnormal{ for } i\in [n+m]\backslash[n],\\
    j^* \sim \textnormal{Unif}([m]), &\textnormal{ independently of } (S_i)_{1 \le i \le n+m},
    \end{cases}
\end{equation}
where we treat $B_{1:(n+m)}$ as fixed for convenience. Next, we consider 
the distribution $Q$, 
which is identical to $P$, except 
$S_i \indepsim P_{s(X,Y) \mid A = 1, X \in {B_i}} \textnormal{ for } i > n$.
Note that under $Q$, the dataset $\{B_i,Z_i\}_{i \in [n+m]}$ satisfies the assumptions of Theorem~\ref{thm:MAR_in_exp}. Therefore, we have
\[\Ppst{Q}{S_{n+j^*} \le Q_{1-\alpha}\left(\sum_{k=1}^M \sum_{i \in I_k^{\B,1}}\frac{1}{m}\cdot \frac{N_k^{\B,0}}{N_k^\B}\cdot \delta_{S_i} + \frac{1}{m}\sum_{k=1}^M \frac{(N_k^{\B,0})^2}{N_k^\B} \cdot\delta_{+\infty}\right)}{B_{1:(n+m)}} \ge 1-\alpha.\]
The event inside the probability depends on $(S_{n+j})_{j \in [m]}$ and $j^*$ only through $S_{n+j^*}$, and the distribution of $(S_i)_{i \in [n]}$ is the same under $P$ and $Q$ and is independent of $S_{n+j^*}$. 
Therefore, to show~\eqref{eqn:target_ineq}, it is sufficient to prove
$\dtv(P^*, Q^*) \le \eps$,
where $P^*$ and $Q^*$ denote the distribution of $S_{n+j^*}$ under $P$ and $Q$, respectively.

Take any measurable set $D \subset \R$. It holds that
\begin{align*}
    &\Pp{P}{S_{n+j^*} \in D} - \Pp{Q}{S_{n+j^*} \in D}
    =\frac{1}{m}\sum_{j=1}^m \Big[\Pp{P}{S_{n+j} \in D}-\Pp{Q}{S_{n+j} \in D}\Big]\\
    &=\frac{1}{m}\sum_{j=1}^m \Big[\PPst{S \in V}{A=0, X \in D_{n+j}} - \PPst{S \in V}{A=1, X \in D_{n+j}}\Big].
\end{align*}
Applying Lemma~\ref{lem:dtv_0_1} for each $j \in [m]$, we have that
\[-\eps \le \PPst{S \in V}{A=0, X \in D_{n+j}} - \PPst{S \in V}{A=1, X \in D_{n+j}} \le \eps, \textnormal{ for all } j \in [m],\]
which implies
$-\eps \le \Pp{P}{S_{n+j^*} \in D} - \Pp{Q}{S_{n+j^*} \in D} \le \eps$
by the above equality. Since this holds for an arbitrary $D$, we have shown that $\dtv(P^*, Q^*) \le \eps$, and the desired inequality follows.

\subsection{Proof of Theorem~\ref{opt}}
\label{lb}

Denote $p_x = \mathbb{P}(A=1 \mid X = x)$ for any $X$. Due to the form of
$\operatorname{Cover}(\mathcal{V}, P, A = a)$, 
to find a bound for $\Delta_\mathcal{V}(P)$,
it is enough to bound the probabilities
$\mathbb{P}(Y \in  V \mid X \in D, A = a)$ for fixed sets $V \subset \Y$ and $D \in \mathcal{D}$. 
Similarly to the proof of 
Lemma \ref{lem:dtv_0_1},
we compute:
\begin{align*}
&\PP{Y \in V \mid A=1, X \in D}
= \frac{\PPst{Y \in V, A = 1}{X \in D}}{\PPst{A = 1 }{X \in D}}\\
&= \frac{\EEst{\PPst{Y \in V, A = 1}{X,X \in D}}{X \in D}}{\EEst{\PPst{A = 1 }{X,X \in D}}{X \in D}} = \frac{\EEst{\PPst{Y \in V}{X} \cdot p_X}{X \in D}}{\EEst{p_X}{X \in D}},
\end{align*}
where we apply the missing at random assumption in the last equality. Therefore, we have
\begin{align*}
&\PP{Y \in V \mid A = 1, X \in D} - \PP{Y \in V \mid A = 0, X \in D} \\
&= \frac{\EEst{\PPst{Y \in V}{X} \cdot p_X}{X \in D}}{\EEst{p_X}{X \in D}} - \frac{\EEst{\PPst{Y \in V}{X} \cdot (1-p_X)}{X \in D}}{\EEst{1-p_X}{X \in D}}\\
&= \EEst{\PPst{Y \in V}{X} \cdot h_D(X)}{X \in D},
\end{align*}
where
\[h_D(x) = \frac{p_x}{\EEst{p_X}{X \in D}} - \frac{1-p_x}{\EEst{1-p_X}{X \in D}}.\]
It follows that, unless $V = \emptyset$ or $\Y$,
\begin{equation}\label{eqn:g_d_p}
\begin{split}
    g(D,P_{X \mid X \in D}) &:= \sup_{P_{Y \mid X}, V} |\PP{Y \in V \mid A = 1, X \in D} - \PP{Y \in V \mid A = 0, X \in D}|\\
    &= \EEst{h_D(X)\cdot \one{h_D(X) \geq 0}}{X \in D} = -\EEst{h_D(X)\cdot \one{h_D(X) < 0}}{X \in D},
\end{split}
\end{equation}
where the last equality holds since $\EEst{h_D(X)}{X \in D} = 0$ by the definition of $h_D$. If $V = \emptyset$ or $\Y$, $g(D,P_{X \mid X \in D}) = 0$ clearly holds.
Then, 
indexing the sets in $\mathcal{V}$ via the sets in $\mathcal{D}$,
observe that
\begin{equation}\label{eqn:delta_s_p}
\begin{split}
    \Delta_\mathcal{V}(P_{A|X}) &= \sup_{P_{Y \mid X}, P_X}\left|\sum_{D \in \mathcal{D}}\PP{X \in D} \cdot \big(\PPst{Y \in V_D}{X \in D, A=0} - \PPst{Y \in V_D}{X \in D, A=1}\big)\right|\\
    &= \sup_{D \in \mathcal{D}} \sup_{P_X} g(D,P_{X \mid X \in D}).
\end{split}
\end{equation}
Intuitively, this holds because we can assign all the weight $\PP{X \in D}$ to the set $D$ that attains the largest value of $\sup_{P_X} g(D, P_X)$ in order to maximize the sum.

Now we investigate the term $\sup_{P_X} g(D, P_{X \mid X \in D})$ for a fixed $D = D_\lambda$ with $\lambda \in \Lambda_\mathcal{V}$.
Consider a distribution $P_{X \mid X \in D}^0$ supported on $\{x_1, x_2\} \subset D$ defined as
\[P_{X \mid X \in D}^0 = \frac{1}{2}\delta_{x_1} + \frac{1}{2}\delta_{x_2}.\]
Let $p_1 = p_{x_1}$ and $p_2 = p_{x_2}$, and, without loss of generality, assume that $p_1 > p_2$---we consider the case where $x \mapsto p_x$ is not constant on $D$ so that $x_1$ and $x_2$ can be chosen while satisfying $p_1 \neq p_2$; otherwise, $h_D \equiv 0$, and hence $g(D, P_{X \mid X \in D})$ is trivially zero.

It is easy to verify that under $P_{X \mid X \in D}^0$, we have $h_D(x_1) > 0$ and $h_D(x_2) < 0$, and that $\EEst{p_X}{X \in D} = \frac{1}{2}(p_1 + p_2)$. Therefore,
\begin{multline*}
    g(D,P_{X \mid X \in D}^0) = \EEst{h_D(X)\cdot \one{h_D(X) \geq 0}}{X \in D} = \frac{1}{2}\cdot h_D(x_1) = \frac{p_1}{p_1+p_2} - \frac{1-p_1}{2-p_1-p_2}\\
    = \frac{p_1(2-p_1-p_2) - (1-p_1)(p_1+p_2)}{(p_1+p_2)(2-p_1-p_2)} = \frac{p_1-p_2}{(p_1+p_2)(2-p_1-p_2)}.
\end{multline*}
Now define 
\[\eps(x_1,x_2) := \left|\frac{p_1/(1-p_1)}{p_2/(1-p_2)} - 1\right| = \left|\frac{p_1(1-p_2)}{p_2(1-p_1)} - 1\right| = \frac{p_1-p_2}{p_2(1-p_1)}.\]
Then
\begin{align*}
    &g(D,P_{X \mid X \in D}^0) = \frac{p_1-p_2}{(p_1+p_2)(2-p_1-p_2)} = \frac{p_1-p_2}{p_2(1-p_1)} \cdot \frac{p_2(1-p_1)}{(p_1+p_2)(2-p_1-p_2)}\\
    &= \eps(x_1,x_2) \cdot \frac{1}{\left(1+\frac{p_1}{p_2}\right)\left(1+\frac{1-p_2}{1-p_1}\right)} \geq \eps(x_1,x_2) \cdot \frac{1}{4 \cdot \frac{p_1}{p_2} \cdot \frac{1-p_2}{1-p_1}} = \frac{\eps(x_1,x_2)}{4(1+\eps(x_1,x_2))},
\end{align*}
where the inequality holds since $p_1/p_2 \geq 1$ and 
$(1-p_2)/(1-p_1) \geq 1$. 
Therefore, we have
\[\sup_{P_X} g(D,P_{X \mid X \in D}) \geq g(D,P_{X \mid X \in D}^0) \geq \frac{\eps(x_1,x_2)}{4(1+\eps(x_1,x_2))}.\]
Since the above inequality holds for any choice of $x_1$ and $x_2$, we have
\[\sup_{P_X} g(D,P_{X \mid X \in D}) \geq \sup_{x_1,x_2 \in D} \frac{\eps(x_1,x_2)}{4(1+\eps(x_1,x_2))} = \frac{\mathcal{E}(D, P_{A|X})}{4(1+\mathcal{E}(D, P_{A|X}))},\]
noting that $t \mapsto t/(4(1+t))$ is a continuous and nondecreasing function. Since the above inequality holds for any $D_\lambda$ with $\lambda \in \Lambda_\mathcal{V}$, applying~\eqref{eqn:delta_s_p}, we have
\[\Delta_\mathcal{V}(P_{A|X}) = \sup_{D \in \mathcal{D}} \sup_{P_X} g(D,P_{X \mid X \in D}) \geq \sup_{\lambda \in \Lambda_\mathcal{V}} \frac{\mathcal{E}(D, P_{A|X})}{4(1+\mathcal{E}(D, P_{A|X}))} = \frac{\mathcal{E}_\mathcal{V}(\mathcal{D}, P_{A|X})}{4(1+\mathcal{E}_\mathcal{V}(\mathcal{D}, P_{A|X}))},\]
which implies the desired bound.

\subsection{Proof of Theorem~\ref{thm:MAR_estimated_p_A_X}}
The proof is similar to that of Theorem~\ref{thm:MAR_known_p_A_X}. 
In the last step,  it is sufficient to prove
$\dtv(P^*, Q^*) \le  \eps + \delta_{\hat{p}_{A\mid X}} + \eps \cdot \delta_{\hat{p}_{A\mid X}}$,
and it is again enough to prove that for any $B \in \B$ (where $\B$ is constructed based on~\eqref{eqn:eps_partition} with $\hat{p}_{A \mid X}$),
\[\dtv(P_{S \mid A=1, X \in D}, P_{S \mid A=0, X \in D}) \le   \eps + \delta_{\hat{p}_{A\mid X}} + \eps \cdot \delta_{\hat{p}_{A\mid X}}.\]
By the definition of $\delta_{\hat{p}_{A\mid X}}$, for any $x \in \X$, it holds that
$-\log (1+\delta_{\hat{p}_{A\mid X}}) \le 2\log f_{p,\hat{p}}(x) \le \log (1+\delta_{\hat{p}_{A\mid X}})$,
and consequently
$1/\sqrt{1+\delta_{\hat{p}_{A\mid X}}} \le f_{p,\hat{p}}(x) \le \sqrt{1+\delta_{\hat{p}_{A\mid X}}}$.
Fix any $k \in \mathbb{Z}$. 
Then for any $x \in D_k$, we have
\[(1+\eps)^k \le \frac{\hat{p}_{A \mid X}(x)}{1-\hat{p}_{A \mid X}(x)} \le (1+\eps)^{k+1},\]
by the construction of $\B$, and it follows that
\[(1+\eps)^k \cdot \frac{1}{\sqrt{1+\delta_{\hat{p}_{A\mid X}}}}\le \frac{p_{A \mid X}(x)}{1-p_{A \mid X}(x)} = \frac{\hat{p}_{A \mid X}(x)}{1-\hat{p}_{A \mid X}(x)} \cdot f_{p,\hat{p}}(x) \le (1+\eps)^{k+1} \cdot \sqrt{1+\delta_{\hat{p}_{A\mid X}}},\]
for any $x \in D_k$. Therefore, by Lemma~\ref{lem:dtv_0_1}, we have
\[\dtv(P_{S \mid A=1, X \in D}, P_{S \mid A=0, X \in D}) \le (1+\eps)(1+\delta_{\hat{p}_{A\mid X}})-1 = \eps + \delta_{\hat{p}_{A\mid X}} + \eps \cdot \delta_{\hat{p}_{A\mid X}}
%\le \eps + 2\delta_{\hat{p}_{A\mid X}}
,\]
as desired.

\subsection{Proof of Theorem~\ref{thm:MAR_squared}}

The proof follows
similar steps as the proof of Theorem~\ref{thm:MAR_in_exp}. Throughout the proof, we denote $\nu = \frac{1}{m^2}$. Let $j_1^*$ and $j_2^*$ be the two independent draws from $\textnormal{Unif}([m])$, independent of the data. Then, with $E_{j_1,j_2} = \{Y_{n+j_1} \notin \ch^2(X_{n+j_1}), Y_{n+j_2} \notin \ch^2(X_{n+j_2})\}$, 
we observe that
\begin{align*}
    &\PPst{E_{j_1^*,j_2^*}}{X_{1:(n+m)}}
    = \EEst{\One{E_{j_1^*,j_2^*}}}{X_{1:(n+m)}}
     = \EEst{\sum_{j_1,j_2 \in [m]} 
     \One{j_1^* = j_1, j_2^* = j_2} \cdot\One{E_{j_1^*,j_2^*}}}{X_{1:(n+m)}}\\
    & =\sum_{j_1,j_2 \in [m]} \nu \cdot \EEst{E_{j_1,j_2} }{X_{1:(n+m)}}
    = \EEst{\Bigg(\frac{1}{m}\sum_{j=1}^m \One{Y_{n+j} \notin \ch^2(X_{n+j})}\Bigg)^2}{X_{1:(n+m)}},
\end{align*}
which implies that it is equivalent to prove that the simultaneous miscoverage rate for the two randomly chosen missing outcomes is bounded by $\alpha^2$.

Fix $X_{1:(n+m)}$, and define $q_{1-\alpha^2} : \R^{N_1}\times \R^{N_2} \times \ldots \times \R^{N_M} \rightarrow \R$,
for $(\tilde{s}_1, \ldots, \tilde{s}_M)$
with  $\tilde{s}_k = (\tilde{s}_{k1},\tilde{s}_{k2}, \ldots, \tilde{s}_{k N_k})$ for each $k$
by
\begin{multline}\label{eqn:squared_quantile_fn}
    q_{1-\alpha^2}(\tilde{s}_1, \ldots, \tilde{s}_M) = Q_{1-\alpha^2}\left(\rule{0cm}{1.2cm} \nu\sum_{k : N_k^0 \ge 1} \sum_{i=1}^{N_k}  \frac{N_k^0}{N_k}\delta_{\tilde{s}_{ki}} \right. \\
    \left. +\nu \sum_{k : N_k^0 \ge 2}\sum_{1 \le i \neq j \le N_k} 
    \frac{N_k^0(N_k^0-1)}{N_k(N_k-1)}\delta_{\min\{\tilde{s}_{ki},\tilde{s}_{kj}\}} + \nu\sum_{\substack{k \neq k :\\ N_k^0 \ge 1, N_{k'}^0 \ge 1}}\sum_{i=1}^{N_k} \sum_{j=1}^{N_{k'}} 
    \frac{N_k^0N_{k'}^0}{N_kN_{k'}}\delta_{\min\{\tilde{s}_{ki},\tilde{s}_{kj}\}} \rule{0cm}{1.2cm}\right).
\end{multline}
 For the function $q_{1-\alpha^2}$, we observe the following properties. First, for any permutations $\sigma_1 \in \mathcal{S}_{N_1}, \ldots, \sigma_M \in \mathcal{S}_{N_M}$, it holds that
\begin{equation}\label{eqn:q_sq_exch}
    q_{1-\alpha^2}(\tilde{s}_1, \ldots, \tilde{s}_M) = q_{1-\alpha^2}(\tilde{s}_1^{\sigma_1}, \ldots, \tilde{s}_M^{\sigma_M}),
\end{equation}
where $\tilde{s}_1^{\sigma_1}, \ldots, \tilde{s}_M^{\sigma_M}$ are defined as in the proof of Theorem~\ref{thm:MAR_in_exp}. Next, let $\tilde{S}_k = (S_i)_{i \in I_k}$ for $k\in [M]$. 
% Then we have
% \begin{multline*}
%     q_{1-\alpha^2}(\tilde{S}_1, \ldots, \tilde{S}_M) = Q_{1-\alpha^2}\left(\rule{0cm}{1.2cm} \sum_{k : N_k^0 \ge 1} \sum_{i \in I_k} \nu \frac{N_k^0}{N_k}\delta_{S_i}\right.  \\
%     \left.  + \sum_{k : N_k^0 \ge 2}\sum_{\substack{i,j\in I_k \\ i \neq j}} \nu  \frac{N_k^0(N_k^0-1)}{N_k(N_k-1)}\delta_{\min\{S_i,S_j\}} + \sum_{\substack{k \neq k :\\ N_k^0 \ge 1, N_{k'}^0 \ge 1}}\sum_{i \in I_k} \sum_{j \in I_{k'}} \nu  \frac{N_k^0}{N_k}\frac{N_{k'}^0}{N_{k'}}\delta_{\min\{S_i,S_j\}} \rule{0cm}{1.2cm}\right).
% \end{multline*}
By definition of $Q_{1-\alpha^2}$, it holds deterministically that
\begin{equation}\label{eqn:q_sq_ineq}
\begin{split}
    &\sum_{k : N_k^0 \ge 1} \sum_{i \in I_k} \nu\cdot \frac{N_k^0}{N_k}\cdot\One{S_i \le q_{1-\alpha^2}(\tilde{S}_1, \ldots, \tilde{S}_M)}  \\
      &+ \sum_{k : N_k^0 \ge 2}\sum_{\substack{i,j\in I_k \\ i \neq j}} \nu \cdot \frac{N_k^0(N_k^0-1)}{N_k(N_k-1)}\cdot\One{\min\{S_i,S_j\} \le q_{1-\alpha^2}(\tilde{S}_1, \ldots, \tilde{S}_M)}\\
      &+ \sum_{\substack{k \neq k :\\ N_k^0 \ge 1, N_{k'}^0 \ge 1}}\sum_{i \in I_k} \sum_{j \in I_{k'}} \nu \cdot \frac{N_k^0}{N_k}\cdot\frac{N_{k'}^0}{N_{k'}}\cdot \One{\min\{S_i,S_j\} \le q_{1-\alpha^2}(\tilde{S}_1, \ldots, \tilde{S}_M)} \ge 1-\alpha^2.
      \end{split}
\end{equation}
Next, since $\{Y_i : i \in I_k\}$ is an exchangeable draw from $P_{Y \mid X=X_k}$ conditional on $X_k$ by the missing at random assumption, $\tilde{S}_k$ is exchangeable conditional on $X_{1:(n+m)}$, for each $k\in [M]$. 
Since this holds jointly, we have
\[(\tilde{S}_1, \tilde{S}_2, \ldots, \tilde{S}_M) \stackrel{\text{d}}{=} (\tilde{S}_1^{\sigma_1}, \tilde{S}_2^{\sigma_2}, \ldots, \tilde{S}_M^{\sigma_M}) \mid X_{1:(n+m)},\]
for any permutations $\sigma_1 \in \mathcal{S}_{N_1}, \ldots, \sigma_M \in \mathcal{S}_{N_M}$. Therefore, conditional on $X_{1:(n+m)}$,
\begin{align*}
    &\One{\min\{S_{n+j_1^*}, S_{n+j_2^*}\} \le q_{1-\alpha^2}(\tilde{S}_1,\ldots,\tilde{S}_M)} \stackrel{\text{d}}{=} \One{\min\{S_{n+j_1^*}^\sigma, S_{n+j_2^*}^\sigma\} \le q_{1-\alpha^2}(\tilde{S}_1^{\sigma_1},\ldots,\tilde{S}_M^{\sigma_M})}\\
    &= \One{\min\{S_{n+j_1^*}^\sigma, S_{n+j_2^*}^\sigma\} \le q_{1-\alpha^2}(\tilde{S}_1,\ldots,\tilde{S}_M)}
    = \One{\min\{S_{\sigma(n+j_1^*)}, S_{\sigma(n+j_2^*)}\} \le q_{1-\alpha^2}(\tilde{S}_1,\ldots,\tilde{S}_M)},
\end{align*}
where the second equality applies~\eqref{eqn:q_sq_exch}, and $\sigma \in \mathcal{S}_n$ is a permutation determined by $\sigma_1,\ldots,\sigma_M$, defined as in the proof of Theorem~\ref{thm:MAR_in_exp}. 

Define $\mathcal{S}_{I_{1:M}}$ as in~\eqref{eqn:S_1_M}, and let
\begin{align*}
    &\mathcal{S}_{I_{1:M},k}^{i,i'} = \{\sigma \in \mathcal{S}_{I_{1:M}} : \sigma(i') = i\},\quad
    \mathcal{S}_{I_{1:M},k_1,k_2}^{(i,i'),(j,j')} = \{\sigma \in \mathcal{S}_{I_{1:M}} : \sigma(i') = i, \sigma(j') = j\},
\end{align*}
for $i,i',j,j' \in [n+m]$ and $k\in[M]$. 
The sizes 
$|\mathcal{S}_{I_{1:M}}|$, $ |\mathcal{S}_{I_{1:M},k}^{i,i'}|$ are given in the
proof of Theorem~\ref{thm:MAR_in_exp}
and the size of $\mathcal{S}_{I_{1:M},k}^{i,i'}$ is
\begin{align*}
    % &|\mathcal{S}_{I_{1:M}}| = \prod_{k=1}^M N_{k} !,\qquad
    % |\mathcal{S}_{I_{1:M},k}^{i,i'}| = \prod_{k \neq k_1} N_k! \cdot (N_{k_1}-1)! \text{ if } i,i' \in I_{k_1},\\
    &|\mathcal{S}_{I_{1:M},k_1,k_2}^{(i,i'),(j,j')}| =
    \begin{cases}
        \prod_{k \neq k_1} N_k ! \cdot (N_{k_1}-2)! &\text{ if } i,i',j,j' \in I_{k_1},\\
        \prod_{k \neq k_1, k_2} N_k ! \cdot (N_{k_1}-1)! (N_{k_2}-1)! &\text{ if } i,i' \in I_{k_1}, j,j' \in I_{k_2}\; (k_1 \neq k_2).
    \end{cases}
\end{align*}
From the above observations, 
with $U_{ij} = \One{\min\{S_i, S_j\} \le q_{1-\alpha^2}(\tilde{S}_1,\ldots,\tilde{S}_M)}$ for all $i,j$,
and for each $k\in[M]$,
introducing
the dummy indices 
$i_1,i_2 \in I_k$, $j_1,j_2 \in I_k^0$ 
as in the proof of Theorem~\ref{thm:MAR_in_exp},
we have
\begin{align*}
    &    \EEst{U_{n+j_1^*,n+j_2^*}}{X_{1:(n+m)}}
    = \EEst{\frac{1}{|\mathcal{S}_{I_{1:M}}|}\sum_{\sigma \in \mathcal{S}_{I_{1:M}}} U_{\sigma(n+j_1^*)\sigma(n+j_2^*)}}{X_{1:(n+m)}}\\
    & =\mathbb{E}\left[ \frac{1}{|\mathcal{S}_{I_{1:M}}|}\sum_{\sigma \in \mathcal{S}_{I_{1:M}}} \left[\rule{0cm}{0.8cm} \sum_{k : N_k^0 \ge 1} \sum_{i_1 \in I_k} \sum_{j_1 \in I_k^0} \One{j_1^* = j_2^* = j_1}\One{i=\sigma(n+j_1)}\One{S_{i_1} \le q_{1-\alpha^2}(\tilde{S}_1,\ldots,\tilde{S}_M)} \right. \right.\\
    &+\sum_{k : N_k^0 \ge 2} \sum_{\substack{i_1,i_2\in I_k \\ i_1 \neq i_2}} \sum_{\substack{j_1,j_2\in I_k^0 \\ j_1 \neq j_2}} \One{j_1^* = j_1, j_2^* = j_2}\One{i_1 = \sigma(n+j_1), i_2=\sigma(n+j_2)}  U_{i_1 i_2}\\
    &+\sum_{\substack{k \neq k': \\ N_k^0 \ge 1, N_{k'}^0 \ge 1}} \sum_{i_1 \in I_k} \sum_{j_1 \in I_k^0} \sum_{i_2 \in I_{k'}} \sum_{j_2 \in I_{k'}^0} \One{j_1^* = j_1, j_2^* = j_2}\One{i_1 = \sigma(n+j_1), i_2=\sigma(n+j_2)}
    \left. \left. U_{i_1 i_2} \right] \ \middle| \ X_{1:(n+m)} \right].
\end{align*}
By collecting indices,
this further equals
 \begin{align*}   
    &\mathbb{E}\left[\rule{0cm}{1.2cm} \nu \sum_{k : N_k^0 \ge 1} \sum_{i_1 \in I_k} \sum_{j_1 \in I_k^0} \frac{|\mathcal{S}_{I_{1:M},k}^{i_1,n+j_1}|}{|\mathcal{S}_{I_{1:M}}|}\cdot \One{S_{i_1} \le q_{1-\alpha^2}(\tilde{S}_1,\ldots,\tilde{S}_M)} + \nu\sum_{k : N_k^0 \ge 2} \sum_{\substack{i_1,i_2\in I_k \\ i_1 \neq i_2}} \sum_{\substack{j_1,j_2\in I_k^0 \\ j_1 \neq j_2}} \frac{|\mathcal{S}_{I_{1:M},k,k}^{(i_1,i_2),(n+j_1,n+j_2')}|}{|\mathcal{S}_{I_{1:M}}|} U_{i_1 i_2}\right.\\
    &\quad \left. +\nu \sum_{\substack{k \neq k': \\ N_k^0 \ge 1, N_{k'}^0 \ge 1}} \sum_{i_1 \in I_k} \sum_{j_1 \in I_k^0} \sum_{i_2 \in I_{k'}} \sum_{j_2 \in I_{k'}^0} \frac{|\mathcal{S}_{I_{1:M},k,k'}^{(i_1,i_2),(n+j_1,n+j_2)}|}{|\mathcal{S}_{I_{1:M}}|} U_{i_1 i_2} \rule{0cm}{1.2cm}\right]\\
    =&\mathbb{E}\left[\rule{0cm}{1.2cm} \sum_{k : N_k^0 \ge 1} \sum_{i_1 \in I_k} \nu \frac{N_k^0}{N_k}  \One{S_{i_1} \le q_{1-\alpha^2}(\tilde{S}_1,\ldots,\tilde{S}_M)}+ \sum_{k : N_k^0 \ge 2} \sum_{\substack{i_1,i_2\in I_k \\ i_1 \neq i_2}} \nu \cdot N_k^0 (N_k^0-1) \cdot \frac{1}{N_k(N_k-1)} U_{i_1 i_2} \right.\\
    &\quad\left. + \sum_{\substack{k \neq k': \\ N_k^0 \ge 1, N_{k'}^0 \ge 1}} \sum_{i_1 \in I_k} \sum_{i_2 \in I_{k'}} \nu \cdot N_k^0 N_{k'}^0 \cdot \frac{1}{N_k N_{k'}} U_{i_1 i_2} \rule{0cm}{1.2cm}\right]
    \ge 1-\alpha^2,
\end{align*}
where the last step applies~\eqref{eqn:q_sq_ineq}. 

Next, let $\tilde{\bar{S}}_k = (\bar{S}_i)_{i \in I_k}$ for $k\in [M]$, and observe that
for all $x\in \mathcal{X}$,
$\ch^2(x) = \{y \in \Y : s(x,y) \le q_{1-\alpha^2}(\tilde{\bar{S}}_1,\ldots,\tilde{\bar{S}}_M)\}$.
From the calculations above, we have
\begin{align*}
&\EEst{\Bigg(\frac{1}{m}\sum_{j=1}^m \One{Y_{n+j} \notin \ch^2(X_{n+j})}\Bigg)^2}{X_{1:(n+m)}} = \PPst{E_{j_1^*,j_2^*}}{X_{1:(n+m)}}\\
& =\PPst{\min\{S_{n+j_1^*}, S_{n+j_2^*}\} > q_{1-\alpha^2}(\tilde{\bar{S}}_1,\ldots,\tilde{\bar{S}}_M)}{X_{1:(n+m)}}\\
&\le \PPst{\min\{S_{n+j_1^*}, S_{n+j_2^*}\} > q_{1-\alpha^2}(\tilde{S}_1,\ldots,\tilde{S}_M)}{X_{1:(n+m)}}
\le \alpha^2,
\end{align*}
where the first inequality holds since $\bar{S}_i \ge S_i$ deterministically for each $i \in [n+m]$ and $q_{1-\alpha^2}$ is monotone increasing with respect to each component of its inputs.

\subsection{Proof of Corollary~\ref{cor:sq_partition}}
Applying Theorem~\ref{thm:MAR_squared} for each $\ell \in [L]$, we have
\[\Delta_n:=\EEst{\Bigg(\frac{1}{N_{\ell}^0}\sum_{j \in U_{\ell}} \One{Y_{n+j} \notin \ch_U^2(X_{n+j})}\Bigg)^2}{X_{1:(n+m)}} \le \alpha_{\ell}^2,\]
for all $\ell \in [L]$. Next, we have
\begin{multline*}
    \sum_{l = 1}^L\Bigg(\frac{1}{N_{\ell}^0}\sum_{j \in U_{\ell}} \One{Y_{n+j} \notin \ch_U^2(X_{n+j})}\Bigg)^2 \cdot \sum_{l = 1}^L \bigg(\frac{N_{\ell}^0}{m}\bigg)^2\\
    \ge \left(\sum_{l = 1}^L\frac{N_{\ell}^0}{m} \cdot \frac{1}{N_{\ell}^0}\sum_{j \in U_{\ell}} \One{Y_{n+j} \notin \ch_U^2(X_{n+j})}\right)^2 = \bigg(\frac{1}{m}\sum_{i \in I_{A=0}} \One{Y_{n+j} \notin \ch_U^2(X_{n+j})}\bigg)^2,
\end{multline*}
by the Cauchy-Schwartz inequality. Therefore,
\begin{align*}
    &\EEst{\bigg(\frac{1}{m}\sum_{i \in I_{A=0}} \One{Y_{n+j} \notin \ch_U^2(X_{n+j})}\bigg)^2}{X_{1:(n+m)}} \\
    &\le \sum_{l = 1}^L \bigg(\frac{N_{\ell}^0}{m}\bigg)^2 \cdot \EEst{\sum_{l = 1}^L\Bigg(\frac{1}{N_{\ell}^0}\sum_{j \in U_{\ell}} \One{Y_{n+j} \notin \ch_U^2(X_{n+j})}\Bigg)^2}{X_{1:(n+m)}}\\
    &\le \sum_{l = 1}^L \bigg(\frac{N_{\ell}^0}{m}\bigg)^2 \cdot \sum_{l = 1}^L \alpha_{\ell}^2 = \alpha^2 \cdot \sum_{l = 1}^L \frac{{N_{\ell}^0}^2}{m^2}\cdot \sum_{l = 1}^L \frac{{N_{\ell}^0}^2 m^2}{(\sum_{l = 1}^L {N_{\ell}^0}^2)^2} = \alpha^2.
\end{align*}

\subsection{Proof of Theorem~\ref{thm:squared_p_A_X}}
The proof is similar to that  of Theorem~\ref{thm:MAR_known_p_A_X}. Let $j_1^*$ and $j_2^*$ be independent draws from $\textnormal{Unif}([m])$. Applying a similar argument to the proof of Theorem~\ref{thm:MAR_squared}, we have
\begin{align*}
    &\PPst{Y_{n+j_1^*} \notin \chps(X_{n+j_1^*}), Y_{n+j_2^*} \notin \chps(X_{n+j_2^*})}{X_{1:(n+m)}}\\
    &= \EEst{\Bigg(\frac{1}{m}\sum_{j=1}^m \One{Y_{n+j} \notin \chps(X_{n+j})}\Bigg)^2}{X_{1:(n+m)}}.
\end{align*}
Now consider the two distributions $P$ and $Q$, defined in~\eqref{eqn:P}. 
Under $Q$, the discretized data $(B_i,Z_i)_{i \in [n+m]}$ satisfies the assumptions of Theorem~\ref{thm:MAR_squared}, and thus
\[\Ppst{Q}{\min\{S_{n+j_1^*}, S_{n+j_2^*}\} > q_{1-\alpha^2}(\tilde{S}_1,\ldots,\tilde{S}_M)}{X_{1:(n+m)}}\\
\le \alpha^2,\]
by the proof of Theorem~\ref{thm:MAR_squared}, where $q_{1-\alpha^2}$ is defined in~\eqref{eqn:squared_quantile_fn}. 
Also, the target inequality is equivalent to
\[\Ppst{P}{\min\{S_{n+j_1^*}, S_{n+j_2^*}\} > q_{1-\alpha^2}(\tilde{S}_1,\ldots,\tilde{S}_M)}{X_{1:(n+m)}}\\
\le \alpha^2 + 2\eps,\]
and therefore it is sufficient to prove
$\dtv(P^{**}, Q^{**}) \le 2\eps$,
where $P^{**}$ and $Q^{**}$ denote the joint distribution of $(S_{n+j_1^*}, S_{n+j_2^*})$ under $P$ and $Q$, respectively. Let $D$ be any measurable subset of $\R^2$. Then
\begin{multline*}
    \Pp{P}{(S_{n+j_1^*}, S_{n+j_2^*}) \in D} - \Pp{Q}{(S_{n+j_1^*}, S_{n+j_2^*}) \in D}\\
    = \frac{1}{m^2} \sum_{j_1,j_2 \in [m]} \Big[\Pp{P}{(S_{n+j_1},S_{n+j_2}) \in D} - \Pp{Q}{(S_{n+j_1},S_{n+j_2}) \in D}\Big].
\end{multline*}
Let $P_{S_i}^P$ and $P_{S_i}^Q$ denote the distributions of $S_i$ under $P$ and $Q$, respectively. Now we observe that if $i=j$,
\begin{multline*}
    \Pp{P}{(S_i,S_j) \in D} - \Pp{Q}{(S_i,S_j) \in D} = \Pp{P}{(S_i,S_i) \in D} - \Pp{Q}{(S_i,S_i) \in D}\\
\le \dtv(P_{S_i}^P, P_{S_i}^Q) = \dtv(P_{S \mid X \in D_i, A=0}, P_{S \mid X \in D_i, A=1}) \le \eps,
\end{multline*}
and if $i \neq j$, we have
\begin{multline*}
    \Pp{P}{(S_i,S_j) \in D} - \Pp{Q}{(S_i,S_j) \in D} \le \dtv(P_{S_i}^P, P_{S_i}^Q) + \dtv(P_{S_j}^P, P_{S_j}^Q)\\
    = \dtv(P_{S \mid X \in D_i, A=0}, P_{S \mid X \in D_i, A=1}) + \dtv(P_{S \mid X \in D_j, A=0}, P_{S \mid X \in D_j, A=1}) \le 2\eps,
\end{multline*}
by Lemma~\ref{lem:dtv_0_1}. It follows that
$\Pp{P}{(S_{n+j_1^*}, S_{n+j_2^*}) \in D} - \Pp{Q}{(S_{n+j_1^*}, S_{n+j_2^*}) \in D} \le 2\eps$,
and we 
obtain the lower bound $-2\eps$ by an analogous argument. Since this holds for an arbitrary $D$, we have $\dtv(P^{**}, Q^{**}) \le 2\eps$, which implies the desired inequality.

\subsection{Proof of Corollary~\ref{cor:sq_partition_general}}

The target inequality follows directly from Theorem~\ref{thm:squared_p_A_X}, by following the  steps of the proof of Corollary~\ref{cor:sq_partition}.

\subsection{Proof of Theorem~\ref{thm:squared_p_A_X_estimated}}
We follow the steps of the proof of Theorem~\ref{thm:squared_p_A_X}. Then it turns out that it is sufficient to show
$\dtv(P^{**}, Q^{**}) \le 2( \eps + \delta_{\hat{p}_{A\mid X}} + \eps \cdot \delta_{\hat{p}_{A\mid X}})$.
This bound follows directly from the definition of $P^{**}, Q^{**}$ and the result in the proof of Theorem~\ref{thm:MAR_estimated_p_A_X}, where we prove
$\dtv(P^*, Q^*) \le  \eps + \delta_{\hat{p}_{A\mid X}} + \eps \cdot \delta_{\hat{p}_{A\mid X}}$,
and therefore the claim is proved.

\end{document}